\documentclass[onecollarge]{svjour2}                    
\smartqed  
\usepackage{graphicx}
\usepackage{amsmath,amsfonts,amssymb,amsopn,amscd}
\usepackage{bbm,wasysym,stmaryrd}
\usepackage{subfigure}
\usepackage{color,colortbl}
\usepackage{hyperref}

\renewcommand{\P}{\mathbbm{P}}
\newcommand{\Q}{\mathbbm{Q}}
\newcommand{\J}{\mathbf{J}}
\newcommand{\btau}{\boldsymbol{\tau}}
\newcommand{\W}{\mathbf{W}}
\newcommand{\Gv}{\mathbf{G}}

\newcommand{\m}{\mathbf{m}}

\newcommand{\Exp}{\mathbbm{E}}
\newcommand{\E}{\mathcal{E}}
\DeclareMathOperator{\eqlaw}{\stackrel{\mathcal{L}}{=}}

\newcommand{\erf}{\text{erf}}

\newcommand{\G}{\mathcal{G}}

\newcommand{\R}{\mathbbm{R}}
\newcommand{\N}{\mathbbm{N}}

\newcommand{\pa}{P_{\alpha}}

\newcommand{\derpart}[2]{ \frac{\partial #1}{\partial #2} } 
\newcommand{\der}[2]{ \frac{\text{d} #1}{\text{d} #2} }  
\newcommand{\derpartsxy}[3]{ \frac{\partial^2 #1}{\partial #2 \partial #3}} 

\newcommand{\M}{\mathcal{M}}
\newcommand{\C}{\mathcal{C}}
\newcommand{\Ct}{\mathcal{C}_{\tau}}

\newcommand{\ta}{\theta_{\alpha}}
\newcommand{\taag}{\tau_{\alpha \gamma}}
\newcommand{\ka}{k_{\alpha}}
\newcommand{\la}{\lambda_{\alpha}}
\newcommand{\Ja}{\bar{J}_{\alpha}}
\newcommand{\Na}{N_{\alpha}}
\newcommand{\Ng}{N_{\gamma}}
\newcommand{\sag}{\sigma_{\alpha \gamma}}
\newcommand{\Jag}{\bar{J}_{\alpha \gamma}}
\newcommand{\Sag}{S_{\alpha \gamma}}
\newcommand{\tk}{t^{(k)}}
\newcommand{\sk}{s^{(k)}}
\newcommand{\uk}{u^{(k)}}

\journalname{}

\graphicspath{{./Figures/}}

\begin{document}

\title{Large deviations, dynamics and phase transitions in large stochastic heterogeneous neural networks
}


\author{Tanguy Cabana \and Jonathan Touboul}


\institute{ The Mathematical Neuroscience Laboratory \\
			Coll\`ege de France / CIRB and INRIA Bang Laboratory\\
			11, place Marcelin Berthelot,
			75005 Paris, France
            Tel.: +33-144271388\\
            \email{jonathan.touboul@college-de-france.fr}           
}

\date{Received: date / Accepted: date}

\maketitle

\begin{abstract}
We analyze the macroscopic behavior of multi-populations randomly connected neural networks with interaction delays. Similar to cases occurring in spin glasses, we show that the sequences of empirical measures satisfy a large deviation principle, and converge towards a self-consistent non-Markovian process. The proof differs in that we are working in infinite-dimensional spaces (interaction delays), non-centered interactions and multiple cell types. The limit equation is qualitatively analyzed, and we identify a number of phase transitions in such systems upon changes in delays,  connectivity patterns and dispersion, particularly focusing on the emergence of non-equilibrium states involving synchronized oscillations.
\end{abstract}

\bigskip

\hrule
\setcounter{tocdepth}{2}

\tableofcontents

\medskip

\hrule

\section{Introduction}
Brain's activity results from the interplay of an extraordinary large number of neurons gathering into large-scale populations. This view is supported by anatomical evidences: neurons with the same dynamics and interconnection properties gather into cortical columns (\emph{neural populations}) with a diameter of about $50 \mu m$ to $1 mm$, containing of the order of few thousands to one hundred thousand neurons in charge of specific functions~\cite{mountcastle:97}. This macroscopic signal emerging from the interaction of a large number of cells is recorded by most non-invasive imaging techniques (EEG/MEG/Optical Imaging and Local Field Potentials). These columns have specific functions and spatial locations resulting in the presence of delays in their interactions due to the transport of information through axons and to the typical time the synaptic machinery needs to transmit it. These delays have a clear role in shaping the neuronal activity, as established by different authors (see e.g.~\cite{coombes-laing:11,series-fregnac:02}). Each neuron involved in this processing has a nonlinear dynamics and is subject to an intense noise.

This picture motivated the development of models of large-scale activity for stochastic neuronal networks~\cite{amari:72,wilson-cowan:73,brunel:00,bressloff:09} for different models. Some of the approaches are heuristic, and other based on statistical physics methods. Recently, extensions of the theory of mean-field limits to neuronal systems including spatial extension and delays have been developed in a mathematically rigorous framework~\cite{touboulNeuralfields:11,touboul-hermann-faugeras:11,touboulNeuralFieldsDynamics:11}. These papers exhibited the limit as the number of neurons goes to infinity, of multi-population (up to a continuum) neuronal networks in the presence of external noise and delays, and showed an essential role of noise in the qualitative dynamics of the macroscopic activity.

However, all these work overlook a prominent aspect of brain networks. Indeed, deeper analysis of the brain's connectivity evidence a high degree of heterogeneity, in particular in the interconnections~\cite{parker:03,marder-goaillard:06}. These are due to a number of phenomena, among which intervene the precise number of receptors and the extremely slow plasticity mechanisms. All these phenomena induce a static disorder termed \emph{static random synaptic heterogeneities}. Moreover, thermal noise, channel noise and the intrinsically noisy mechanisms of release and binding of neurotransmitter~\cite{aldo-faisal:08} result in stochastic variations of the synaptic weights termed \emph{stochastic synaptic noise}.

Experimental studies of cortical areas~\cite{aradi-soltesz:02} showed that the degree of heterogeneity in the connections significantly impacts the input-output function, rhythmicity and synchrony. Yet, qualitative effects induced by connection heterogeneities are still poorly understood theoretically. One notable exception is the work of Sompolinsky and collaborators~\cite{sompolinsky-crisanti-etal:88}. In the thermodynamic limit of a one-population firing-rate neuronal network with synaptic weights modeled as centered independent Gaussian random variables, they evidenced a phase transition between a stationary and a chaotic regime as the disorder is increased. More recently, we showed~\cite{touboul-hermann:12} that an additional phase transition towards synchronized activity could occur upon increase of the heterogeneity in multi-populations networks.

The topic of the present manuscript is to analyze in a mathematically rigorous manner the effect of heterogeneities and noise in large-scale networks in the presence of delays. One objective of the present manuscript is to understand from a mathematical perspective the mean-field equations used by Sompolinsky and colleagues in~\cite{sompolinsky-crisanti-etal:88} and to generalize their approach in order to be closer from the biological problem that motivates this study. The paper is organized as follows. In section~\ref{sec:MathSetting}, we will present the mathematical framework and the equations that will be studied throughout the paper, and in particular the inclusion of  interaction delays between cells, multiple populations and non-centered synaptic coefficients which were not present of prior models analyzed. From a mathematical perspective, this problem falls in the class of interacting diffusions in a random environment and is reminiscent of the work developed by Ben Arous and Guionnet on the Sherrington-Kirpatrick model~\cite{ben-arous-guionnet:95,guionnet:97} for spin-glass dynamics. Following their approach, we will extend their framework in order to address our biological problem. A heuristic physical argument and a summary of the results is provided in that section, and proofs will be presented in section~\ref{sec:BenArous} and appendix~\ref{append:Proofs}. Similarly to their analysis of spin glass dynamics, we will use large-deviation techniques to characterize the macroscopic limits of these systems in the limit where the number of neurons tends to infinity. There are two main technical distinctions with existing literature. First is the fact that we deal with multi-populations networks, which involves highly nontrivial difficulties, some of which will be further analyzed in a forthcoming study. Second, the fact that the connectivity weights are not centered introduce additional terms that need special care. And third is the presence of delays in the interactions, leading to work in infinite-dimensional spaces.

The limit of the systems being found, we will analyze the solutions of the asymptotic system in section~\ref{sec:Solutions}. Fortunately, we will show that the solutions of the limit system has Gaussian solutions that are exponentially attractive. These Gaussian solutions are univocally characterized by their mean and covariance functions, and these variables satisfy a closed set of deterministic equations, similar to the ones generally used by physicists from the dynamic mean-field theory~\cite{sompolinsky-crisanti-etal:88}. The mean satisfies a delayed differential equation coupled to the variance of the solution. These are extremely hard to study analytically. In particular, in contrast with the case of deterministic synaptic weights~\cite{touboul-hermann-faugeras:11,touboulNeuralFieldsDynamics:11}, the variance does not satisfies an ordinary differential equation but can be written as the solution of a fixed point equation. This distinction is fundamentally related to the non-Markovian nature of the asymptotic equation. However, based on the bifurcation analysis of the equation on the mean of the Gaussian solution, an alternative method to more usual analyzes of these mean-field equations, we will be able to characterize the solutions and demonstrate that noise is directly related to the emergence of synchronized oscillations, a highly relevant macroscopic state related to fundamental cortical functions such as memory, attention, sleep and consciousness, and its impairments relate to serious pathologies such as epilepsy or Parkinson's disease~\cite{uhlhaas-singer:06,buszaki:06}, and that may account for the results of Aradi and colleagues showing that increased heterogeneity was related with epilepsy.

\section{Mathematical setting}\label{sec:MathSetting}
In this section we introduce the model that we will be investigating in the present manuscript and the main results of the analysis. A heuristic discussion is provided to explain from a physical viewpoint the results that we will establish in section~\ref{sec:BenArous}.
\subsection{The Random Neural Network Model}
In all the manuscript, we are working in a complete probability space $(\Omega, \mathcal{F},\mathbbm{P})$ endowed with a filtration $\big(\mathcal{F}_t\big)_t$ satisfying the usual conditions. We consider a network composed of $N$ neurons falling into one of $M$ populations. We define by $p:\N\mapsto \{1,\cdots,M\}$ the \emph{population function} associating to a neuron index the population label it belongs to. The state of each neuron $i$ in population $p(i)=\alpha\in\{1,\cdots,M\}$ is described by its membrane potential $x^i\in\R$, and considered to satisfy a firing-rate equation (similar to the Hopfield network well known to physicists, see~\cite{amari:72,hopfield:82,wilson-cowan:72}):
\begin{equation}\label{eq:Network}
	dx^i_t=\left(-\frac{1}{\theta_{\alpha}}x^i_t + \sum_{\gamma=1}^M \sum_{j: p(j)=\gamma} J_{ij} S_{\alpha\gamma}(x^j_{t-\tau_{\alpha\gamma}}) \right)\,dt+\lambda_{\alpha} dW^{i}_t
\end{equation}
where $\theta_{\alpha}$ is the characteristic time of neurons of population $\alpha$, $J_{ij}$ is the synaptic efficiency of the synapse between neuron $j$ and neuron $i$, $S_{\alpha\gamma}$ the sigmoidal voltage-to-rate function, $\tau_{\alpha\gamma}$ the propagation delay between neurons of population $\gamma$ and neurons of population $\alpha$, and $\lambda_{\alpha}$ the noise intensity. The Brownian motions $W^i_t$ account for the noisy input received by all neurons, and the heterogeneity of the connections is integrated in the synaptic weights $J_{ij}$. These weights are assumed to be independent Gaussian random variables, with law $\mathcal{N}(\frac{\bar{J}_{p(i)p(j)}}{N_{p(j)}},\frac{\sigma_{p(i)p(j)}}{\sqrt{N_{p(j)}}})$, and in the stochastic synaptic noise case, these will be assumed to be stochastic processes, further specified below. Let us emphasize that all the results readily generalize to neurons having a nonlinear intrinsic dynamics:
\[	dx^i_t=\left(f_{\alpha}(t,x^i_t) + \sum_{\gamma=1}^M \sum_{j: p(j)=\gamma} J_{ij} S_{\alpha\gamma}(x^j_{t-\tau_{\alpha\gamma}}) \right)\,dt+\lambda_{\alpha} dW^{i}_t\]
under the condition that the functions $f_{\alpha}$ for $\alpha\in \{1,\cdots,M\}$ satisfy two standard regularity conditions:
\begin{enumerate}
	\item The functions $f_{\alpha}$ are uniformly locally Lipschitz-continuous in their second variable.
	\item They satisfy a uniform monotone growth condition: for any $t\in\R$, $xf_{\alpha}(t,x) \leq C(1+x^2)$.
\end{enumerate}
These refinements allow considering more biologically relevant neuron models. However, generalizations of the interaction that would depend on both the pre- and post-synaptic neurons (consisting in replacing the function $S_{\alpha\gamma}(x^j_{t-\taag})$ by a nonlinear function $b(x^i_t,x^j_{t-\taag})$ in equation~\eqref{eq:Network}) are harder to achieve and the methodology used in the present manuscript cannot be extended to these more general cases.

Let us define $\tau=\max_{\alpha\gamma} \tau_{\alpha\gamma}$, and let $(\mu_{\alpha})_{\alpha\in\{1\cdots M\}}$ a collection of $M$ probability distributions on $\Ct=\C([-\tau,0], \R)$ that will be used to characterize the initial condition of neurons in population $\alpha$. The initial condition on the network are considered chaotic, given by:
\begin{equation}\label{eq:ICNet}
	\text{Law of } (x_t)_{t\in[-\tau,0]} = \bigotimes_{i=1}^N \mu_{p(i)}.
\end{equation}

The function $(S_{\alpha\gamma})_{\alpha,\gamma\in\{1,\cdots,M\}}:\R\mapsto \R$ are Lipschitz-continuous sigmoidal functions, i.e. non-decreasing functions tending to $0$ at $-\infty$ and to $1$ at $\infty$. This model differs from the work of Ben-Arous and Guionnet in that the voltage is not bounded, there exist multiple populations and the interactions are nonlinear and delayed, which sets the problem in an infinite-dimensional space. The first question that may arise at this point is the well-posedness of the system.

\begin{proposition}\label{pro:ExistenceUniquenessNetwork}
	For each $J\in\R^{N\times N}$ and $T>0$, there exists a unique weak solution to the system~\eqref{eq:Network} defined on $[-\tau, T]$ with initial condition~\eqref{eq:ICNet}. Moreover, this solution is square integrable.
\end{proposition}
The proof of this result direct stems of standard theory of delayed stochastic differential equations~\cite{da-prato:92,mao:08}. Indeed, the network equations correspond to delayed stochastic differential equation in dimension $N$. For these equations, it is established that existence, uniqueness and square integrability of the solution hold under the assumption that the drift and diffusion functions are Lipschitz-continuous and enjoy a linear growth property, which clearly hold for our system.

\begin{remark}
	Note that if the initial condition was given by $(x^i_t)_{t\in[-\tau,0]}=\zeta^i$ with $\zeta^i\eqlaw \mu_{p(i)}$, we can also prove strong existence and uniqueness of solutions.
\end{remark}

\subsection{Heuristic approach to mean-field equations}\label{sec:heuristic}

The mean-field approach consists in reducing the $N$-dimensional delayed stochastic differential equation~\eqref{eq:Network} to a $M$-dimensional equation on the law of the system in the limit where all $N_{\alpha}\to \infty$. In~\eqref{eq:Network}, one needs to characterize the behavior of $\sum_{\gamma=1}^M \sum_{j: p(j)=\gamma} J_{ij} S_{\alpha\gamma}(x^j_{t-\taag})$ in the large $N_{\alpha}$ limit. Following the idea of molecular chaos of Boltzmann ("Sto\ss zahlansatz"), we may assume that all neurons behave independently and independently of the disorder. This physical ansatz is the basis of our heuristical approach. Assuming that the processes $x^j_{t-\taag}$ and the coefficients $J_{ij}$ are independent, the local interaction term $\sum_{\gamma=1}^M \sum_{j: p(j)=\gamma} J_{ij} S_{\alpha\gamma}(x^j_{t-\taag})$ is a sum of independent gaussian processes. A functional version of the central limit theorem (under suitable regularity and boundedness conditions) would ensure that
\[
\sum_{\gamma=1}^M \sum_{j: p(j)=\gamma} J_{ij} S_{\alpha\gamma}(x^j_{t-\taag})\to U^{\alpha, \bar{X}}_t \sim \mathcal{N}\left(\sum_{\gamma=1}^M \Jag \Exp\big[\Sag(\bar{X}^{\gamma}_{t-\taag})\big], \sum_{\gamma=1}^M \sag^2\Exp\big[\Sag(\bar{X}^{\gamma}_{t-\taag})^2\big]  \right)
\]
where the Gaussian limits are pairwise independent. In other words, under the molecular chaos hypothesis, averaging effects occur at the level of all neurons, and an effective interaction term in the form of a Gaussian process, will be obtained as an effective collective interaction term:
\begin{align}\label{eq:MeanfieldEquationInhomogeneity}
\begin{cases}
	d\bar{X}^{\alpha}_t=\left(-\frac{1}{\theta_{\alpha}}\bar{X}^{\alpha}_t + U^{\alpha, \bar{X}}_t \right)\,dt + \la dW^{\alpha}_t\\
U^{\alpha, \bar{X}}_t \sim \mathcal{N}\left(\sum_{\gamma=1}^M \Jag \Exp\big[\Sag(\bar{X}^{\gamma}_{t-\taag})\big], \sum_{\gamma=1}^M \sag^2\Exp\big[\Sag(\bar{X}^{\gamma}_{t-\taag})^2\big]  \right) pairwise \; independent.
\end{cases}
\end{align}

This is precisely the same kind of results one obtains in the mean-field theory for spin-glasses. From a mathematical viewpoint, the assumptions of independence of the variables $x^i$ and of the $J_{ij}$ cannot be rigorously legitimated.

And unfortunately, as is the case in the theory of spin glasses, rigorous approaches up to now rely on relatively abstract methods using in particular large-deviations theory. This is the approach we will now develop. We shall eventually note that a posteriori validation of the intuition provided by this heuristic argument will be obtained as a side result of our analysis: we will indeed show that a propagation of chaos phenomenon occurs in the limit where all $N_{\alpha} \to \infty$, i.e. that, provided the initial conditions are independent identically distributed for all neurons in the same population, then finite subsets of neurons behave independently for all times and have the same law given by the mean-field equation~\eqref{eq:MeanfieldEquationInhomogeneity}.

\subsection{Summary of the results}

Mathematical demonstration of such convergence results for multi-populations networks are still to be developed in general settings, and up to our knowledge no result of this kind for multi-populations networks with different population size exist in the literature. In particular, Sanov's theorem, hinge of most studies for the convergence of large-scale networks with heterogeneities, is still to be proved for such systems where the elements are one of a few populations and not identically distributed. In order to show rigorous results on the dynamics, we will slightly specify the problem and consider that all populations have identical sizes $n$ (hence, $N=Mn$). In that case, the analysis can be reduced to the analysis of one-populations systems but with $M$-dimensional variables, whose $\alpha\in\{1\cdots M\}$ component is the voltage of one of the neurons in population $\alpha$.

We now group the $N$ variables $(x^i)_{i=1\cdots N}$ into $M$-dimensional variables $X^i=(x^{i_{\alpha}},\alpha=1\cdots M)_{i=1\cdots n}$ where $p(i_{\alpha})=\alpha$, and such that for any $i\neq j$, $i_{\alpha}\neq j_{\alpha}$. The dynamics of these variables, readily deduced from the original network dynamics, can be written in vector form as:
\begin{equation}\label{eq:NetworkVector}
	dX^i_t=\left(L\cdot X^i_t + \sum_{j=1}^n \J_{ij}\odot S(X^j_{t-\btau}) \right)\,dt + \Sigma\cdot d\W^i_t
\end{equation}
where $L$ is the diagonal matrix with coefficient $(\alpha,\alpha)$ equal to $-1/\theta_{\alpha}$, $\J_{ij}$ is the $M\times M$ matrix with elements $(J_{i_{\alpha}j_{\beta}})_{\alpha,\beta=1\cdots M}$, and $S(X^j_{t-\btau})$ is the $M\times M$ matrix with elements $(\alpha,\gamma)$ given by $S_{\gamma\alpha}(X^{j_{\alpha}}_{t-\tau_{\gamma\alpha}})$. The linear operator $\odot$ acts on $M\times M$ matrices by giving the diagonal (as a $M$-dimensional vector) of the classical matrix product of two matrices, so that the component $\alpha$ of $\J_{ij}\odot S(X^j_{t-\btau})$ is precisely equal to $\sum_{\gamma=1}^M J_{i_{\alpha}j_{\gamma}} S_{\alpha\gamma}(x^{j_{\gamma}}_{t-\tau_{\alpha\gamma}})$. $\Sigma$ is the diagonal matrix with diagonal element $\lambda_{\alpha}$ and $\W^i_t=(W^{i_{\alpha}}_t)_{\alpha=1\cdots M}$.

We now work with an arbitrary fixed time $T>0$ and denote by $\Q^n(J)$ the unique law solution of the network equations~\eqref{eq:NetworkVector} restricted to the $\sigma$-algebra $\sigma(X^i_s, 1\leq i\leq n, -\tau\leq s\leq T)$. $\Q^n(J)$ is a probability measure on $\C^n$  where $\C$ is the space of continuous functions of $[-\tau,T]$ with value in $\R^M$. When neurons are not coupled (i.e. $J_{ij}=0$ for all $(i,j)$), the law of all variables $X^i_t$ are identical, independent, and given by the unique solution $P$ of the one-dimensional standard SDE:
\begin{equation}\label{eq:Uncoupled}
	\begin{cases}
			dX_t=L X_t dt + \Sigma d\W_t\\
			(X_0)\eqlaw \mu^0
	\end{cases}
\end{equation}
The uncoupled system hence has the law of the Ornstein-Uhlenbeck process. $P$ is the law of this process restricted to the $\sigma$-algebra $\G_T=\sigma(x_s,s\leq T)$, it is a probability measure on the space $\C_0$ of continuous functions of $[0,T]$ in $\R^M$. We will denote by $\pa$ the law of its component $\alpha$ (in our case, $P=\otimes_{\alpha=1}^M \pa$).

We want to study the behavior of the empirical law on each population:
\[\hat{\mu}_n = \frac{1}{n}\sum_{i=1}^n \delta_{X^i}\]
under $\Q^n(J)$ where the coefficients $J_{ij}$ are independent Gaussian with law given above. By a direct application of Girsanov theorem, $\Q^n(J)$ is absolutely continuous with respect to $P^{\otimes n}$ and we have:
\begin{multline}\label{eq:Density}
	\der{\Q^n(J)}{P^{\otimes n}} = \exp\Bigg(\sum_{i=1}^n  \int_{0}^{T}  \Big(\Sigma^{-1}\cdot \sum_{j=1}^n   \J_{ij} \odot S(X_{t-\btau}^j)\Big)'   d\W_{t}^{i}\\
	- \frac{1}{2} \int_{0}^{T} \Big\|\Sigma^{-1} \sum_{j=1}^n \J_{ij} \odot S(X^j_{t-\btau})\Big\|^2 dt\Bigg)
\end{multline}
where the prime denotes the transposition.

%

The aim of the present manuscript is to prove the following:

\begin{theorem}\label{thm:Convergence}
	Under $Q^{n}=\mathbbm{E}_J(\Q^{n}(J))$ the law of the $N$-neurons system averaged over all possible configurations (realizations of the synaptic weights $(J_{ij})$), the sequence of empirical measures $\hat{\mu}_{n}$ converges towards  $\delta_{Q}$ as $n$ goes to infinity for some probability measure $Q$.
\end{theorem}
%

In detail, we will show that the sequence empirical measures $\hat{\mu}_{n}$ satisfies a weak large deviation upper bound property and is tight. Precisely, we will demonstrate that:
\begin{theorem}\label{thm:LDP}
	There exists a good rate function $H$ such that for any compact subset $K$ of $\M_1^+ (\C)$ where $\C=\C([-\tau,T],\R^M)$:
	\[\limsup_{n\to\infty} \frac 1 n \log Q^{n}(\hat{\mu}_{n}\in K)\leq -\inf_{K} H.\]
\end{theorem}
Moreover, we will show a tightness result on that sequence, namely:
\begin{theorem}\label{thm:Tightness}
	For any real number $\varepsilon>0$, there exists a compact subset $K_{\varepsilon}$ such that for any integer $n$,
	\[Q^{n}(\hat{\mu}_{n}\notin K_{\varepsilon})\leq \varepsilon.\]
\end{theorem}
These two results imply the convergence result provided that we characterize uniquely the minima of $H$, which will be the subject of theorem:
\begin{theorem}\label{thm:Limit}
	The good rate function $H$ is such that:
	\begin{enumerate}
		\item It achieves its minimal value at the probability measure $Q\ll P$ satisfying the implicit equation:
		\[\der{Q}{P}=\mathcal{E}\left[\sum_{\alpha=1}^M\frac 1 {\lambda_{\alpha}}\int_0^T G_t^{\alpha,Q^{\alpha}} dW^{\alpha}_t - \frac 1 {2\lambda_{\alpha}^2} \int_0^T (G_t^{\alpha,Q^{\alpha}})^2 dt\right]\]
		where $\mathcal{E}$ denotes the expectation over $G^{Q}=(G^{\alpha,Q^{\alpha}})_{\alpha=1\cdots M}$ a Gaussian process with mean:
		\[\mathcal{E}[G^{\alpha,Q^{\alpha}}_t] = \sum_{\gamma=1}^M \Jag \int \Sag(x_{t-\tau_{\alpha\beta}})dQ^{\beta}(x) \]
		and covariance:
		\[\mathcal{E}[G^{\alpha,Q^{\alpha}}_tG^{\gamma,Q^{\gamma}}_s] = \delta_{\alpha\gamma}\sum_{\beta=1}^M \sigma_{\alpha\beta}^2\int S_{\alpha\beta}(x_{t-\tau_{\alpha\beta}})S_{\alpha\beta}(x_{s-\tau_{\alpha\beta}})dQ^{\beta}(x)\]
		where $\delta_{\alpha\gamma}$ equals $1$ if $\alpha=\gamma$ and $0$ otherwise.
		\item This provides an implicit self-consistent equation on $Q$ the limit distribution, which has a unique solution.
	\end{enumerate}
\end{theorem}
This theorem will be demonstrated in section~\ref{sec:LimitIdentification}.

Eventually, we will show that the unique solution has Gaussian local equilibria (theorem~\ref{pro:GaussianSolutions}), with mean $(\mu^{\alpha}(t))_{\alpha=1\cdots M}$ and covariance $C^{\alpha\beta}(s,t)$ satisfying the well-posed system of deterministic equations:
\begin{equation}\label{eq:MeansSummary}
		\dot{\mu}^{\alpha}(t)=-\frac{1}{\theta_{\alpha}} \mu^{\alpha}(t) + \sum_{\gamma=1}^M \bar{J}_{\alpha\gamma}f_{\alpha\gamma}(\mu^{\gamma}(t-\tau_{\alpha\gamma}), C^{\alpha\alpha}(t-\tau_{\alpha\gamma}, t-\tau_{\alpha\gamma}))
	\end{equation}
	where $f_{\alpha\gamma}(\mu,v)=\int_{\R} S_{\alpha\gamma}(x) \frac{e^{-(x-\mu)^2/2v}}{\sqrt{2\pi v}} dx$. The covariance is equal to zero when $\beta\neq \alpha$ and:
	\begin{equation}
		C^{\alpha}(t,s)=e^{-(t+s)/\theta_{\alpha}}\Big[C^{\alpha}(0,0)+ \frac{\ta \la^2}{2} (\exp{2(t\wedge s)/\ta} -1) + \sum_{\gamma=1}^P \sigma_{\alpha\gamma}^2\int_{0}^t\int_{0}^s e^{(u+v)/\theta_{\alpha}} \Delta_{\mu,C}^{\alpha\gamma}(u-\tau_{\alpha\gamma},v-\tau_{\alpha\gamma}) dudv\Big]
	\end{equation}
	where $\Delta_{\mu,C}^{\alpha\gamma}(u,v)=\mathbb{E}\Big[S_{\alpha\gamma}(X_{u}^{\gamma})S_{\alpha\gamma}(X_{v}^{\gamma})\Big]$ is a nonlinear function of $\mu^{\gamma}(u)$, $\mu^{\gamma}(v)$, $C^{\gamma\gamma}(u,v)$, $C^{\gamma\gamma}(u,u)$ and $C^{\gamma\gamma}(v,v)$.
	
	In other words, the solution can be written, in law, as the solution of an implicit equation:
	\begin{equation}\label{eq:MFEwithUSummary}
		d\bar{X}^{\alpha}_t = \left(-\frac{1}{\theta_{\alpha}}\bar{X}^{\alpha}_t + U^{\alpha,\bar{X}}_t\right)\,dt+\lambda_{\alpha}dW^{\alpha}_t
	\end{equation}
	where the processes $(W^{\alpha}_t)$ are independent Brownian motions and the processes $U^{\alpha,\bar{X}}_t$ are Gaussian processes with law as $G^{Q}$ as described in theorem~\ref{thm:Limit}. These characterizations will be very handy to analyze the qualitative dynamics and phase transitions of the system.

	\section{Large Deviations and Mean-Field limits}\label{sec:BenArous}
	\subsection{Construction of the good rate function}
	The aim of this section is to identify a good rate function for the system that we will use in Theorem \ref{thm:LDP} to show a large deviation principle. In order to introduce our good rate function, it is convenient to analyze for a moment a time-discretization of the equation (section~\ref{sec:DiscreteTime}), which will expedite the analysis of our continuous time problem.

	\subsubsection{Analysis of the discrete time system}\label{sec:DiscreteTime}
	Given an integer $k$, we define $\Delta^k = \big\{0=t_0 < t_1 < ... < t_k < t_{k+1}=T \big\}$ a partition of $[0,T]$ and consider the following dynamics for the neurons of population $\alpha$:
	\begin{equation}\label{eq:DiscretizedNetworkAlpha}
		\begin{cases}
				dx^i_t=\left(-\frac{1}{\ta}x^i_t + \sum_{\gamma=1}^M \sum_{j: p(j)=\gamma} J_{ij} S_{\alpha p(j)}(x^j_{\tk-\tau_{\alpha p(j)}}) \right)\,dt+ \la dW^{i}_t\\
	            \tk= \sup\big\{ t_l \in \Delta^k | t_l \leq t \big\}\\
				\text{Law of } (x_t^{\alpha})_{t\in[-\tau,0]} = \mu_{\alpha}^{\otimes \Na}
		\end{cases}
	\end{equation}

	As in Proposition \ref{pro:ExistenceUniquenessNetwork}, this system clearly admits a unique weak solution for any $J \in \R^{N \times N}$. We will denote $\Q^{\Na,n}(J)$ its restriction to the $\sigma$-algebra $\sigma(x^i_s, 1\leq i \leq N ; p(i)=\alpha, -\tau\leq s\leq T)$, and $Q^{\Na,n}=\mathbbm{E}_J(\Q^{\Na,n}(J))$. They are both probability measures on $\C([-\tau,T],\R)^n$. By Girsanov Theorem, $\Q^{\Na,n}(J) \ll \pa^{\otimes n}$ with:
	\begin{multline*}
	\der{\Q^{\Na,n}(J)}{{\pa}^{\otimes n}}   =   \exp  \bigg\{\sum_{i:p(i)=\alpha}  \int_{0}^{T}  \Big(\frac{1}{\la}  \sum_{\gamma=1}^{P} \sum_{j:p(j)=\gamma}   J_{ij} \Sag(x_{\tk-\taag}^j)\Big)   dW_{t}^{i}  \\
	-  \int_{0}^{T} \Big(\frac{1}{\la} \sum_{\gamma=1}^{P} \sum_{j:p(j)=\gamma} J_{ij} \Sag(x_{\tk-\taag}^j)\Big)^2 dt \bigg\}.
	\end{multline*}
	For every $\mu \in {\M_1^+(\C)}$, the relative entropy with respect to $P$ is defined by:
	\begin{equation*}
	I(\mu|P) =
	\begin{cases}
		\displaystyle{\int \log{ \frac{d\mu}{dP} } }d\mu  &  \text{if }  \mu \ll P,\\
		\infty & \text{otherwise }.
	\end{cases}
	\end{equation*}
	We introduce, for $\mu \in \M_1^+(\C)$, the two following functions, respectively defined on $[0,T]^2$ and $[0,T]$:
	\begin{equation*}
	\begin{cases}
		K_{\mu}^{\alpha}(s,t)&=\displaystyle{\frac{1}{\la^2} \sum_{\gamma=1}^P \sag^2 \int_{\C} \Sag(x^{\gamma}_{t-\taag})\Sag(x^{\gamma}_{s-\taag}) d\mu(x)}\\
		m_{\mu}^{\alpha}(t) &=\displaystyle{\frac{1}{\la} \sum_{\gamma=1}^{P} \Jag \int_{\C} \Sag(x^{\gamma}_{t-\taag}) d\mu(x)}
	\end{cases}.
	\end{equation*}

	Remark that, since $\Sag$ takes value in $[0,1]$, both functions are bounded: $K^{\alpha}_{\mu}(s,t) \leq \frac{\ka}{\la^2}$ and $m^{\alpha}_{\mu}(t) \leq \frac{\Ja}{\la}$, with $\ka=\sum_{\gamma=1}^P \sag^2$ and $\Ja = \sum_{\gamma=1}^P |\Jag|$.

	We now define
	\begin{equation*}
	K_{\mu}(s,t) =
	\left(
	\begin{array}{cccc}
	K_{\mu}^1(s,t) & 0 & \ldots & 0 \\
	0 & K_{\mu}^2(s,t) & \ddots & \vdots \\
	\vdots & \ddots & \ddots & 0\\
	0 & \ldots & 0 & K_{\mu}^M(s,t)
	\end{array}
	\right)
	\end{equation*}
	It is well known that we can find a $M$-dimensional stochastic process $ \Gv = (G^{\alpha})_{\alpha=1\cdots M}$ on $(\Omega,\mathcal{F},\P)$ such that, for any $M\times M$ variance-covariance matrix $K$ on $[0,T]^2$, there exists a probability measure $\gamma_K$, under which $\Gv$ is a centered gaussian process with covariance $K$. We shall use the shorthand notation $\gamma_{\mu}$ for $\gamma_{K_{\mu}}$, and ${\E}_{\mu}$ the expectation under $\gamma_{\mu}$. We now define for $\mu \in \M_1^+(\C)$:
	\begin{multline*}
	\Gamma^{k}(\mu) = \int_{\C} \log\bigg( \int \exp\bigg\{ \sum_{l=0}^k \big(\Gv_{t_l}(\omega)+\m_{\mu}(t_l)\big)' \cdot \big(\W_{t_{l+1}} - \W_{t_l}\big)(x) \\- \frac{1}{2} \sum_{l=0}^k \Big\|\Gv_{t_l}(\omega)+\m_{\mu}(t_l)\Big\|^2(t_{l+1}-t_l)\bigg\}  d\gamma_{K_{\mu}}(\omega) \bigg) d\mu(x)
	\end{multline*}

	\begin{multline*}
	\Gamma^{\alpha,k}(\mu) = \int_{\C} \log\bigg( \int \exp\bigg\{ \sum_{l=0}^k \big(G_{t_l}^{\alpha}(\omega)+m_{\mu}^{\alpha}(t_l)\big)\big(W^{\alpha}_{t_{k+1}} - W^{\alpha}_{t_l}\big)(x) \\- \frac{1}{2} \sum_{l=0}^k \big(G_{t_l}^{\alpha}(\omega)+m_{\mu}^{\alpha}(t_l)\big)^2(\omega)(t_{l+1}-t_l)\bigg\}  d\gamma_{K_{\mu}}(\omega) \bigg) d\mu(x)
	\end{multline*}
	where $\W_t(x)= \left( W^{\alpha}_t(x)=\frac{x^{\alpha}_t-x^{\alpha}_0}{\la} + \int_0^t \frac{x^{\alpha}_s}{\ta \la} ds\right)_{\alpha=1\cdots M}$, for $x \in \C$. Eventually, we define the function:
	\begin{align*}
	H^k (\mu) =
	\left\{
	\begin{array}{rl}
	I(\mu|P) - \Gamma^k(\mu) & \text{if } I(\mu|P) < \infty,\\
	\infty & \text{otherwise }.
	\end{array}
	\right.
	\end{align*}

	One can easily see that $\Gamma^{k}= \sum_{\alpha=1}^M \Gamma^{\alpha,k}$, as the components of $\Gv$ are independent under $\gamma_{\mu}$.

	We further define:
	 \[
	\Gamma_1^k(\mu)= \log\Big( \int \exp\big( - \frac{1}{2} \int_0^T \big\|\Gv_{\tk}(\omega)\big\|^2dt \big) \\ d\gamma_{K_{\mu}}(\omega)\Big)  - \frac{1}{2} \int_0^T \big\|\m_{\mu}(\tk)\big\|^2 dt
	 \]
	\[
	\Gamma_2^k(\mu)= \frac{1}{2} \int \int \Big( \int_0^T \Gv_{\tk}' \cdot d\W_t(x) - \m_{\mu}(\tk)dt) \Big)^2 d\gamma_{\widetilde{K}_{\mu}^{T,k}} d\mu(x) + \\ \int \int \m_{\mu}(\tk)' \cdot d\W_t (x) d\mu(x)
	\]

	where
	\[
	\widetilde{K}_{\mu}^{t,k}(s,u)=\Big(\int \frac{ \exp\Big\{ - \frac{1}{2} \int_0^t \big(G^{\alpha}_{\uk}(\omega)\big)^2+\mathbf{1}_{\alpha \neq \gamma}\big(G^{\gamma}_{\uk}(\omega)\big)^2du\Big\} G^{\gamma}_{\uk}(\omega)G^{\alpha}_{\sk}(\omega) }{\int \exp\Big\{ - \frac{1}{2} \int_0^t (G^{\alpha}_{\uk}(\omega)\big)^2+\mathbf{1}_{\alpha \neq \gamma}\big(G^{\gamma}_{\uk}(\omega)\big)^2du\Big\} d\gamma_{\mu}} d\gamma_{\mu} \Big)_{\alpha, \gamma \in \{1 \cdots M\}}.
	\]
	One can easily see that this function takes values in the $M\times M$ diagonal positive matrices. Moreover, we can define $\gamma_{\widetilde{K}_{\mu}^{T,k}}$, probability measure on $\Omega$, such that
	\[d\gamma_{\widetilde{K}_{\mu}^{T,k}}=  \frac{\prod_{\alpha=1}^M \exp\Big\{ - \frac{1}{2} \int_0^T \big(G^{\alpha}_{\tk}(\omega)\big)^2 dt\Big\}}{\int \prod_{\alpha=1}^M \exp\Big\{ - \frac{1}{2} \int_0^T \big(G^{\alpha}_{\tk}(\omega)\big)^2 dt\Big\} d\gamma_{\mu}} d\gamma_{\mu},\]
	under which $\Gv$ is a $M$-dimensional centered gaussian process with covariance $\widetilde{K}_{\mu}^{T,k}$ (this Gaussian calculus property is proved for instance in ~\cite[Appendix A]{ben-arous-guionnet:95}).

	\begin{proposition} \label{pro:DécompositionGamma}
		We have:
	\[
	\Gamma^{k}(\mu)= \Gamma_1^{k}(\mu) + \Gamma_2^{k}(\mu)
	\]
	\end{proposition}

	\begin{proof}
	Let
	 \[
	\Gamma_1^{\alpha,k}(\mu)= \log\Big( \int \exp\big( - \frac{1}{2} \int_0^T {G^{\alpha}_{\tk}}^2(\omega)dt \big) \\ d\gamma_{K_{\mu}}(\omega)\Big)  - \frac{1}{2} \int_0^T (m_{\mu}^{\alpha}(\tk))^2 dt
	 \]
	\[
	\Gamma_2^{\alpha,k}(\mu)= \frac{1}{2} \int \int \Big( \int_0^T G_{\tk}^{\alpha} (dW_t^{\alpha}(x) - m^{\alpha}_{\mu}(\tk)dt) \Big)^2 d\gamma_{\widetilde{K}_{\mu}^{T,k}} d\mu(x) + \\ \int \int m_{\mu}^{\alpha}(\tk)dW^{\alpha}_t (x) d\mu(x)
	\]
	For $\Gamma_i^k= \sum_{\alpha=1}^M  \Gamma_i^{\alpha,k}, i\in \{1,2\}$, it is sufficient to prove that
	\[
	\Gamma^{\alpha,k}(\mu)= \Gamma_1^{\alpha,k}(\mu) + \Gamma_2^{\alpha,k}(\mu).
	\]
	But,
	\begin{align*}
	\Gamma^{\alpha, k}(\mu)  = & \int \log\bigg( \int \exp\bigg\{ \int_0^T \big(G^{\alpha}_{\tk}(\omega)+m_{\mu}^{\alpha}(\tk)\big)dW_t^{\alpha}(x) \\
	& - \frac{1}{2} \int_0^T \big(G_{\tk}^{\alpha}(\omega)+m_{\mu}^{\alpha}(\tk)\big)^2 dt \bigg\}  d\gamma_{K_{\mu}}(\omega) \bigg) d\mu(x)\\
	= & \int \log\bigg\{ \bigg( \exp\Big\{-\frac{1}{2} \int_0^T (m^{\alpha}_{\mu}(\tk))^2dt \Big\} \int \exp\Big\{-\frac{1}{2}\int_0^T (G^{\alpha}_{\tk})^2dt \Big\} d\gamma_{\mu}  \bigg) \\
	& \times \bigg( \exp\Big\{ \int_0^T m_{\mu}^{\alpha}(\tk) dW_t^{\alpha}\Big\} \int \exp\Big\{ \int_0^T G_{\tk}^{\alpha}\big(dW_t^{\alpha} - m_{\mu}^{\alpha}(\tk)dt\big) \Big\} d\gamma_{\widetilde{K}_{\mu}^{T,k}} \bigg) \bigg\} d\mu \\
	= & \log\bigg\{ {\E}_{\mu}\bigg[\exp{ \bigg( -\frac{1}{2} \int_0^T (G^{\alpha}_{\tk})^2 dt \bigg) } \bigg] \bigg\} - \frac{1}{2}\int_0^T (m^{\alpha}_{\mu}(\tk))^2 dt + \int \int_0^T m_{\mu}^{\alpha}(\tk) dW_t^{\alpha} d\mu \\
	& +\int \log\bigg\{ \int \exp{ \bigg(\int_0^T G_{\tk}^{\alpha}\big(dW_t^{\alpha} - m_{\mu}^{\alpha}(\tk)dt\big) \bigg) } d\gamma_{\widetilde{K}_{\mu}^{T,k}} \bigg\} d\mu
	\end{align*}
	Standard gaussian calculus yields:
	\begin{equation*}
	\int \exp{ \bigg(\int_0^T G_{\tk}^{\alpha}\big(dW_t^{\alpha} - m_{\mu}^{\alpha}(\tk)dt\big) \bigg) } d\gamma_{\widetilde{K}_{\mu}^{T,k}} = \exp\bigg\{ \frac{1}{2} \int \bigg( \int_0^T G_{\tk}^{\alpha} \big(dW_t^{\alpha}-m_{\mu}^{\alpha}(\tk)dt\big) \bigg)^2 d\gamma_{\widetilde{K}_{\mu}^{T,k}} \bigg\}
	\end{equation*}
	so that,
	\begin{align*}
	\Gamma^{\alpha, k}(\mu) & = \Gamma^{\alpha, k}_1(\mu) \\
	&+ \int \int_0^T m_{\mu}^{\alpha}(\tk) dW_t^{\alpha} d\mu + \int \log\bigg\{ \exp\bigg\{ \frac{1}{2} \int \bigg( \int_0^T G_{\tk}^{\alpha} \big(dW_t^{\alpha}-m_{\mu}^{\alpha}(\tk)dt\big) \bigg)^2 d\gamma_{\widetilde{K}_{\mu}^{T,k}} \bigg\} \bigg\} d\mu\\
	& = \Gamma_1^{\alpha, n}(\mu) + \int \int_0^T m_{\mu}^{\alpha}(\tk) dW_t^{\alpha} d\mu +  \frac{1}{2} \int \int \bigg( \int_0^T G_{\tk}^{\alpha} \big(dW_t^{\alpha}-m_{\mu}^{\alpha}(\tk)dt\big) \bigg)^2 d\gamma_{\widetilde{K}_{\mu}^{T,k}}  d\mu
	\end{align*}
	which concludes the proof.
	\end{proof}



	For all $\mu, \nu \in \M_1^+(\C)$, let:

	\begin{multline*}
	\Gamma^{k}_{\nu}(\mu) = \int_{\C} \log\bigg( \int \exp\bigg\{ \int_0^T \big(\Gv_{\tk}(\omega)+\m_{\nu}(\tk)\big)' \cdot d\W_t(x) \\
	- \frac{1}{2} \int_0^T \Big\|\Gv_{\tk}(\omega)+\m_{\nu}(\tk)\Big\|^2dt\bigg\}  d\gamma_{K_{\nu}}(\omega) \bigg) d\mu(x)
	\end{multline*}

	\begin{align*}
	& \Gamma_{2,\nu}^{\alpha, k}(\mu) = \frac{1}{2} \int \int \Big( \int G^{\alpha}_{\tk} (dW^{\alpha}_t - m^{\alpha}_{\nu}(\tk)dt) \Big)^2 d\gamma_{\widetilde{K}_{\nu}^{T,n}} d\mu + \int \int m^{\alpha}_{\nu}(\tk)dW^{\alpha}_t d\mu \\
    & \Gamma_{2,\nu}^{k} = \sum_{\alpha=1}^M \Gamma_{2,\nu}^{\alpha, k} \\
	& \Gamma_{\nu}^{\alpha, k}(\mu) =\Gamma_1^{\alpha, k}(\nu) + \Gamma_{2,\nu}^{\alpha, k}(\mu)
	\end{align*}

	One can easily see that $\Gamma^{k}_{\nu} = \sum_{\alpha=1}^M \Gamma_{\nu}^{\alpha, k}$.
	Let

	\begin{align*}
	H^{k}_{\nu} : \M_1^+(\C) & \rightarrow \R^+\\
	\mu & \mapsto
	\begin{cases}
		\displaystyle{I(\mu|P) - \Gamma_{\nu}^{k}(\mu)} & \text{if } I(\mu|P) < \infty,\\
		\infty & \text{otherwise }.	
	\end{cases}
	\end{align*}
	We eventually denote by $d_T$ the Vaserstein distance on $M_1^+(\C)$, compatible with the weak topology:
	\begin{equation*}
	d_T(\mu,\nu)=\inf_{\xi}\bigg\{ \int \sup_{-\tau \leq t \leq T; \gamma \in \{ 1,...,M\} } |x^{\gamma}_t-y^{\gamma}_t|^2 d\xi(x,y) \bigg\}^{\frac{1}{2}}
	\end{equation*}
	the infimum being taken on the laws $\xi$ with marginals $\mu$ and $\nu$.

	We now show the following regularity properties on the introduced functions:
	\begin{lemma}\label{lemma1}
	\begin{enumerate}
	\item There exists a positive constant $C_T$, depending on T but not on n,
	such that: $|\Gamma_1^{k}(\mu)-\Gamma_1^{k}(\nu)| \leq C_T d_T(\mu,\nu)$.
	\item $\Gamma^{k} \leq I(.|P)$ i.e. $H^k$ is a positive function. In particular, $\Gamma^{k}$ is finite whenever $I(.|P)$ is.
	\item There exists real constants $a < 1$ and $\eta >0$ such that $\Gamma^{k}\leq a I(.|P) + \eta$.
	\item There exists a positive constant $C_T$, depending on T but not on n,
	such that: $|\Gamma_{2,\nu}^{k}(\mu)-\Gamma_2^{k}(\mu)| \leq C_T \big(1+I(\mu|P)\big) d_T(\mu,\nu)$.
	\item Defining the following probability measure on $\M_1^+(\C)$:
	\begin{align*}
	dQ_{\nu}^k(x) & = \exp{\Gamma_{\nu}^{k}(\delta_x)}dP(x)  \\
	& = \int \exp{ \Bigg( \int_0^T \big( \Gv_{\tk} + \m_{\nu}(\tk) \big)' \cdot \, d\W_t(x) - \frac{1}{2} \int_0^T \Big\| \Gv_{\tk} + \m_{\nu}(\tk)\Big\|^2 dt \Bigg) } \, d\gamma_{\nu} \, dP(x)
	\end{align*}
	 we have $H_{\nu}^{k}=I(.|Q_{\nu}^k)$, so that $H_{\nu}^k$ is lower semi-continuous on $\M_1^+(\C)$.
	\item $H^k$ is a good rate function.
	\end{enumerate}
	\end{lemma}

	This technical lemma is proved in appendix~\ref{append:ProofLemma1}

	\subsubsection{Analysis in continuous time}

	In this section, we shall prove theorem~\ref{thm:Convergence} as a consequence of theorems~\ref{thm:LDP} and~\ref{thm:Tightness}. We start by extending the function constructed in the previous section to our continuous time setting. To this purpose, let us define
	\begin{align*}
	\Gamma : & \{\mu \in \M_1^+(\C) | I(\mu|P) < \infty\} \to \R \\
	& \mu \to \int \log{\bigg(\int \exp{\bigg(\int_0^T \big(\Gv_t + \m_{\mu}(t)\big)' \cdot d\W_t - \frac{1}{2}\int_0^T \Big\|\Gv_t + \m_{\mu}(t)\Big\|^2 dt \bigg)} d\gamma_{\mu} \bigg)} d\mu,
	\end{align*}
	\begin{align*}
	\Gamma^{\alpha} : & \{\mu \in \M_1^+(\C) | I(\mu|P) < \infty\}\to \R \\
	& \mu  \to \int \log{\bigg(\int \exp{\bigg(\int_0^T \big(G^{\alpha}_t + m_{\mu}^{\alpha}(t)\big) dW^{\alpha}_t - \frac{1}{2}\int_0^T \big(G^{\alpha}_t + m_{\mu}^{\alpha}(t)\big)^2 dt \bigg)} d\gamma_{\mu} \bigg)} d\mu,
	\end{align*}
	and
	\begin{align*}
	H(\mu) =
	\left\{
	\begin{array}{rl}
	I(\mu|P) - \Gamma(\mu) & \text{if } I(\mu|P) < \infty,\\
	\infty & \text{otherwise }.
	\end{array}
	\right.
	\end{align*}

	Remark that, as for the discrete time case, $\Gamma= \sum_{\alpha=1}^M \Gamma^{\alpha}$.

	\begin{proposition}
	\begin{enumerate}
	\item On the compact set $K_L = \{\mu \in \M_1^+(\C)| I(\mu|P) \leq L \}$, $\Gamma^k$ converges uniformly to $\Gamma$.
	\item $\forall \mu \in K_L, \; \Gamma^{\alpha}(\mu) = \Gamma^{\alpha}_1(\mu) + \Gamma^{\alpha}_2(\mu)$ where
	\begin{equation*}
	\Gamma^{\alpha}_1(\mu) = \log{\int \exp{\bigg( -\frac{1}{2} \int_0^T {G^{\alpha}_t}^2 dt \bigg) } d\gamma_{\mu}} -\frac{1}{2} \int_0^T {m^{\alpha}_{\mu}}^2(t) dt,
	\end{equation*}
	\begin{equation*}
	\Gamma^{\alpha}_2(\mu)= \frac{1}{2}\int\int \bigg(\int_0^T G^{\alpha}_t \big(dW^{\alpha}_t - m_{\mu}^{\alpha}(t) dt \big) \bigg)^2 d\gamma_{\widetilde{K}^T_{\mu}} d\mu + \int\int_0^T m_{\mu}^{\alpha}(t)dW^{\alpha}_t d\mu
	\end{equation*}
	\item $\Gamma \leq I(.|P)$ and $\exists a < 1, \eta > 0 | \; \Gamma \leq a I(.|P) + \eta$.
	\item $H$ is a good rate function.
	\end{enumerate}
	\end{proposition}

	\begin{proof}
	\begin{description}
	\item[(i), (ii)]  Following the same demonstration as in Lemma.\ref{lemma1} (i) and (iv), we find that there exists $C_{1,T}$ and $C_{2,T}$ such that (see \eqref{gamma1} and \eqref{gamma2})
	    \begin{align*}
	    \big|\big(\Gamma_1^{\alpha, k}-\Gamma_1^{\alpha, k+p}\big)(\mu)\big| & \leq C_{1,T}  \max_{\gamma=1,M}  \Big( \int \int_0^T \big|\Sag (x^{\gamma}_{\tk-\taag}) - \Sag (x^{\gamma}_{t^{(k+p)}-\taag})\big|^2 dt \, d\mu(x)\Big)^{\frac{1}{2}} \\
	    & \leq C_{1,T} \sqrt{T}  \max_{\gamma=1,M}  \Big( \int \sup_{|t-s| \leq |\Delta_k|} \big|\Sag (x^{\gamma}_t) - \Sag (x^{\gamma}_s)\big|^2  d\mu(x)\Big)^{\frac{1}{2}} \\
	    \big|\big(\Gamma_2^{\alpha,k}-\Gamma_2^{\alpha, k+p}\big)(\mu)\big| & \leq C_{2,T} \big(I(\mu|P) +1\big) \max_{\gamma=1,M}  \Big( \int \sup_{|t-s| \leq |\Delta_k|} \big|\Sag (x^{\gamma}_t) - \Sag (x^{\gamma}_s)\big|^2  d\mu(x)\Big)^{\frac{1}{2}}
	    \end{align*}
	    But, according to \eqref{eq:IneqRelativeEntropy}, we have for any $a\geq0$
	    \begin{equation*}
	    a \int \sup_{|t-s| \leq |\Delta_k|} \big|\Sag (x^{\gamma}_t) - \Sag (x^{\gamma}_s)\big|^2  d\mu(x) \leq I(\mu|P) + \log{\int \exp{\bigg\{ a \sup_{|t-s| \leq |\Delta_k|} \big|\Sag (x^{\gamma}_t) - \Sag (x^{\gamma}_s)\big|^2 \bigg\}} dP(x) }
	    \end{equation*}
	    And by bounded convergence theorem
	    \begin{equation*}
	    \lim_{k\to \infty} \log{\int \exp{\bigg\{ a \sup_{|t-s| \leq |\Delta_k|} \big|\Sag (x^{\gamma}_t) - \Sag (x^{\gamma}_s)\big|^2 \bigg\}} dP(x) } = 0
	    \end{equation*}
	    Let $\varepsilon>0$, choosing $a=\frac{1}{\varepsilon^2}$, it is easy to see that there exists an integer $k(\varepsilon)$ such that, for $k \geq k(\varepsilon), \forall \gamma \in \{1,...M\}$:
	    \begin{equation*}
	    \int \sup_{|t-s| \leq |\Delta_k|} \big|\Sag (x^{\gamma}_t) - \Sag (x^{\gamma}_s)\big|^2  d\mu(x) \leq  \big(I(\mu|P) + 1\big) \varepsilon^2
	    \end{equation*}
	    Hence, for any $k \geq k(\varepsilon)$, any $p$, and any $\mu \in K_L$:
	    \begin{align*}
	    \big|\big(\Gamma_1^{\alpha,k}-\Gamma_1^{\alpha,k+p}\big)(\mu)\big| & \leq C_T\big(1+L\big)^\frac{1}{2}\varepsilon \\
	    \big|\big(\Gamma_2^{\alpha,k}-\Gamma_2^{\alpha,k+p}\big)(\mu)\big| & \leq C_T\big(1+L\big)^\frac{3}{2}\varepsilon
	    \end{align*}
	    Which shows that $\Gamma_1^{\alpha,k}$, $\Gamma_2^{\alpha,k}$, and thus $\Gamma^{\alpha,k}$ converge uniformly on $K_L$. It is not difficult to see that the respective limits are $\Gamma_1^{\alpha}$, $\Gamma_2^{\alpha}$ and $\Gamma^{\alpha}$, which implies $\Gamma^{\alpha}=\Gamma_1^{\alpha}+\Gamma_2^{\alpha}$ on $K_L$. Besides, as $\Gamma^k = \sum_{\alpha=1}^M \Gamma^{\alpha,k}$ and $\Gamma = \sum_{\alpha=1}^M \Gamma^{\alpha}$, we also have the uniform convergence of $\Gamma^k$ towards $\Gamma$ on $K_L$.
	\item[(iii)] The proof is identical to Lemma\ref{lemma1} (iii) and (iv).
	\item[(iv)] Lets show that $\big\{H \leq L\big\}$ is a compact set.
	$H \geq (1-a) I(|P) - \eta$ so that $I(|P) \leq \frac{H+\eta}{1-a}$. Hence $\big\{H \leq L\big\}\subset \big\{I(|P) \leq \frac{L+\eta}{1-a}\big\}$.
	Let $(\mu_p)_p \in \big\{H \leq L\big\}^\mathbb{N} \subset \big\{I(|P) \leq \frac{L+\eta}{1-a}\big\}^\mathbb{N}$.
	As here $\big\{I(|P) \leq \frac{L+\eta}{1-a}\big\}$ is a compact set, there exists a subsequence $(\mu_{p_m})_m$ such that $\mu_{p_m} \to \mu$ as $m \to \infty$. We conclude by stating that, as $H^k$ converge uniformly towards $H$ on $\big\{I(|P) \leq \frac{L+\eta}{1-a}\big\}$, the latest inherits the lower semi-continuity of the firsts. Hence $\big\{H \leq L\big\}$ is a closed set so that $\mu \in \big\{H \leq L\big\}$ and $(\mu_{p_m})_m$ converges in $\big\{H \leq L\big\}$.
	\end{description}
	\end{proof}


	\begin{lemma}\label{lemma2}

	\begin{equation*}
	\frac{dQ^n}{dP^{\otimes n}} = \exp{\Big\{n \Gamma(\hat{\mu}_n) \Big\} }
	\end{equation*}
	\end{lemma}

	\begin{proof}
	By \eqref{eq:Density}, we have:
	\begin{multline*}
		\der{\Q^n(J)}{P^{\otimes n}} = \exp\Bigg(\sum_{i=1}^n  \int_{0}^{T}  \Big(\Sigma^{-1}\cdot \sum_{j=1}^n   \J_{ij} \odot S(X_{t-\btau}^j)\Big)' \cdot  d\W_{t}^{i} - \frac{1}{2} \int_{0}^{T} \Big\|\Sigma^{-1} \sum_{j=1}^n \J_{ij} \odot S(X^j_{t-\btau})\Big\|^2 dt\Bigg)
	\end{multline*}
	Applying Fubini Theorem, we find that $Q^n \ll P^{\otimes n}$ and:
	\begin{equation*}
		\der{Q^n}{P^{\otimes n}} = {\Exp}_{J} \bigg[ \exp\Bigg(\sum_{i=1}^n  \int_{0}^{T} \Big(\Sigma^{-1}\cdot \sum_{j=1}^n   \J_{ij} \odot S(X_{t-\btau}^j)\Big)'  \cdot d\W_{t}^{i}
		- \frac{1}{2} \int_{0}^{T} \Big\|\Sigma^{-1} \sum_{j=1}^n \J_{ij} \odot S(X^j_{t-\btau})\Big\|^2 dt\Bigg) \bigg].
	\end{equation*}
	But, under ${\Exp}_{J}$, the $J_{ij}$ are independent, so that:
	\begin{equation*}
		\der{Q^n}{P^{\otimes n}} = \prod_{i=1}^n {\Exp}_{J} \bigg[ \exp\Bigg(\int_{0}^{T}  \Big(\Sigma^{-1}\cdot \sum_{j=1}^n   \J_{ij} \odot S(X_{t-\btau}^j)\Big)' \cdot  d\W_{t}^{i}
		- \frac{1}{2} \int_{0}^{T} \Big\|\Sigma^{-1} \sum_{j=1}^n \J_{ij} \odot S(X^j_{t-\btau})\Big\|^2 dt\Bigg) \bigg].
	\end{equation*}
	Lets show that $\bigg\{ \Big(\Sigma^{-1}\cdot \sum_{j=1}^n   \J_{ij} \odot S(X_{t-\btau}^j)\Big) , 0 \leq t \leq T \bigg\}$ is an $M$-dimensional Gaussian process with covariance $K_{\hat{\mu}_n}(t,s)$, and mean $\m_{\hat{\mu}_n}(t)$. In fact
	\begin{equation*}
	\Big(\Sigma^{-1}\cdot \sum_{j=1}^n   \J_{ij} \odot S(X_{t-\btau}^j)\Big) = \left(\frac{1}{\la} \sum_{j=1}^n \sum_{\gamma=1}^M J_{i_{\alpha} j_{\gamma}} \Sag(x^{j_{\gamma}}_{t-\taag})\right)_{\alpha=1 \cdots M}
	\end{equation*}
	Lets remember that the functions $x^j, j=1\cdots N$ are, in this equality, seen as determinists as arguments of the studied density. With this in mind and because of the independence of the gaussian random variables $J_{ij}$, we can see that the components of the above vector are gaussian processes, mutually independents. Moreover, the mean and covariance of the component $\alpha \in \{1\cdots M\}$ are the following:
	\begin{align*}
	\Exp\bigg[\frac{1}{\la} \sum_{j=1}^n \sum_{\gamma=1}^M J_{i_{\alpha} j_{\gamma}} \Sag(x^{j_{\gamma}}_{t-\taag})\bigg] & = \frac{1}{\la} \sum_{\gamma=1}^M \Jag \frac{1}{n} \sum_{j=1}^n \Sag(x^{j_{\gamma}}_{t-\taag})\\
	& = m^{\alpha}_{\hat{\mu}_n}(t)
	\end{align*}
	\begin{align*}
	\Exp\bigg[\Big(\frac{1}{\la} &  \sum_{j=1}^n \sum_{\gamma=1}^M \big(J_{i_{\alpha} j_{\gamma}} - \frac{\Jag}{n} \big) \Sag(x^{j_{\gamma}}_{t-\taag})\Big)\Big(\frac{1}{\la} \sum_{j=1}^n \sum_{\gamma=1}^M \big(J_{i_{\alpha} j_{\gamma}} - \frac{\Jag}{n} \big) \Sag(x^{j_{\gamma}}_{s-\taag})\Big) \bigg]\\
	& = \frac{1}{\la^2} \sum_{\gamma=1}^M \sag^2 \frac{1}{n} \sum_{j=1}^n \Sag(x^{j_{\gamma}}_{t-\taag})\Sag(x^{j_{\gamma}}_{s-\taag})\\ = K^{\alpha}_{\hat{\mu}_n}(t)
	\end{align*}

	We eventually find that:
	\begin{align*}
	\der{Q^n}{P^{\otimes n}} & = \exp\bigg\{ \sum_{i=1}^n \log{ \bigg( \int \exp\Big\{ \int_0^T \big(\Gv_t + \m_{\hat{\mu}_n}(t)\big)' \cdot d\W^i_t - \frac{1}{2} \int_0^T \Big\|\Gv_t + \m_{\hat{\mu}_n}(t)\Big\|^2 dt \Big\}  d\gamma_{\hat{\mu}_n} \bigg) }\bigg\}\\
	& = \exp\bigg\{ n \int \log{ \bigg( \int \exp\Big\{ \int_0^T \big(\Gv_t + \m_{\hat{\mu}_n}(t)\big)' \cdot d\W^i_t - \frac{1}{2} \int_0^T \Big\|\Gv_t + \m_{\hat{\mu}_n}(t)\Big\|^2 dt \Big\}  d\gamma_{\hat{\mu}_n} \bigg) } d\hat{\mu}_n \bigg\}
	\end{align*}
	\end{proof}

	Let, for any $\nu \in \M_1^+(\C), \mu \in \M_1^+(\C)$,
	\begin{equation*}
	\Gamma_{\nu}(\mu)= \int \int \exp\bigg\{ \int_0^T (\Gv_t + \m_{\nu}(t))' \cdot d\W_t - \frac{1}{2} \int_0^T \Big\|\Gv_t + \m_{\nu}(t)\Big\|^2 dt \bigg\} d\gamma_{\nu} d\mu,
	\end{equation*}
	\begin{equation*}
	\Gamma^{\alpha}_{\nu}(\mu)= \int \int \exp\bigg\{ \int_0^T (G^{\alpha}_t + m^{\alpha}_{\nu}(t)) dW^{\alpha}_t - \frac{1}{2} \int_0^T \big(G^{\alpha}_t + m^{\alpha}_{\nu}(t)\big)^2 dt \bigg\} d\gamma_{\nu} d\mu,
	\end{equation*}
	\begin{align*}
	H_{\nu}(\mu) =
	\left\{
	\begin{array}{rl}
	I(\mu|P) - \Gamma_{\nu}(\mu) & \text{if } I(\mu|P) < \infty,\\
	\infty & \text{otherwise }.
	\end{array}
	\right.
	\end{align*}
	\begin{equation*}
	dQ_{\nu}(x) = \exp\{ \Gamma_{\nu}(\delta_{x}) \} dP(x).
	\end{equation*}
	As in Lemma~\ref{lemma1}. $(v)$, we can show that $H_{\nu} = I(|Q_{\nu})$ and is therefore semi-continuous.

	\begin{theorem}\label{lemma3}
	For any compact subset $K$ of $\M_1^+(\C)$,
	\begin{equation*}
	\limsup_{n \rightarrow \infty} \frac{1}{n} \log{ Q^n(\hat{\mu}_n \in K)} \leq - \inf_{K} H.
	\end{equation*}
	\end{theorem}

	\begin{proof}
	Let $\delta < 0$. We can find an integer $m$ and a family $(\nu_i)_{1\leq i \leq m}$ of probability measure on $\C$ such that
	$$ K \subset \bigcup_{i=1}^m B(\nu_i,\delta),$$
	where $B(\nu_i,\delta)= \big\{ \mu | d_T(\mu,\nu_i) < \delta \big\}$ and $d_T$ denotes the Vaserstein distance on $\M_1^+(\C)$.
	A very classical result (see e.g.~\cite[lemma 1.2.15]{dembo-zeitouni:09}), ensures that
	\begin{equation*}
	\limsup \frac{1}{n} \log{ Q^n(\hat{\mu}_n \in K)} \leq \max_{1 \leq i \leq m} \limsup \frac{1}{n} \log{Q^n(\hat{\mu}_n \in K\cap B(\nu_i,\delta))  } .
	\end{equation*}
	Let $\nu \in \M_1^+(\C)$. Lemma~\ref{lemma2} gives us:
	\begin{align*}
	Q^n(\hat{\mu}_n \in K \cap B(\nu,\delta)) & = \int_{\hat{\mu}_n \in K \bigcap B(\nu,\delta)} \exp\Big\{ n \Gamma(\hat{\mu}_n)\Big\} dP^{\otimes n} \\
	& = \int_{\hat{\mu}_n \in K \bigcap B(\nu,\delta)} \exp\Big\{ n \big(\Gamma(\hat{\mu}_n)- \Gamma_{\nu}(\hat{\mu}_n)\big)\Big\}  \exp\Big\{ n \Gamma_{\nu}(\hat{\mu}_n)\Big\} dP^{\otimes n}
	 \end{align*}
	But, for any probability measure $\nu \in \M_1^+(\C)$, $Q_{\nu}^n := \exp{\Big\{ n \Gamma_{\nu}(\hat{\mu}_n)\Big\} } P^{\otimes n} = \big(Q_{\nu}\big)^{\otimes n}$ is a probability measure on $\big(\C\big)^{\otimes n}$. Hence, for any conjugate exponents $(p,q)$,
	\begin{align}
	Q^n(\hat{\mu}_n \in K \cap B(\nu,\delta)) & = \int_{\hat{\mu}_n \in K \bigcap B(\nu,\delta)} \exp\Big\{ n \big(\Gamma(\hat{\mu}_n)- \Gamma_{\nu}(\hat{\mu}_n)\big)\Big\} dQ_{\nu}^n \nonumber \\
	& \leq Q_{\nu}^n\big(\hat{\mu}_n \in K \cap B(\nu,\delta)\big)^{\frac{1}{p}}  \bigg(\int_{\hat{\mu}_n \in K \bigcap B(\nu,\delta)} \exp\Big\{ q n\big(\Gamma(\hat{\mu}_n)- \Gamma_{\nu}(\hat{\mu}_n)\big)\Big\} dQ_{\nu}^n\bigg)^{\frac{1}{q}} \label{ineqpi}
	\end{align}

	We will first bound the second term of the right hand side by proving the following lemma~\ref{lemma3.1}. Once this step performed, concluding the proof amounts bounding the first term in the right hand side of (\ref{ineqpi}) using the same arguments as in \cite[Lemma 3.8]{guionnet:97}. Remark that, as the space we work on remains a polish and the $X^i, i=1 \cdots n$ are i.i.d under $Q_{\nu}$, we can still resort to Sanov Theorem.
	\end{proof}

	\begin{lemma}\label{lemma3.1}
	For any real number $a >1$, there exists a strictly positive real number $\delta_a$ such that, for any $\delta < \delta_a$, there exists a function $C_a(.)$ in $\R$ such that $\lim_{\delta \rightarrow 0} C_a(\delta) =0$ and:
	\begin{equation*}
	\int_{\hat{\mu}_n \in B(\nu,\delta)} \exp{\Big\{ a n \big(\Gamma(\hat{\mu}_n)- \Gamma_{\nu}(\hat{\mu}_n)\big)\Big\} } dQ_{\nu}^n \leq \exp\{C_a(\delta) n\}
	\end{equation*}
	\end{lemma}

	The proof of this lemma is relatively technical, and it is provided in appendix~\ref{append:ProofLemma3.1}.

	\begin{theorem}[Tightness]\label{thm:tightness}
	For any real number $\varepsilon > 0$, there exists a compact set $K_{\varepsilon}$ of $\M_1^+(\C)$ such that, for any integer $n$,
	\begin{equation*}
	Q^n(\hat{\mu}_n \notin K_{\varepsilon}) \leq \varepsilon.
	\end{equation*}
	\end{theorem}

	\begin{proof}
	The proof of this theorem consists in using the relative entropy inequality ~\eqref{eq:IneqRelativeEntropy} and to use the exponential tightness of the sequence of laws $P^{\otimes n}$. Indeed, defining $A$ an arbitrary set in $\C^n$ and applying~\eqref{eq:IneqRelativeEntropy} to the function $\Phi=\log(1+P^{\otimes n}(A)^{-1})\mathbbm{1}_A$ yields:
	\[Q^n(A)\leq \frac{\log (2) + I(Q\vert P^{\otimes n})}{\log(1+P^{\otimes n}(A)^{-1})}\]
	and the exponential tightness of $P^{\otimes n}$ (see e.g.~\cite[lemma 3.2.7]{deuschel-stroock:89}) ensures that for any $\varepsilon>0$ there exists a compact subset $K_{\varepsilon}$ of $\M_1^+(\C)$ such that
	\[P^n (\hat{\mu}_n\notin K_{\varepsilon})\leq \exp\left\{-\frac{n}{\varepsilon}\right\}.\]

	The theorem is hence proved as soon as we show that there exists a finite constant $C$, such that for any integer $n$, $I(Q^n|P^{\otimes n}) \leq C n$. Using the expression of the relative entropy and the interchangeability of the neurons of the same population, we find:
	\begin{equation}\label{eq:ExpressionOfRelativeEntropy}
	I(Q^n|P^{\otimes n})=n \int \log{{\E}_{\hat{\mu}_n}\bigg[ \exp\bigg\{ \int_0^T \big(\Gv_t + \m_{\hat{\mu}_n}(t)\big)' \cdot d\W^1_t - \frac 1 2 \int_0^T \Big\|\Gv_t + \m_{\hat{\mu}_n}(t) \Big\|^2dt \bigg\}\bigg]}dQ^n
	\end{equation}
	For every $\alpha \in \{1\cdots M\}$ let,
	\begin{equation*}
	dQ^{\Na}= {\E}_{\hat{\mu}_n}\bigg[ \exp\bigg\{ \int_0^T \big(G^{\alpha}_t + m^{\alpha}_{\hat{\mu}_n}(t)\big) dW^{1_{\alpha}}_t - \frac 1 2 \int_0^T \big(G^{\alpha}_t + m^{\alpha}_{\hat{\mu}_n}(t) \big)^2dt \bigg\}\bigg] d\pa^{\otimes \Na}
	\end{equation*}
	which is a probability measure on $\C([-\tau,T],\R)^{\Na}$. One can see that $Q^n = \otimes_{\alpha=1}^M Q^{\Na}$. After some gaussian computations (see \cite[Lemma 5.15]{ben-arous-guionnet:95}), we find that:
	\begin{equation}\label{eq:NewExpressionOfDensity}
	{\E}_{\hat{\mu}_n}\bigg[ \exp\bigg\{ \int_0^T \big(G^{\alpha}_t + m^{\alpha}_{\hat{\mu}_n}(t)\big) dW^{1_{\alpha}}_t - \frac 1 2 \int_0^T \big(G^{\alpha}_t + m^{\alpha}_{\hat{\mu}_n}(t) \big)^2dt \bigg\}\bigg] = \exp\bigg\{\int_0^T H^{\alpha}_t(Q^{\Na}) dW^{1_{\alpha}}_t-\frac{1}{2}\int_0^T \big(H^{\alpha}_t(Q^{\Na})\big)^2dt\bigg\}
	\end{equation}
	where
	\begin{align*}
	H^{\alpha}_t (Q^{\Na}) & = {\E}_{\hat{\mu}_n} \bigg[ G^{\alpha}_t \Lambda^{\Na}_t  \int_0^t G^{\alpha}_s \big(dW^{1_{\alpha}}_s - m^{\alpha}_{\hat{\mu}_n}(s)ds\big) \bigg] + m^{\alpha}_{\hat{\mu}_n}(t)
	\end{align*}
	\begin{equation*}
	\Lambda^{\Na}_t= \frac{\exp\bigg\{ -\frac{1}{2} \int_0^t {G^{\alpha}_s}^2 ds \bigg\}}{{\E}_{\hat{\mu}_n} \bigg[ \exp\bigg\{ -\frac{1}{2} \int_0^t {G^{\alpha}_s}^2 ds \bigg\} \bigg]}.
	\end{equation*}

	Consequently, there exists a $Q^{\Na}$ brownian motion $B^{1_{\alpha}}$ such that:
	\begin{equation*}
	W^{1_{\alpha}}_t= B^{1_{\alpha}}_t + \int_0^t H^{\alpha}_s (Q^{\Na}) ds
	\end{equation*}

	Using ~\eqref{eq:NewExpressionOfDensity} and this brownian motion in ~\eqref{eq:ExpressionOfRelativeEntropy}, we have:
	\begin{align}
	& I(Q^n|P^{\otimes n}) = \frac{n}{2} \sum_{\alpha=1}^n \int \int_0^T \bigg( {\E}_{\hat{\mu}_n} \bigg[ \Lambda_t^{\Na} G^{\alpha}_t \int_0^t  G^{\alpha}_s \big(dW^{1_{\alpha}}_s-m^{\alpha}_{\hat{\mu}_n}(s)ds\big) \bigg] + m^{\alpha}_{\hat{\mu}_n}(t) \bigg)^2 dt dQ^n \nonumber \\
	& \leq n \sum_{\alpha=1}^n \int \int_0^T  \bigg( {\E}_{\hat{\mu}_n} \bigg[ \Lambda_t^{\Na} G^{\alpha}_t \times \int_0^t  G^{\alpha}_s \big(dW^{1_{\alpha}}_s-m^{\alpha}_{\hat{\mu}_n}(s)ds\big) \bigg] \bigg)^2 + {m^{\alpha}_{\hat{\mu}_n}}^2(t) dt dQ^n \nonumber \\
	& \leq n \sum_{\alpha=1}^n \times \bigg\{ \int_0^T \int \underbrace{\int  \bigg( {\E}_{\hat{\mu}_n} \bigg[ \Lambda_t^{\Na} G^{\alpha}_t \times \int_0^t  G^{\alpha}_s \big(dW^{1_{\alpha}}_s-m^{\alpha}_{\hat{\mu}_n}(s)ds\big) \bigg] \bigg)^2 dQ^{\Na}}_{f^{\alpha}(t)} d(\otimes_{\gamma \neq \alpha} Q^{\Ng}) dt + \frac{\Ja^2 T}{\la^2} \bigg\} \label{ineqI}
	\end{align}
	We now will bound $f^{\alpha}(t)$:
	\begin{align*}
	f^{\alpha}(t) & \leq 3 \bigg\{ \int  {\E}_{\hat{\mu}_n} \bigg[ \Lambda_t^{\Na} G^{\alpha}_t  \int_0^t  G^{\alpha}_s dB^{1_{\alpha}}_s \bigg]^2 dQ^{\Na} + \int  {\E}_{\hat{\mu}_n} \bigg[ \Lambda_t^{\Na} G^{\alpha}_t \int_0^t  G^{\alpha}_s m^{\alpha}_{\hat{\mu}_n}(s)ds \; \bigg]^2 dQ^{\Na} \\
	& + \int  \bigg( \int_0^t {\E}_{\hat{\mu}_n} \bigg[ \Lambda_t^{\Na} G^{\alpha}_t G^{\alpha}_s \bigg] {\E}_{\hat{\mu}_n} \bigg[ \Lambda_s^{\Na} G^{\alpha}_s \int_0^s G^{\alpha}_u \big(dW^{1_{\alpha}}_u-m^{\alpha}_{\hat{\mu}_n}(u)du\big) \bigg] ds \bigg)^2 dQ^{\Na} \bigg\}
	\end{align*}
	But by Cauchy-Schwarz inequality,
	\begin{align*}
	\int  {\E}_{\hat{\mu}_n} \bigg[ \Lambda_t^{\Na} G^{\alpha}_t \times & \int_0^t  G^{\alpha}_s dB^{1_{\alpha}}_s \; \bigg]^2 dQ^{\Na}  \leq \int \bigg\{  {\E}_{\hat{\mu}_n} \bigg[ \Big(\Lambda_t^{\Na} G^{\alpha}_t\Big)^2 \bigg] {\E}_{\hat{\mu}_n} \bigg[ \Big( \int_0^t  G^{\alpha}_s dB^{1_{\alpha}}_s\Big)^2 \; \bigg] \bigg\} dQ^{\Na}
	\end{align*}
	As ${\E}_{\hat{\mu}_n} \bigg[ \Big(\Lambda_t^{\Na} G^{\alpha}_t\Big)^2 \bigg] \leq \frac{\ka}{\la^2}$ (see \cite{ben-arous-guionnet:95} Appendix A), we have by Fubini Theorem and Ito's isometry:
	\begin{align*}
	\int {\E}_{\hat{\mu}_n} \bigg[ \Lambda_t^{\Na} G^{\alpha}_t \times & \int_0^t  G^{\alpha}_s dB^{1_{\alpha}}_s \; \bigg]^2 dQ^{\Na} \leq  \frac{\ka}{\la^2} \int \bigg\{ {\E}_{\hat{\mu}_n} \bigg[ \bigg( \int_0^t  {G^{\alpha}_s}^2 ds \bigg) \bigg] \bigg\} dQ^{\Na} \leq \frac{\ka^2 T}{\la^4}
	\end{align*}
	By similar arguments,
	\begin{align*}
	\int {\E}_{\hat{\mu}_n} \bigg[ \Lambda_t^{\Na} G^{\alpha}_t \times \int_0^t  G^{\alpha}_s m^{\alpha}_{\hat{\mu}_n}(s)ds \; \bigg]^2 dQ^{\Na} & \leq \int \bigg\{  {\E}_{\hat{\mu}_n} \bigg[ \Big(\Lambda_t^{\Na} G^{\alpha}_t\Big)^2 \bigg] {\E}_{\hat{\mu}_n} \bigg[ \Big( \int_0^t  G^{\alpha}_s m^{\alpha}_{\hat{\mu}_n}(s) ds\Big)^2 \; \bigg] \bigg\} dQ^{\Na}\\
	& \leq \frac{\ka}{\la^2} \int \bigg\{ {\E}_{\hat{\mu}_n} \bigg[ \frac{\Ja^2 T}{\la^2} \bigg( \int_0^t  {G^{\alpha}_s}^2 ds \bigg)^2 \; \bigg] \bigg\} dQ^{\Na} \leq \frac{\Ja^2 \ka^2 T^2}{\la^6}
	\end{align*}

	And,
	\begin{align*}
	\int  \bigg( \int_0^t &  {\E}_{\hat{\mu}_n} \bigg[ \Lambda_t^{\Na} G^{\alpha}_t G^{\alpha}_s \bigg] {\E}_{\hat{\mu}_n} \bigg[ \Lambda_s^{\Na} G^{\alpha}_s \int_0^{s} G^{\alpha}_u \big(dW^{1_{\alpha}}_u-m^{\alpha}_{\hat{\mu}_n}(u)du\big) \bigg] ds \bigg)^2 dQ^{\Na} \leq \Big(\frac{\ka}{\la^2}\Big)^2 T  \int_0^t f^{\alpha}(s) ds.
	\end{align*}

	Eventually, we proved the following inequality:
	\begin{equation*}
	f^{\alpha}(t) \leq 3\Big(\frac{\ka^2 T}{\la^4} + \frac{\Ja^2 \ka^2 T^2}{\la^6} \Big) + 3\frac{\ka^2 T}{\la^4 } \int_0^t f^{\alpha}(s) ds
	\end{equation*}
	So that, by Gronwall lemma,
	\begin{equation*}
	\sup_{t \leq T} f^{\alpha}(t) \leq  3\frac{\ka^2 T}{\la^4}\big(1+\frac{\Ja^2 T}{\la^2}\big)\exp{\Big\{3\frac{\ka^2 T}{\la^4 } \Big\} }.
	\end{equation*}
	Thus, (\ref{ineqI}) implies:
	\begin{equation*}
	I(Q^n|P^{\otimes n}) \leq n \sum_{\alpha=1}^n \Big( 3\frac{\ka^2 T^2}{\la^4}\big(1+\frac{\Ja^2 T}{\la^2}\big)\exp{\Big\{3\frac{\ka^2 T}{\la^4 } \Big\} } + \frac{\Ja^2 T}{\la^2} \Big)
	\end{equation*}
	\end{proof}

	\subsection{Identification of the mean-field equations}\label{sec:LimitIdentification}
	We have seen that the empirical laws $\hat{\mu}_n$ satisfy a weak large deviations upper bound with good rate function $H$. In order to identify the limit of the system, we study the minima of the functions $H$ through a variational study. We will show in section~\ref{sec:Charact}that any minimum of all $H$ are measures $Q \ll P$ satisfying the implicit equation:
	\begin{equation}\label{eq:densityMinimum}
		\der{Q}{P} = \int \exp\left\{ \sum_{\alpha}\int_0^T G_t^{\alpha}+m^{\alpha}_Q(t) dW^{\alpha}_t - \frac 1 2 \int_0^T (G_t^{\alpha}+m^{\alpha}_Q(t))^2 dt\right\} d\gamma_Q
	\end{equation}
	We will then prove in section~\ref{sec:ExistUniqueLim} that there exists a unique probability measure $Q$ satisfying~\eqref{eq:densityMinimum}. This law will be further analyzed in the next section.

	\subsubsection{Variational characterization of the minima of the good rate function}\label{sec:Charact}

	The large deviation principle ensures that the sequence of empirical measures converge, and that the possible limits minimize the good rate functions. We hence need to identify the minima and show that these are uniquely defined by equation~\eqref{eq:densityMinimum}. To this purpose, we start by showing that any probability density achieving the minimum of all the $H$ is equivalent to $P$. To this end, we show the following technical result:
	\begin{lemma}\label{lem:QequivP}
		Let $Q$ be a probability measure on $\C$ which minimizes $H$. Then $Q\ll P$. Moreover, noting $B=\{\omega ; \der{Q}{P}=0\}$ and $\delta=P(B)$, we have, noting $Q_s=\frac{Q+s\mathbbm{1}_{B}P}{1+s\delta}$:
		\begin{itemize}
			\item $I(Q_s\vert P) = I(Q\vert P) + s\delta \log(s) + O(s)$
			\item $\Gamma(Q_s)=\Gamma(Q)+O(s)$
		\end{itemize}
	\end{lemma}

	\begin{proof}
		If $Q$ minimizes $H$, then necessarily $I(Q\vert P)$ is finite, meaning that $Q\ll P$. Moreover, it is easy to see that:
		\begin{align*}
			I(Q_s\vert P)&=\int \log(\der{Q_s}{P}) dQ_s\\
			&=\frac{1}{1+s\delta} \bigg\{\int (\log(\der{Q}{P} + s\mathbbm{1}_{B}) - \log(1+s\delta)) (\der{Q}{P} +s \mathbbm{1}_{B}) dP\bigg\}\\
			&=\frac{1}{1+s\delta} \bigg\{ \int_{B^c} \log(\der{Q}{P})\der{Q}{P}dP +  s\log(s)\int_{B} dP - \log(1+s\delta) \int (\der{Q}{P} +s \mathbbm{1}_{B}) dP \bigg\}\\
			&= \frac{1}{1+s\delta} \Big( I(Q\vert P) + s\delta \log(s) - (1+s\delta)\log(1+s\delta)\Big)
		\end{align*}
		which proves the first point.

		The second point is proved using standard Gaussian calculus noting that if $G$ and $V$ are independent centered Gaussian processes with covariance $K_Q$ and $K_{\mathbbm{1}_B P}$, then
		 $G^s=\frac{G+\sqrt{s}V}{\sqrt{1+\delta s}}$ has the covariance $K_{Q_s}$.
	We hence have $\m_{Q_s}(t) = \frac{\m_{Q}(t)+s\,\m_{\mathbbm{1}_B P}(t)}{(1+s\delta)}$, and we can write:
		\begin{align*}
			\Gamma(Q_s) & = \int \log\Bigg (\int \exp \bigg \{\int_0^T \bigg( \frac{G_t+\sqrt{s}V_t}{\sqrt{1+s\delta}} + \frac{\m_Q(t)+s\,\m_{\mathbbm{1}_B P}(t)}{1+s\delta}\bigg)' \cdot d\W_t \\
			& -\frac 1 2 \int_0^T \left\Vert\frac{G_t+\sqrt{s}V_t}{\sqrt{1+s\delta}} + \frac{\m_Q(t)+s\,\m_{\mathbbm{1}_B P}(t)}{1+s\delta}\right\Vert^2 \, dt\bigg\}d\gamma_Q \otimes\gamma_{\mathbbm{1}_B P}\Bigg) dQ_s
		\end{align*}
		The exponential term is given by:
		\begin{align*}
			\exp \bigg \{\sum_{\alpha}\int_0^T \Big( G_t^{\alpha}+m_Q^{\alpha}(t) \Big)\, dW^{\alpha}_t -\frac 1 2 \int_0^T \left(G^{\alpha}_t + m^{\alpha}_Q(t)\right)^2 \, dt\bigg\} \\
			\times \Bigg ( 1 + \sqrt{s}  \Big\{\sum_{\alpha} \int_0^T V_t^{\alpha} dW^{\alpha}_t -\int_0^T V_t^{\alpha}\,(G^{\alpha}_t+m^{\alpha}_Q(t))dt\Big\} + sR(s)\Bigg)
		\end{align*}
		Using the fact that $G^{\alpha}$ and $V^{\alpha}$ are Gaussian processes with bounded covariances and using the fact that the mean quadratic variation of $W^{\alpha}$ under $Q_s^{\alpha}$ is also bounded (see~\cite{ben-arous-guionnet:95}) we can obtain that:
		\[\int \int  R(s)\frac{\exp \bigg \{\sum_{\alpha=1}^M\int_0^T \Big( G^{\alpha}_t+m^{\alpha}_Q(t)\, \Big) dW^{\alpha}_t -\frac 1 2 \int_0^T \left(G^{\alpha}_t + m^{\alpha}_Q(t)\right)^2 \, dt\bigg\}}{\int \exp \bigg \{\sum_{\alpha=1}^M\int_0^T  \Big(G^{\alpha}_t+m^{\alpha}_Q(t)\Big)\, dW^{\alpha}_t -\frac 1 2 \int_0^T \left(G^{\alpha}_t + m^{\alpha}_Q(t)\right)^2 \, dt\bigg\}d\gamma_Q\otimes\gamma_{\mathbbm{1}_B P}} d\gamma_Q\otimes\gamma_{\mathbbm{1}_B P} dQ_s = O(1)\]
		and we finally obtain:
		\begin{align*}
			\Gamma(Q_s)  & = \int \log\Bigg (\int \exp \bigg \{\sum_{\alpha=1}^M\int_0^T \Big(G^{\alpha}_t+m^{\alpha}_Q(t)\Big) \, dW^{\alpha}_t -\frac 1 2 \int_0^T \left(G^{\alpha}_t + m^{\alpha}_Q(t) \right)^2 \, dt\bigg\} \\
			&\times \left(1+\sqrt{s} \left( \sum_{\alpha=1}^M \int_0^T V_t^{\alpha} dW^{\alpha}_t -\int_0^T V_t^{\alpha}\,(G^{\alpha}_t+m^{\alpha}_Q(t))dt\right)\right)  d\gamma_Q \otimes\gamma_{\mathbbm{1}_B P} \Bigg )dQ_s +O(s)
		\end{align*}
		Integrating with respect to $\gamma_{\mathbbm{1}_B P}$ and using the independence of $G$ and $V$ and the fact that $V$ is centered, we obtain the desired result.
	\end{proof}

	This lemma ensures that any minimum of $H$ is equivalent to $P$. Indeed, if $\delta^{\alpha}>0$, then the result of lemma~\ref{lem:QequivP} implies that $H(Q_s)-H(Q)\sim s\delta \log(s)$ which is strictly negative for small $s$, contradicting the fact that we assumed that $Q$ minimized $H$. We therefore necessarily have $\delta=P(B)=0$, that is $Q\simeq P$.

	Let us now characterize the minima of $H$. To this purpose, we use a variational formulation to show that any minimum of $H$ satisfies equation~\eqref{eq:densityMinimum}, and start by proving the following technical result:
	\begin{lemma}\label{lem:Density}
		Let $\Phi$ be a positive and bounded measurable function on $\C$ such that $\int \Phi dQ=1$, and denote $\psi=\Phi-1$ and $Q_s(\Phi)=\frac{1+s\Phi}{1+s} Q$. We have:
		\begin{itemize}
			\item $I(Q_s(\Phi)\vert P)=I(Q\vert P) + s\int \Psi \log \der{Q}{P}dQ+O(s^2)$
			\item $\Gamma(Q_s(\Phi)) = \Gamma(Q) + s \bigg\{ \int \left(\log \int \exp \Big\{\sum_{\alpha}\int_0^T( G_t^{\alpha}+m^{\alpha}_Q(t))\, dW^{\alpha}_t - \frac 1 2 \int_0^T (G_t^{\alpha}+m^{\alpha}_Q(t))^2 dt\Big\}d\gamma_Q+Y_T\right) d\Psi Q + \int Y_T(y) dQ(y) + C_Q(\Phi) \bigg\} + O(s^{\frac 3 2})$ where $Y_T$ is an adapted process with finite variation and $C_Q$ is a bounded function.
		\end{itemize}	
	\end{lemma}

	\begin{proof}
		The first point is simply proved as follows:
		\begin{align*}
			I(Q_s(\Phi)\vert P) &= \int \log\left(\frac{1+s\Phi}{1+s}\right)\frac{1+s\Phi}{1+s}dQ + \int \log\left(\der{Q}{P}\right)\frac{1+s\Phi}{1+s}dQ.
		\end{align*}
		Noting that $\frac{1+s\Phi}{1+s} = 1+s\psi + O(s^2)$ and $\int\psi dQ=0$, we readily obtain the desired result.

		The second point is slightly more delicate but based on the same argument as outlined in the proof of lemma~\ref{lem:QequivP}. We introduce two independent centered Gaussian processes $G$ and $V$ with covariances $K_Q$ and $K_{\Phi Q}$ respectively and build the Gaussian process $\frac{G+\sqrt{s}V}{\sqrt{1+s}}$ that has the covariance given by $K_{Q_s(\Phi)}$. We then express $\Gamma(Q_s(\Phi))$ as:
		\begin{multline*}
			\Gamma(Q_s(\Phi)) = \int \Bigg(\log \int \exp \Big\{\sum_{\alpha}\int_0^T \frac{G^{\alpha}_t+\sqrt{s}V_t^{\alpha}}{\sqrt{1+s}} + \frac{m^{\alpha}_Q(t) +s\,m^{\alpha}_{\Phi Q}(t)}{1+s} dW^{\alpha}_t \\- \frac 1 2 \int_0^T \left(\frac{G^{\alpha}_t+\sqrt{s}V^{\alpha}_t}{\sqrt{1+s}}+\frac{m^{\alpha}_Q(t)+s\,m^{\alpha}_{\Phi Q}(t)}{1+s} \right)^2 dt\Big\}d\gamma_Q\otimes \gamma_{\Phi Q}\Bigg)dQ_s(\Phi)
		\end{multline*}
		and a series expansion ensures that:
		 \begin{align*}
		 	& \exp  \Big\{ \sum_{\alpha} \int_0^T \left(\frac{G^{\alpha}_t+\sqrt{s}V^{\alpha}_t}{\sqrt{1+s}} + \frac{m^{\alpha}_Q(t)+s\,m^{\alpha}_{\Phi Q}(t)}{1+s} \right)dW^{\alpha}_t - \frac 1 2 \int_0^T \left(\frac{G^{\alpha}_t+\sqrt{s}V^{\alpha}_t}{\sqrt{1+s}}+\frac{m^{\alpha}_Q(t)+s\,m^{\alpha}_{\Phi Q}(t)}{1+s} \right)^2 dt\Big\}\\
			& =\exp  \Big\{ \sum_{\alpha}\int_0^T \left(G^{\alpha}_t + m^{\alpha}_Q(t)\right) dW^{\alpha}_t - \frac 1 2 \int_0^T \left(G^{\alpha}_t + m^{\alpha}_Q(t)\right)^2 dt\Big\}\\
			&\quad \times \Bigg(1+ \sqrt{s} \bigg\{\sum_{\alpha} \int_0^T V^{\alpha}_t dW^{\alpha}_t - \int_0^T V^{\alpha}_t\left(G^{\alpha}_t + m^{\alpha}_Q(t)\right)dt\bigg\} \\
			& \qquad +  s\Bigg\{ \Bigg[ \sum_{\alpha}\frac{1}{2} \left(\int_0^T V^{\alpha}_t dW^{\alpha}_t - \int_0^T V^{\alpha}_t\left(G^{\alpha}_t + m^{\alpha}_Q(t)\right)dt\right)^2  + \int_0^T \bigg(-\frac{G^{\alpha}_t}{2}- m^{\alpha}_Q(t) +m^{\alpha}_{\Phi Q}(t)\bigg)dW^{\alpha}_t \\
			&\qquad \qquad- \frac 1 2 \int_0^T \left((V^{\alpha}_t)^2 + 2( G^{\alpha}_t + m^{\alpha}_Q(t)) ( -\frac{ G^{\alpha}_t}{2}+m^{\alpha}_{\Phi Q}(t)-m^{\alpha}_Q(t)) \right) dt\Bigg] \\
 &\qquad \qquad + \sum_{\alpha \neq \gamma} \bigg( \int_0^T V^{\alpha}_t dW^{\alpha}_t - \int_0^T V^{\alpha}_t\left(G^{\alpha}_t + m^{\alpha}_Q(t)\right)dt\bigg) \bigg(\int_0^T V^{\gamma}_t dW^{\gamma}_t - \int_0^T V^{\gamma}_t\left(G^{\gamma}_t + m^{\gamma}_Q(t)\right)dt\bigg)\Bigg\} + O(s^{\frac 3 2})\Bigg)
		 \end{align*}
		Integrating with respect to $\gamma_{\Phi Q}$ and injecting this expansion in the expression of $\Gamma$ yields the expression:
		\begin{multline*}
			\Gamma(Q_s(\Phi)) = \Gamma(Q) + s \Bigg \{\int \Bigg(\log \int \exp \Big\{ \sum_{\alpha} \int_0^T (G^{\alpha}_t+m^{\alpha}_Q) dW^{\alpha}_t \\- \frac 1 {2} \int_0^T \left(G^{\alpha}_t + m^{\alpha}_Q\right)^2 dt\Big\}d\gamma_Q\Bigg) \Psi dQ + \int X_T(x,\Phi)dQ \Bigg\}+O(s^{\frac 3 2})
		\end{multline*}
		with
		\begin{multline*}
			X_T(x,\Phi)=\sum_{\alpha} \int d\gamma^{x}\otimes d\gamma_{\Phi Q} \Bigg\{\frac{1}{2} \left(\int_0^T V^{\alpha}_t dW^{\alpha}_t - \int_0^T V^{\alpha}_t\left(G^{\alpha}_t + m^{\alpha}_Q(t)\right)dt\right)^2 \\
			+ \int_0^T \bigg(-\frac{G^{\alpha}_t}{2}- m^{\alpha}_Q(t) +m^{\alpha}_{\Phi Q}(t)\bigg)dW^{\alpha}_t
	- \frac 1 2 \int_0^T \left((V^{\alpha}_t)^2 + 2( G^{\alpha}_t + m^{\alpha}_Q(t)) ( -\frac{ G^{\alpha}_t}{2}+m^{\alpha}_{\Phi Q}(t)-m^{\alpha}_Q(t)) \right) dt\Bigg\}
		\end{multline*}
		and
		\[d\gamma^{x}(\omega) = \frac{\exp\left\{\sum_{\alpha}\int_0^T \left(G^{\alpha}_t(\omega) + m^{\alpha}_Q(t)\right) dW^{\alpha}_t(x) - \frac 1 2 \int_0^T \left(G^{\alpha}_t(\omega) + m^{\alpha}_Q(t)\right)^2 dt\right\}}{\int \exp\left\{\sum_{\alpha}\int_0^T \left(G^{\alpha}_t(\tilde{\omega}) + m^{\alpha}_Q(t)\right) dW^{\alpha}_t(x) - \frac 1 2 \int_0^T \left(G^{\alpha}_t(\tilde{\omega}) + m^{\alpha}_Q(t)\right)^2 dt\right\}d\gamma_{Q}(\tilde{\omega})} \, d\gamma_{Q}(\omega)\]

	It is easy to see that $m^{\alpha}_{\Psi Q}= m^{\alpha}_{\Phi Q} - m^{\alpha}_Q$, so that:
		\begin{multline*}
			X_T(x,\Phi)=\sum_{\alpha}\int \Bigg\{\frac{1}{2} \left(\int_0^T V^{\alpha}_t dW^{\alpha}_t - \int_0^T V^{\alpha}_t\left(G^{\alpha}_t + m^{\alpha}_Q(t)\right)dt\right)^2 + \int_0^T \bigg(-\frac{G^{\alpha}_t}{2}+m^{\alpha}_{\Psi Q}(t)\bigg)dW^{\alpha}_t \\
	- \frac 1 2 \int_0^T ((V^{\alpha}_t)^2 - {G^{\alpha}_t}^2) dt + \int_0^T  G^{\alpha}_t \left( \frac{m^{\alpha}_Q(t)}{2} - m^{\alpha}_{\Psi Q}(t) \right) dt - \int_0^T m^{\alpha}_{\Psi Q}(t)m^{\alpha}_Q(t) dt\Bigg\} d\gamma^{x}\otimes d\gamma_{\Phi Q}
		\end{multline*}
		or,
		\begin{equation*}
			X_T(x,\Phi)=\frac{1}{2} \int \int Y_T(x,y)  d\Phi Q(y) d\gamma^{x} + C_Q(x,\Phi)
		\end{equation*}
	where
	\begin{equation*}
	Y_T(x,y)= \sum_{\alpha,\gamma} \frac{\sag^2}{\la^2} \left(\int_0^T \Sag(y^{\gamma}_{t-\taag}) dW^{\alpha}_t - \int_0^T \Sag(y^{\gamma}_{t-\taag}) \left(G^{\alpha}_t + m^{\alpha}_Q(t)\right)dt\right)^2.
	\end{equation*}
	and
	\begin{multline*}
	C_Q(x,\Phi)= \sum_{\alpha}\int \int_0^T \bigg(-\frac{G^{\alpha}_t}{2}+m^{\alpha}_{\Psi Q}(t)\bigg)dW^{\alpha}_t + \int_0^T \frac 1 2 {G^{\alpha}_t}^2 + G^{\alpha}_t \left( \frac{m^{\alpha}_Q(t)}{2} - m^{\alpha}_{\Psi Q}(t) \right) dt  d\gamma^{x} \\
	- \frac 1 2 \int_0^T K^{\alpha}_{\Phi Q}(t,t) dt  - \int_0^T m^{\alpha}_{\Psi Q}(t)m^{\alpha}_Q(t) dt
	\end{multline*}

	Let $Y_T(y)=\int\int Y_T(x,y) d\gamma^x dQ(x)$.
	As $Y_T(x,y) \geq 0$, Fubini Theorem ensures that:
	\begin{multline*}
	\int X_T(x,\Phi) dQ(x) = \int Y_T(y) d\Psi Q(y) + \int Y_T(y) dQ(y) +\Exp_Q \bigg[ C_Q(x,\Phi) \bigg] \\
	\end{multline*}
	As $\forall \alpha \in \{ 1 \cdots M \}  \forall \mu, K^{\alpha}_{\mu} \leq \frac{\ka}{\la^2}, m^{\alpha}_{\mu} \leq \frac{\Ja}{\la}$, we conclude by stating that $C_Q(x,\Phi)$ is an $Q$-integrable process, which ends the proof.
	\end{proof}

	We are now able to prove that any minimum satisfies equation~\eqref{eq:densityMinimum}.
	A necessary condition for $Q$ to minimize $H$ is
	\begin{equation}\label{eq:necessaryCondition}
	\lim_{s \to 0} \frac 1 s \left(H(Q_s(\Phi))- H(Q)\right) \geq 0.
	\end{equation}
	But Lemma~\ref{lem:Density} implies that, for every $\Phi$ such that $\int \Phi dQ=1$ ,
	\begin{multline}\label{eq:diffH}
	H(Q_s(\Phi))- H(Q)= s\int\bigg\{  \log{\der{Q}{P}} - \log\bigg(\int \exp\Big\{\sum_{\alpha}\int_0^T G^{\alpha}_t + m^{\alpha}_Q(t)dW^{\alpha}_t\\
	- \frac 1 2 \int_0^T \Big(G^{\alpha}_t + m^{\alpha}_Q(t)\Big)^2 dt \Big\} d\gamma_Q \bigg) - Y_T \bigg\} d\Psi Q - s \bigg\{ \int Y_T dQ +C_Q(\Phi) \bigg\} + O(s^{\frac 3 2})
	\end{multline}
	Let
	\[
	Z_T=\log{\der{Q}{P}} - \log\bigg(\int \exp\Big\{\sum_{\alpha} \int_0^T G^{\alpha}_t + m^{\alpha}_Q(t)dW^{\alpha}_t - \frac 1 2 \int_0^T \Big(G^{\alpha}_t + m^{\alpha}_Q(t)\Big)^2 dt \Big\} d\gamma_Q \bigg) - Y_T.
	\]
	We can rewrite ~\eqref{eq:diffH} as:
	\begin{equation*}
	H(Q_s(\Phi))- H(Q)= s\int Z_T \Psi dQ - s \bigg\{ \int Y_T dQ +C_Q(\Phi) \bigg\} + O(s^{\frac 3 2}).
	\end{equation*}
	Lets show that for any bounded measurable function $\Psi$ with $\int \Psi dQ=0$, we have $\int \Psi Z_T dQ=0$ and, as a consequence, we can find a constant $c_Q$ such that $Z_T=c_Q$ almost surely under $Q$, and thus $P$-almost surely.
	Indeed, if not, we can find $\Psi$ such that $\int Z_T \Psi dQ\neq 0$ and $ \int \Psi dQ=0$. Now let $\Psi_c=c\Psi$. Choosing $c=-d\;\textrm{sign}(\int Z_T \Psi dQ)$ with $d>0$ large enough, it is clear that $\Psi_c$ satisfies the condition of Lemma~\ref{lem:Density}, and moreover:
	\begin{align*}
	\lim_{s \to 0} \frac 1 s \left(H(Q_s(\Phi_c))- H(Q)\right) & = c \int Z_T \Psi dQ -  \int Y_T dQ -C_Q(\Phi_c)  \\
	& \leq -d \vert\int  Z_T \Psi dQ \vert -   \int Y_T dQ +{c_Q}
	\end{align*}
	which is strictly negative for $d$ big enough. Hence, we find a contradiction with condition ~\eqref{eq:necessaryCondition}, so that $P$ almost surely we have the equality $Z_T=c_Q$ .
	But $\left(\der{Q}{P}\mid_{\mathcal{F}_t} \right)_{t\leq T}$ must be a $(\C,(\mathcal{F}_t)_{t\leq T}, \mathcal{F}_T, P)$ local martingale. Since $\left(\int \exp\Big\{\sum_{\alpha}\int_0^T G^{\alpha}_t + m^{\alpha}_Q(t)dW^{\alpha}_t - \frac 1 2 \int_0^T \Big(G^{\alpha}_t + m^{\alpha}_Q(t)\Big)^2 dt \Big\} d\gamma_Q\right)_{t\leq T}$ is a local martingale and $(Y_t)_{t\leq T}$ a process with finite variation, we have by uniqueness of semimartingale decomposition, we conclude that:
	\begin{equation*}
	\der{Q}{P} = \int \exp\left\{ \sum_{\alpha}\int_0^T G_t^{\alpha}+m^{\alpha}_Q(t) dW^{\alpha}_t - \frac 1 2 \int_0^T (G_t^{\alpha}+m^{\alpha}_Q(t))^2 dt\right\} d\gamma_Q.
	\end{equation*}

	\subsection{Existence, Uniqueness and precise characterization of the limit}\label{sec:ExistUniqueLim}
	As a consequence of the form of the density obtained from equation~\eqref{eq:densityMinimum} and the fact that the law of the uncoupled process $P$ is Gaussian, it is easy to show that when considering Gaussian initial conditions, any possible minimum of $H$ is a Gaussian process. Indeed, the characterization of the minima readily implies (as a simple application of Girsanov theorem) that the possible limits of the network equations are the law of the solutions of the implicit equation:
	\begin{equation}\label{eq:MFEwithU}
		d\bar{X}^{\alpha}_t = \left(-\frac{1}{\theta_{\alpha}}\bar{X}^{\alpha}_t + U^{\alpha,\bar{X}}_t\right)\,dt+\lambda_{\alpha}dW^{\alpha}_t
	\end{equation}
	where the processes $(W^{\alpha}_t)$ are independent Brownian motions and the processes $U^{\alpha,\bar{X}}_t$ are Gaussian processes with mean
	\[m^{\alpha}(t) = \sum_{\gamma=1}^M \bar{J}_{\alpha\gamma} \Exp[S_{\alpha\gamma}(\bar{X}^{\gamma}_{t-\tau_{\alpha\gamma}})]\]
	and covariance $C^{\alpha\gamma}(t,s)=0$ if $\alpha\neq \gamma$ and
	\[C^{\alpha\alpha}(t,s)=\sum_{\gamma=1}^M \sigma_{\alpha\gamma}^2 \Delta^{\alpha\gamma, \bar{X}}(t-\tau_{\alpha\gamma},s-\tau_{\alpha\gamma})\]
	where
	\[\Delta_{\alpha\gamma, \bar{X}}(t,s)=\Exp[S_{\alpha\gamma}(\bar{X}^{\gamma}_{t})S_{\alpha\gamma}(\bar{X}^{\gamma}_{s})]\]
	The relation given by~\eqref{eq:densityMinimum} provides a self-consistent equation on the moments of the Gaussian, which is given in the following:
	\begin{theorem}\label{pro:GaussianSolutions}
		Considering that the initial conditions are Gaussian, then possible minimum $Q=(Q^{\alpha},\alpha=1\cdots M)$ of the good rate function is Gaussian. Denoting by $\mu^{\alpha}(t)$ the mean of $Q^{\alpha}$ and by $C^{\alpha\beta}(t,s)=\Exp[X^{\alpha}_tX^{\beta}_s]$ their covariance, we have:
		\begin{equation}\label{eq:Means}
			\dot{\mu}^{\alpha}(t)=-\frac{1}{\theta_{\alpha}} \mu^{\alpha}(t) + \sum_{\gamma=1}^M \bar{J}_{\alpha\gamma}f_{\alpha\gamma}(\mu^{\gamma}(t-\tau_{\alpha\gamma}), C^{\alpha\alpha}(t-\tau_{\alpha\gamma}, t-\tau_{\alpha\gamma}))
		\end{equation}
		where $f_{\alpha\gamma}(\mu,v)=\int_{\R} S_{\alpha\gamma}(x) \frac{e^{-(x-\mu)^2/2v}}{\sqrt{2\pi v}}dx$. The covariance is equal to zero when $\beta\neq \alpha$ and:
		\begin{equation}\label{eq:Covariance}
			C^{\alpha}(t,s)=e^{-(t+s)/\theta_{\alpha}}\Big[C^{\alpha}(0,0)+ \frac{\ta \la^2}{2} (\exp{2(t\wedge s)/\ta} -1) + \sum_{\gamma=1}^M \sigma_{\alpha\gamma}^2\int_{0}^t\int_{0}^s e^{(u+v)/\theta_{\alpha}} \Delta_{\mu,C}^{\alpha\gamma}(u-\tau_{\alpha\gamma},v-\tau_{\alpha\gamma}) dudv\Big]
		\end{equation}
		where $\Delta_{\mu,C}^{\alpha\gamma}(u,v)=\mathbb{E}\Big[S_{\alpha\gamma}(X_{u}^{\gamma})S_{\alpha\gamma}(X_{v}^{\gamma})\Big]$ is a nonlinear function of $\mu^{\gamma}(u)$, $\mu^{\gamma}(v)$, $C^{\gamma\gamma}(u,v)$, $C^{\gamma\gamma}(u,u)$ and $C^{\gamma\gamma}(v,v)$.

		Moreover, there exists a unique solutions to these self consistent equations~\eqref{eq:Means} and~\eqref{eq:Covariance}.
	\end{theorem}

	The above theorem hence characterizes univocally the limits of the network equations considered. The proof of this proposition was done in~\cite{faugeras-touboul-etal:09b} starting from equations~\eqref{eq:MFEwithU} which were introduced using a heuristic argument.

	Note also that if the initial condition is not Gaussian, the solutions are not Gaussian. However, as time goes by, solutions get exponentially fast attracted to the Gaussian solutions described in theorem~\ref{pro:GaussianSolutions}. That description hence provides a handy procedure to analyze the solutions of the mean-field equations and their dynamics as a function of the parameters. In particular, we observe that the levels of heterogeneity, $(\sigma_{\alpha\gamma})$, appear as parameters of the equations. The moment equations provided above hence allow analyzing the qualitative effects of heterogeneity on the behavior of the network.

	If one is interested in the possible solutions of the mean-field equations for non-Gaussian initial conditions, or when the intrinsic dynamics of the system is not linear (i.e. $P$ is not a Gaussian measure), existence and uniqueness of solutions to the implicit equations~\eqref{eq:densityMinimum} also hold, but the demonstration of this property is more involved. The basic idea of the proof is based on a contraction principle, as usually done for proving existence and uniqueness of solutions. It is possible to show that one has a contraction in the Vaserstein distance. The steps of the proof are given in a different context in~\cite[Section 5.2]{ben-arous-guionnet:95}, and is not given here since we are not dealing with these more general cases.

	\subsection{Convergence of the process}

	We are now in a position to prove theorem~\ref{thm:Convergence}.

	\begin{proof}[Theorem~\ref{thm:Convergence}]
		Indeed, for $\delta$ a strictly positive real number and $B(Q,\delta)$ the open ball of radius $\delta$ centered in $Q$ for the Vaserstein distance. We prove that $Q^n(\hat{\mu}_n\notin B(Q,\delta))$ tends to zero as $n$ goes to infinity. Indeed, for $K_{\varepsilon}$ a compact defined in theorem~\ref{thm:tightness}, we have for any $\varepsilon>0$:
		\[Q^n(\hat{\mu}_n \notin B(Q,\delta)) \leq \varepsilon + Q^n(\hat{\mu}_n \in B(Q,\delta)^c \cap K_{\varepsilon}).\]
		The set $B(Q,\delta)^c \cap K_{\varepsilon}$ is a compact, and theorem~\ref{lemma3} now ensures that
		\[\limsup_{n\to\infty} \frac 1 n \log Q^n(\hat{\mu}_n\in B(Q,\delta)^c \cap K_{\varepsilon} )\leq -\inf_{B(Q,\delta)^c \cap K_{\varepsilon}} H\]
		and eventually, theorem~\ref{thm:Limit} ensures that the righthand side of the inequality is strictly negative, which implies that
		\[\lim_{n\to\infty}Q^n(\hat{\mu}_n \notin B(Q,\delta)) \leq \varepsilon,\]
		that is:
		\[\lim_{n\to\infty}Q^n(\hat{\mu}_n \notin B(Q,\delta)) =0.\]
	\end{proof}

	Based on this result, we can further conclude on the following:
	\begin{theorem}\label{thm:PropagationOfChaos}
		The system enjoys the propagation of chaos property. In other terms, $Q^n=\Exp_J[\Q^n(J)]$ is $Q$-chaotic, i.e. for any bounded continuous functions $(f_1,\cdots,f_m)$ and any neuron indexes $(k_1,\cdots, k_m)$, we have:
		\[\lim_{N\to\infty}\int \prod_{j=1}^m f_j(x^{k_j})dQ^N(x) = \prod_{j=1}^m \int f_j(x)dQ^{p(j)}(x)\]
	\end{theorem}
	This is a direct consequence theorem~\ref{thm:Convergence}, thanks to a result due to Alain-Sol Sznitman, see~\cite[Lemma 3.1]{sznitman:84}.

	All these results can be readily confirmed by numerical simulations of the network equations. Considering for instance a two-populations network with parameters given in section~\ref{sec:2popsOscill}, we simulated a network of $12\,000$ neurons ($6\,000$ in each population) and considered the distribution of the values of the membrane potentials as a statistical sample. The empirical distribution, superimposed with the theoretical Gaussian distribution, is plotted in figure Fig.~\ref{fig:Statistics} and shows a very clear fit, which we confirmed using the Kolmogorov-Smirnov test. For each population, the Kolmogorov-Smirnov test comparing the sample obtained by numerical simulations with the predicted Gaussian distribution ensures that the sample has indeed the Gaussian distribution, with a p-value equal to $1$. Moreover, we used a chi-square test of independence which validates the independence between the two populations and this independence test was validated with a p-value of $0.87$.
	\begin{figure}[htbp]
		\centering
			\includegraphics[width=.5\textwidth]{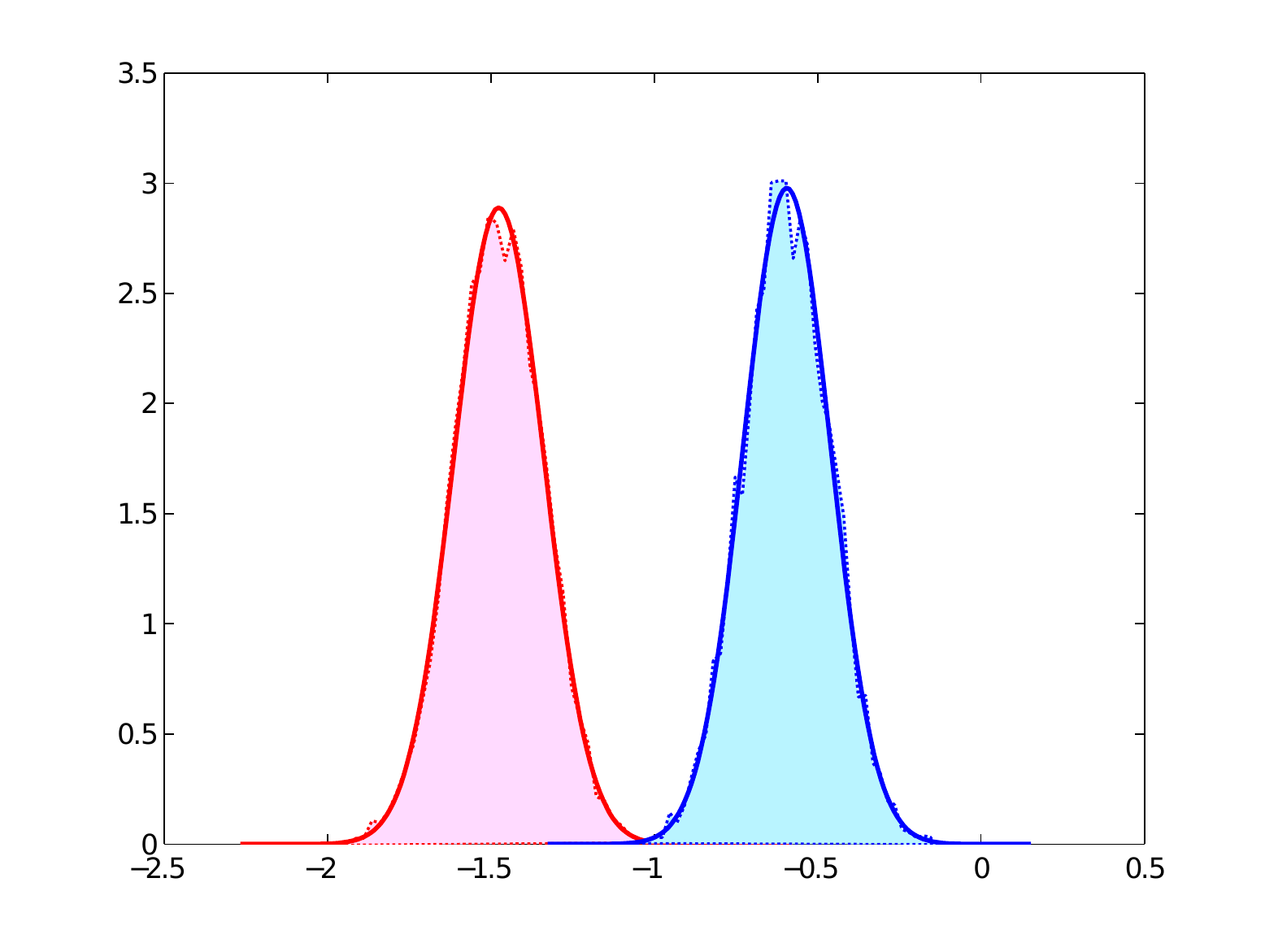}
		\caption{Empirical distribution (colored histogram with dotted lines) and theoretical Gaussian distributions for a 2 populations neuronal network (parameters given in section~\ref{sec:2popsOscill} with common heterogeneity parameter $\sigma=1$ and noise $\lambda=0.5$). }
		\label{fig:Statistics}
	\end{figure}

\section{Dynamics and phase transitions}\label{sec:Solutions}
In this section, we analyze the obtained limit equations to analyze the qualitative macroscopic behavior of networks, with a particular focus on the effect of the disorder parameters. The first step of this analysis consists in confronting our results with the seminal study of Sompolinsky, Crisanti and Sommers (SCS)~\cite{sompolinsky-crisanti-etal:88} dealing with one population model with centered coefficients, centered sigmoidal functions $S$ and no delays. We will then discuss the persistence of the phase transition they identified when the assumptions on the parameters (non-centered synaptic weights or sigmoids, delays, multiple populations\ldots) are relaxed. A particularly important phenomenon in neuroscience essentially absent of their initial study, synchronized oscillations, will be an important focus of the present section.

Let us eventually note that the results of SCS hold for zero-noise limits ($\lambda_{\alpha}=0$). We will also sometimes consider this limit equations, even if rigorously, the proofs of section~\ref{sec:BenArous} hold for non-trivial noise. These non-noisy regimes correspond to limits of the mean-field equations where $\lambda_{\alpha}\to 0$, and correspond to a sort of viscosity solution of the system: all the properties of convergence, existence and uniqueness of solution hold for arbitrarily small $\lambda_{\alpha}$ and provide solutions that have a limit when $\lambda_{\alpha}\to 0$.

\subsection{The generalized Somplinsky-Crisanti-Sommers Equations}

In their article, Sompolinsky, Crisanti and Sommers (SCS) introduce a set of equations governing the dynamics of covariance of possible stationary solutions to the mean-field equations. These equations are used to analyze the dynamics of the limit process and in particular to show a striking transition between stationary and chaotic solutions. We derive here a generalized equation of the type of the SCS equations in our framework with multiple populations and delays, and use these equations to explore the boundaries of the SCS phase transition when considering different models.

\begin{proposition}\label{pro:StationaryEquations}
	Possible stationary solutions are Gaussian with mean $\bar{\mu}^{\alpha}$ and covariance $\bar{C}^{\alpha}(\tau)=C^{\alpha\alpha}(t+\tau,t)$ for any $t\geq 0$. These two variables satisfy the system of equations:
	\begin{equation}\label{eq:StationMoments}
		\begin{cases}
			\displaystyle{0=-\frac{1}{\theta_{\alpha}} \bar{\mu}^{\alpha} + \sum_{\gamma=1}^M \bar{J}_{\alpha\gamma}f_{\alpha\gamma}(\bar{\mu}^{\gamma}, \bar{C}^{\alpha}(0)})\\
			\displaystyle{\ddot{\bar{C}}^{\alpha}(\zeta)=\frac{\bar{C}^{\alpha}(\zeta)}{\ta^2} + \sum_{\gamma=1}^M \bar{\Delta}^{\alpha\gamma}_{\bar{\mu},\bar{C}}(\zeta)}
		\end{cases}
	\end{equation}
\end{proposition}

\begin{remark}
	Note that the above equations do not constitute a dynamical system, but rather correspond to implicit equations. In particular, an important difficulty is the choice of the initial condition $C^{\alpha}(0)$ which corresponds to the variance of the stationary solution, which is obviously unknown. This quantity parametrizes both the equation on the first moment and the form of the term $\Delta^{\alpha\gamma}$ on the second moment equations.
\end{remark}

\begin{proof}
	The equation on the mean $\bar{\mu}^{\alpha}$ is a simple rewriting of equation~\eqref{eq:MeansSummary} under stationarity condition. The equation on the stationary covariance requires more care. For arbitrary time $t$, denoting $X^{\alpha}_t$ the solution of the mean-field equation with for all $\alpha$, $\lambda_{\alpha}=0$, we have, using equation~\eqref{eq:MFEwithUSummary}:
	\begin{align}
		\nonumber \dot{C}^{\alpha}(\zeta)&=\der{}{\zeta} \Exp{[(\bar{X}^{\alpha}(t+\zeta)-\mu^{\alpha}(t+\zeta))(\bar{X}^{\alpha}(t)-\mu^{\alpha}(t))]}\\
		\label{eq:SompoFirstStep} &=-\frac{C^{\alpha}(\zeta)}{\theta_{\alpha}} + \Exp{[\bar{X}^{\alpha}(t) {U}^{\alpha, \bar{X}}(t+\zeta)]}.
	\end{align}
	The second term is not easy to characterize. The method used by Sompolinsky and collaborators to deal with this term is to derive a second time with respect to $\zeta$. However, the differential of $U^{\alpha}$ is unknown. Fortunately, we can express this term as a function of $\delta^{\alpha}(\zeta)=\Exp[\bar{X}^{\alpha}(t+\zeta) {U}^{\alpha}(t)]$. This function is way easier to handle since using the differential equation~\eqref{eq:MFEwithU} and differentiating this expression with respect to $\xi$, one obtains:
	\[\dot{\delta}^{\alpha}(\zeta)=-\frac{\delta^{\alpha}}{\theta_{\alpha}} + \sum_{\beta=1}^M\sigma_{\alpha\beta}^2 \Delta^{\alpha\beta}_{\bar{\mu},\bar{C}} (\zeta)\]
where we denoted with a slight abuse of notations $\Delta^{\alpha\beta}_{\bar{\mu},\bar{C}} (\zeta)$ the common value of $\Delta^{\alpha\beta}_{\bar{\mu},\bar{C}}(t+\zeta,t)$ for any $t>0$ using the assumed stationarity of the solution.

	In order to relate the second term of the righthand side of~\eqref{eq:SompoFirstStep} with $\delta^{\alpha}$, we compute $\dot{C}(\zeta+\xi)$ expressing it the differential with respect to $\xi$ of $\Exp[{\bar{X}^{\alpha}(t+\zeta+\xi/2)\bar{X}^{\alpha}(t-\xi/2)}]$. In this computation, most of the terms cancel out and we obtain the simple expression at $\xi=0$:
	\[2 \dot{C}^{\alpha}(\zeta)= \Exp{[\bar{X}^{\alpha}(t) U^{\alpha,\bar{X}}(t+\zeta)]} - \delta^{\alpha}(\zeta).\]
	Plugging this expression into~\eqref{eq:SompoFirstStep} we obtain:
	\[\dot{C}^{\alpha}(\zeta)=\frac{C^{\alpha}(\zeta)}{\ta} - \delta^{\alpha} (\zeta).\]
	Differentiating this expression with respect to $\zeta$ and reinjecting the latter equation in the obtained expression, we get:
	\begin{align}
		\nonumber \ddot{C}^{\alpha}(\zeta)&=\frac{\dot{C}^{\alpha}(\zeta)}{\ta} - \left(-\frac{\delta^{\alpha}}{\ta} + \sum_{\beta=1}^M \sigma_{\alpha\beta}^2 \Delta^{\alpha\beta}(\zeta)\right)\\
		\label{eq:Sompo} &= \frac{C^{\alpha}(\zeta)}{\ta^2} - \sum_{\beta=1}^M \sigma_{\alpha\beta}^2 \Delta^{\alpha\beta}(\zeta).
	\end{align}
\end{proof}
	This equation is very similar to the original SCS equation. As they remarked, this equation does not characterize the process. Indeed, we know that $\dot{C}^{\alpha}(0)=0$ using the fact that the covariance is even, but the initial condition $C^{\alpha}(0)$ is not fixed: it is the asymptotic stationary variance of the process, when it exists, and this initial condition is a parameter of both the stationary mean equation and stationary covariance equation. Let us emphasize also the that it does not involve the average of the synaptic coefficients $\bar{J}_{\alpha\beta}$, and that the function $\Delta^{\alpha\beta}(\zeta)$ is a nonlinear function of $C^{\beta}(\zeta)$ and $C^{\beta}(0)$ which can be written as:
	\begin{align*}
		\Delta^{\alpha\beta}(\zeta) &= \Exp{[S_{\alpha\beta}(V^{\beta}(\zeta))S_{\alpha\beta}(V^{\beta}(0))]}\\
		&=\int_{\R^2} S_{\alpha\beta}\left(\sqrt{\frac{\bar{C}^{\beta}(0)^2-\bar{C}^{\beta}(\zeta)^2}{\bar{C}^{\beta}(0)}} x + \frac{\bar{C}^{\beta}(\zeta)}{\sqrt{\bar{C}^{\beta}(0)}} y + \bar{\mu}^{\beta}\right) \quad S_{\alpha\beta}\left(\sqrt{\bar{C}^{\beta}(\zeta)}y + \bar{\mu}^{\beta}\right) Dx\,Dy
	\end{align*}
	where $Dx$ and $Dy$ are the probability measure of standard Gaussian random variables ($Dx=e^{-x^2/2}/\sqrt{2\pi} dx$ and similarly for $y$). The linear term has a positive sign, hallmark of diverging systems. And indeed, the solution can be formally written as:
	\begin{equation*}
		\bar{C}^{\alpha}(\zeta)=\bar{C}^{\alpha}(0)\cosh(\frac{\zeta}{\tau_{\alpha}}) -\sum_{\beta=1}^M\frac{\sigma_{\alpha\beta}^2 \tau_{\alpha}}{2} \Big(\int_0^\zeta e^{\frac{\xi-\zeta}{\tau_{\alpha}}}\Delta_{\beta}(\xi)d\xi
		-\int_0^\zeta e^{-\frac{\xi-\zeta}{\tau_{\alpha}}}\Delta_{\beta}(\xi)d\xi\Big).
	\end{equation*}
	In that equation, the term in hyperbolic cosine diverges very fast, and the nonlinear term can overcome the divergence. When one does not consider the precise initial condition corresponding to stationary solutions, the solutions of the second order ODE in $\bar{C}^{\alpha}$ diverge very fast, as remarked by Sompolinsky and colleagues in~\cite{sompolinsky-crisanti-etal:88}, and therefore one needs to analyze the stability of the possible solutions, which is relatively complex to perform. Numerical simulations of the system are very intricate as well, because of the time-consuming calculation of of the nonlinear function $\Delta$ and because of the possible divergence of the solutions.
	
	However, Sompolinsky and collaborators show very elegantly an important phase transition taking place in this system, analyzing the shape of the potential together with a stability analysis of the solutions. We revisit their results in our more general framework, first in one population systems, and then in higher dimensional systems, and particularly focus on the effects of delays, non-zero mean connectivity and non-centered sigmoids.


\subsection{One population networks}\label{sec:SCS1Pop}
	The heterogeneity level appears as a parameter in~\eqref{eq:Sompo}. In their one-population setting, equation~\eqref{eq:Sompo} can be written the equation of the position of a particle submitted to a force deriving from a potential $\Phi_1$ (the label $1$ denotes the number of populations) which is equal to $-\frac 1 2 C ^2 + \sigma^2\psi$ where $\psi$ is a primitive of $\Delta$ considered as a function of $C$. The shape of the potential showing a transition from convex to double-well as $\sigma$ is increased allowed the authors to conclude on a phase transition between a stationary solution where~\eqref{eq:Sompo} has a unique equilibrium equal to zero and a chaotic regime where the covariance is non-zero. The value of the noise at this transition corresponds to $\sigma\,S'(0)=1/\tau$.

\subsubsection{Non-delayed networks with non-centered synapses}\label{sec:OnePopNonCentered}
	Let us start by a one-population network with no delays and $\bar{J}\neq 0$. In that case, it is easy to see that fixed point with mean $\mu=0$ is stable if and only if $J\,\derpart{f}{\mu}(0,C(0))<1$. When letting $\sigma$ fixed and increasing $J$, we can see that at $J=\derpart{f}{\mu}(0,C(0))^{-1}$, the system undergoes a pitchfork bifurcation, and two new equilibria $\mu^+>0$ and $\mu^-<0$ appear, which are stable. For these equilibria, the null covariance is no more a solution to the equations, and we observe a stationary behavior of neurons with a non-zero standard deviation, i.e. a dispersion of the individual trajectories, that remain stationary. At these points, the system undergoes also a phase transition from stationary to chaotic activity when the heterogeneity coefficient crosses the critical value $\sigma=1/\derpart{f}{\mu}(\mu^{\pm},C(0))$. This equation is an implicit equation since the righthand side depends on $\sigma$ through the stationary variance of the process. Let us denote by $\Gamma(\sigma)$ the stationary standard deviation $C(0)$. It is clear that $\Gamma(\sigma)$ is an increasing function of $\sigma$, and to fix ideas, let us consider that the sigmoid used is an erf function $S(x)=\erf(g\,x)=\int_{0}^{gx} e^{-x^2/2}$. Then it is easy to show using a change of variables (see~\cite{touboul-hermann-faugeras:11}) that
	\[f(x,\sigma)=\erf(\frac{g\,x}{\sqrt{1+g^2 \Gamma(\sigma)}})\]
	and therefore the pitchfork bifurcation arises along the parameter curve $J=\sqrt{\frac{1}{g^2}+\Gamma(\sigma)}$, and the phase transition from stationary to chaotic behavior along the curve $-1+\sigma S'(\mu^{\pm})=0$. Since the differential of $S$ takes its maximum at $0$ and decreases to zero at $\pm \infty$, the value of $\sigma$ corresponding to the secondary phase transition to chaos is an increasing function of $\sigma$. Moreover, in that case, the chaotic activity will be no more centered around zero but around the new fixed point $\mu^{\pm}$. Eventually, it is interesting to note that for $J$ smaller than the value corresponding to the pitchfork bifurcation, the stationary covariance $\Gamma(\sigma)$ precisely equal to zero. A hand-drawn bifurcation diagram reflecting this behavior, together with simulations of the trajectories, is plotted in figure Fig.~\ref{fig:OnePopulationJ}.
	\begin{figure}[!h]
		\centering
			\includegraphics[width=.7\textwidth]{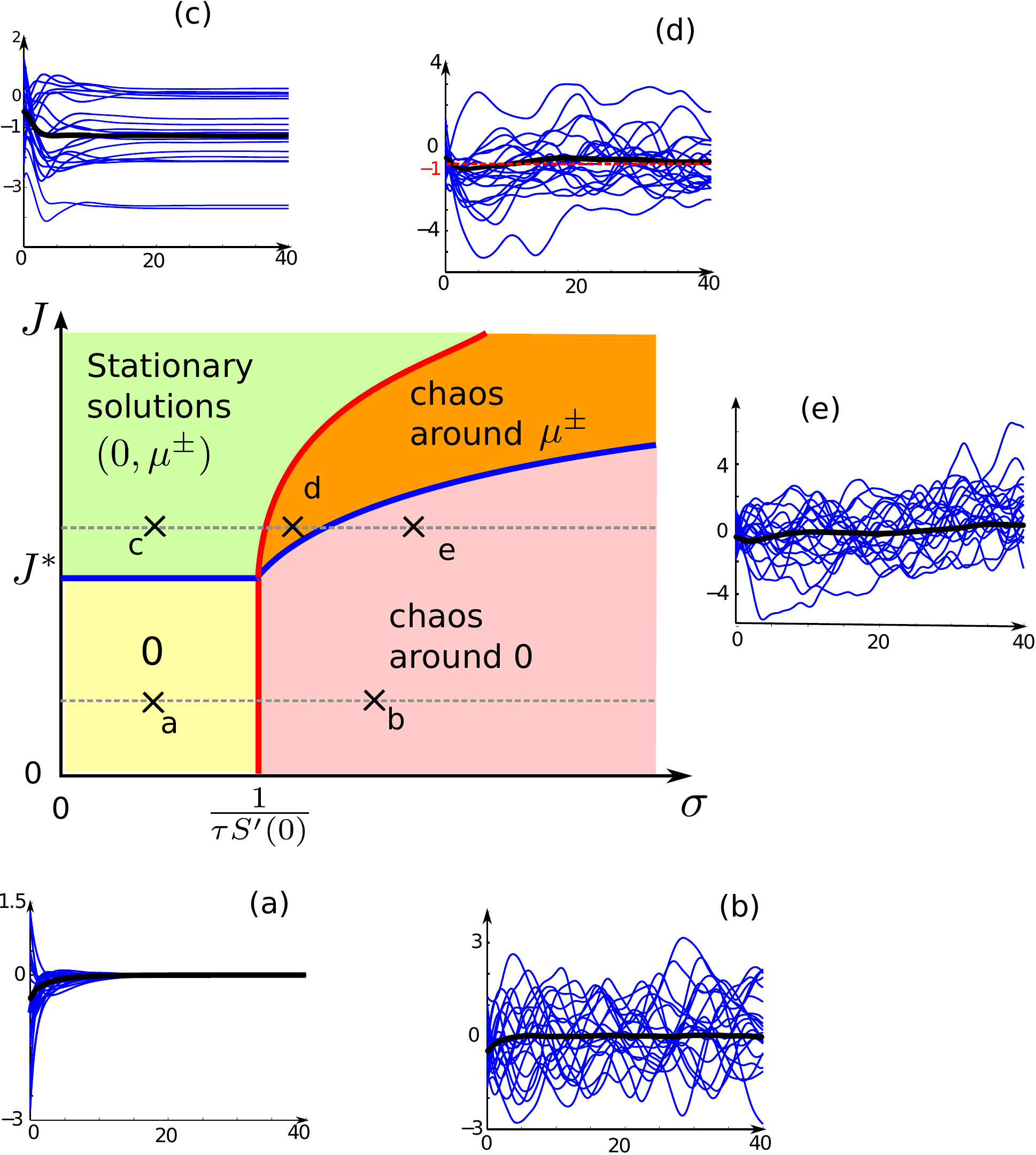}
		\caption{Behavior of a one-population system with non-centered synaptic coefficients. Center: Bifurcation diagram (hand-drawn) segmented into four regions: two regions of stationary behavior (yellow: centered at zero and green: centered on $\mu^{\pm}$) and two chaotic regions (pink: centered around zero and orange: centered around $\mu^{\pm}$). The boundaries of these regions are: a pitchfork bifurcation (blue curve) separating the stationary or chaotic regions centered on $0$ to the ones centered on $\mu^{\pm}$, and a generalized SCS phase transition (red curve) separating the stationary and chaotic regimes. The subfigures (a)-(e) show the time course of $30$ arbitrarily chosen neurons in the network corresponding to the points a-e of the diagram: $\tau=1$, $S'(0)=1$, (a): $J=0.5$, $\sigma=0.5$, (b): $J=0.5$, $\sigma=1.5$, (c): $J=1.5$, $\sigma=1.5$, (d): $J=1.5$, $\sigma=1.7$, (e): $J=1.5$, $\sigma=2$.}
		\label{fig:OnePopulationJ}
	\end{figure}
	%

\subsubsection{Delay-induced oscillations}
We now consider a one-population network with delays. Without loss of generality, we consider that the time constant is equal to $1$. The solutions of the mean-field equations with no heterogeneity are Gaussian processes whose moments reduce to a dynamical system:
\[
\begin{cases}
	\dot{\mu}&=-\mu + J f(\mu(t-\tau),v(t-\tau))\\
	\dot{v}&=-2\,v + \lambda^2
\end{cases}
\]
and hence the variance converges towards $\lambda^2/2$. To fix ideas, we consider $S(x)=\erf(gx)$, so that $f(x,\sigma)=\erf(\frac{g\,x}{\sqrt{1+g^2 \Gamma(\sigma)}})$. Since for any $v$, $f(0,v)=0$, the null mean is a stationary solution of the equation. Its stability depends on the roots of the characteristic equation (or dispersion relationship):
\[\xi = -1+J\left .\derpart{f}{\mu}\right \vert_{0,\frac{\lambda^2}{2}} e^{-\xi \tau} = -1+J\frac{g}{\sqrt{1+g^2\frac{\lambda^2}{2}}} e^{-\xi \tau}.\]
If all characteristic roots have negative real part, the fixed point $\mu=0$ is stable. As a function of the parameters of the system, characteristic roots can cross the imaginary axis and yield a destabilization of the fixed point. Turing-Hopf instabilities arise when there exists purely imaginary characteristic roots $\xi=\mathbf{i}\omega$. In that case, we obtain the following equivalent system:
\[
\begin{cases}
	-1+J\frac{g}{\sqrt{1+g^2\frac{\lambda^2}{2}}} \cos(\omega \tau)= 0\\
	\omega = -J\frac{g}{\sqrt{1+g^2\frac{\lambda^2}{2}}} \sin(\omega\tau)
\end{cases}
\]
which has real solutions only for $J\frac{g}{\sqrt{1+g^2\frac{\lambda^2}{2}}} >1$. It is then easy to show that Turing-Hopf bifurcations arise when the parameters satisfy the relationship:
\[
	\tau = \frac{\arccos\left(\frac {\sqrt{1+g^2\frac{\lambda^2}{2}}}{Jg}\right)}{\sqrt{\frac{J^2g^2}{1+g^2\frac{\lambda^2}{2}} - 1}}
\]
and these correspond to characteristic roots $	\omega =\sqrt{\frac{J^2g^2}{1+g^2\frac{\lambda^2}{2}} - 1}$. These result into oscillations of the solutions at a pulsation equal by $\omega$.

Let us now return to the case of random coefficients with variance $\sigma$ and no additive noise $\lambda=0$. The mean of the Gaussian solution satisfy the same equation as the one studied above with $\lambda=\Gamma(\sigma)$, and as noted in the previous section, the stationary covariance is an increasing function of $\sigma$. For $J>\derpart{f}{\mu}(0,0)$, the fixed point $0$ is unstable, and the covariance is non-zero. This implies that for sufficiently large values of the delay, the network displays oscillations. Thanks to the propagation of chaos property, all neurons have the same distribution, which is a Gaussian with oscillatory mean, and hence the network displays phase-locked oscillations. Eventually, as noise is increased beyond a critical value, a SCS phase transition occurs and the system no more displays phase locked oscillations but asynchronous chaotic activity. This is illustrated in figure Fig.~\ref{fig:DelayOscillations}

\begin{figure}
	\centering
		\subfigure[Turing-Hopf bifurcation curve]{\includegraphics[width=.3\textwidth]{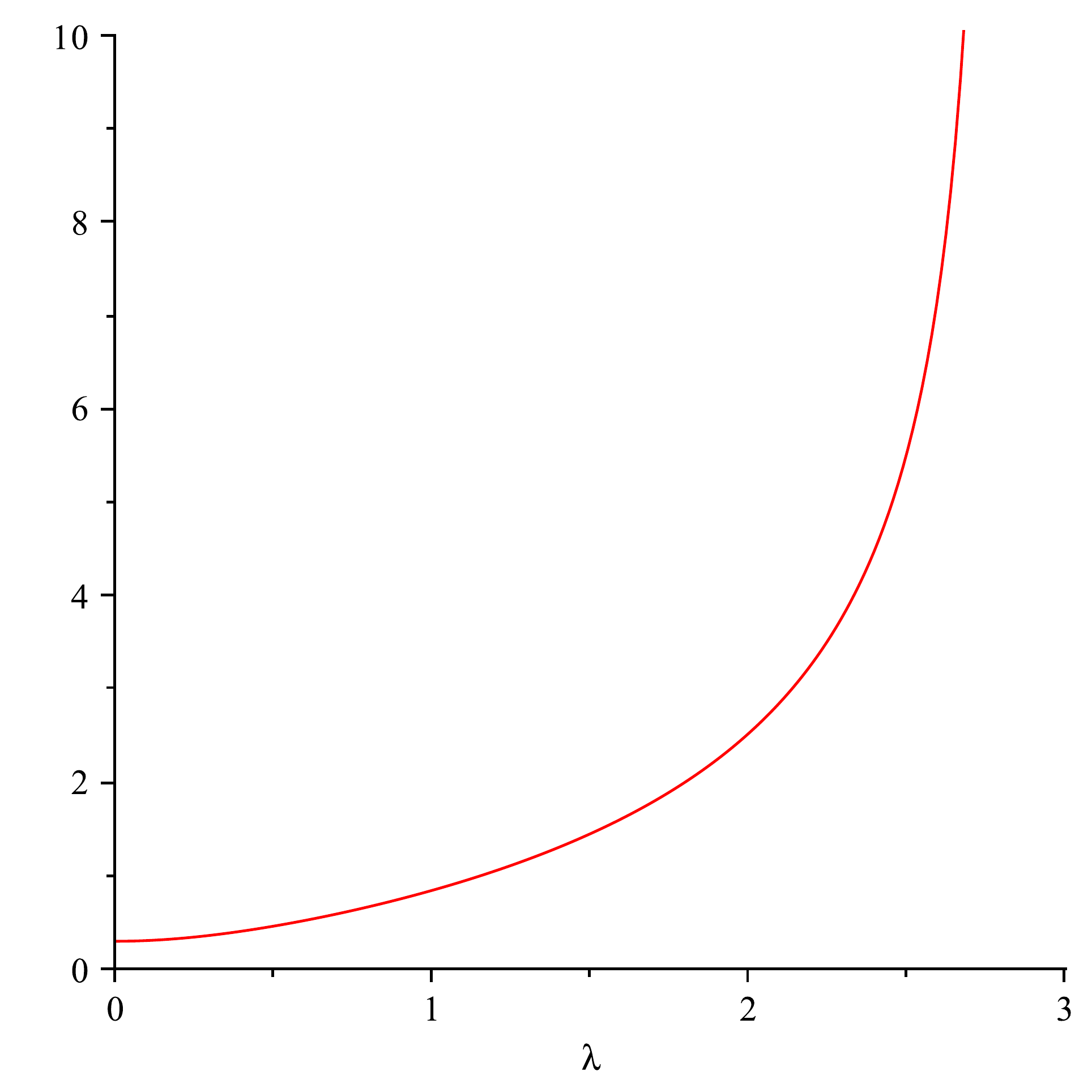}}\qquad
		\subfigure[$\tau=0.1,\; \sigma=0.5$]{\includegraphics[width=.4\textwidth]{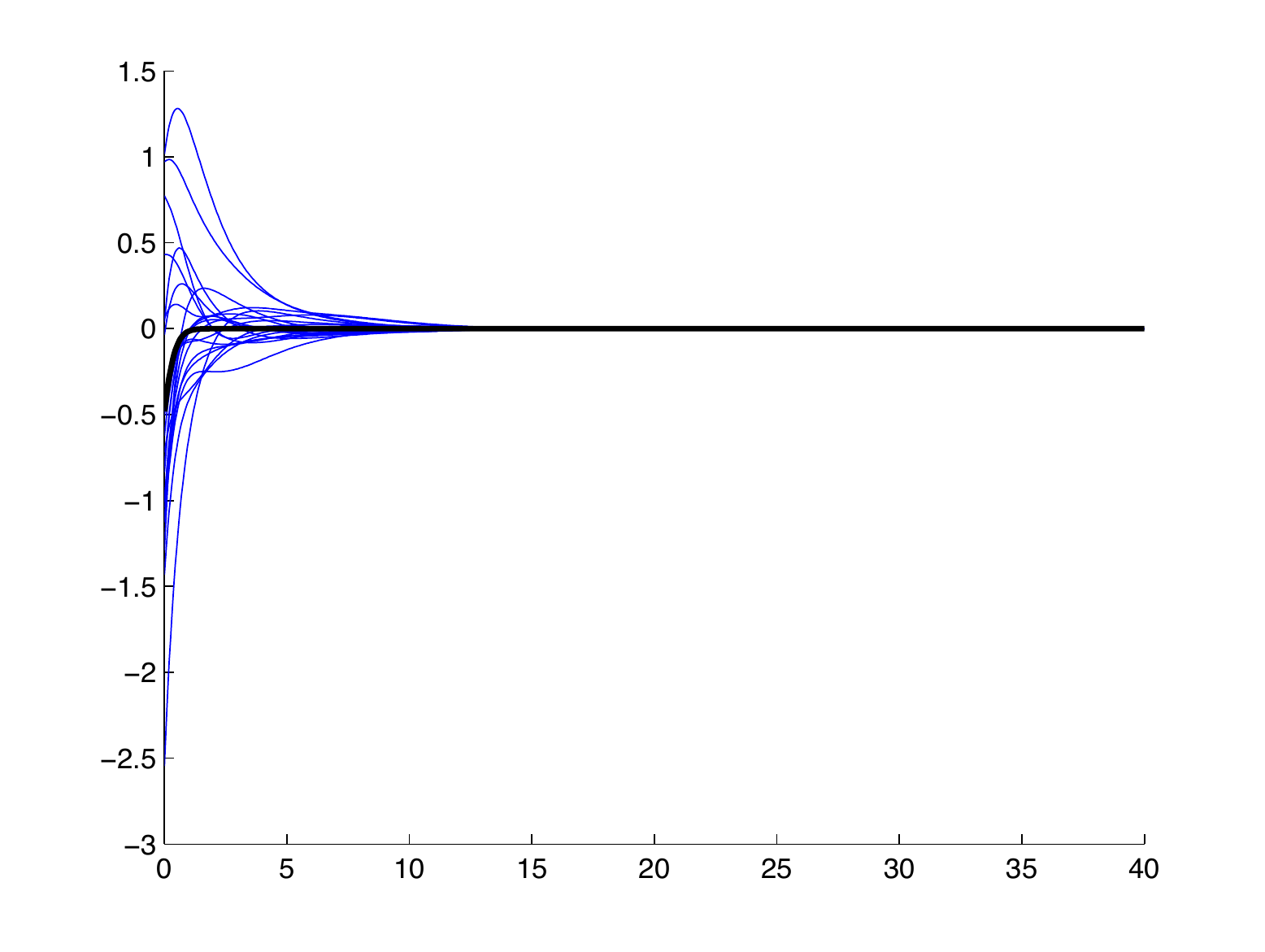}}
		\subfigure[$\tau=0.5,\; \sigma=0.5$]{\includegraphics[width=.4\textwidth]{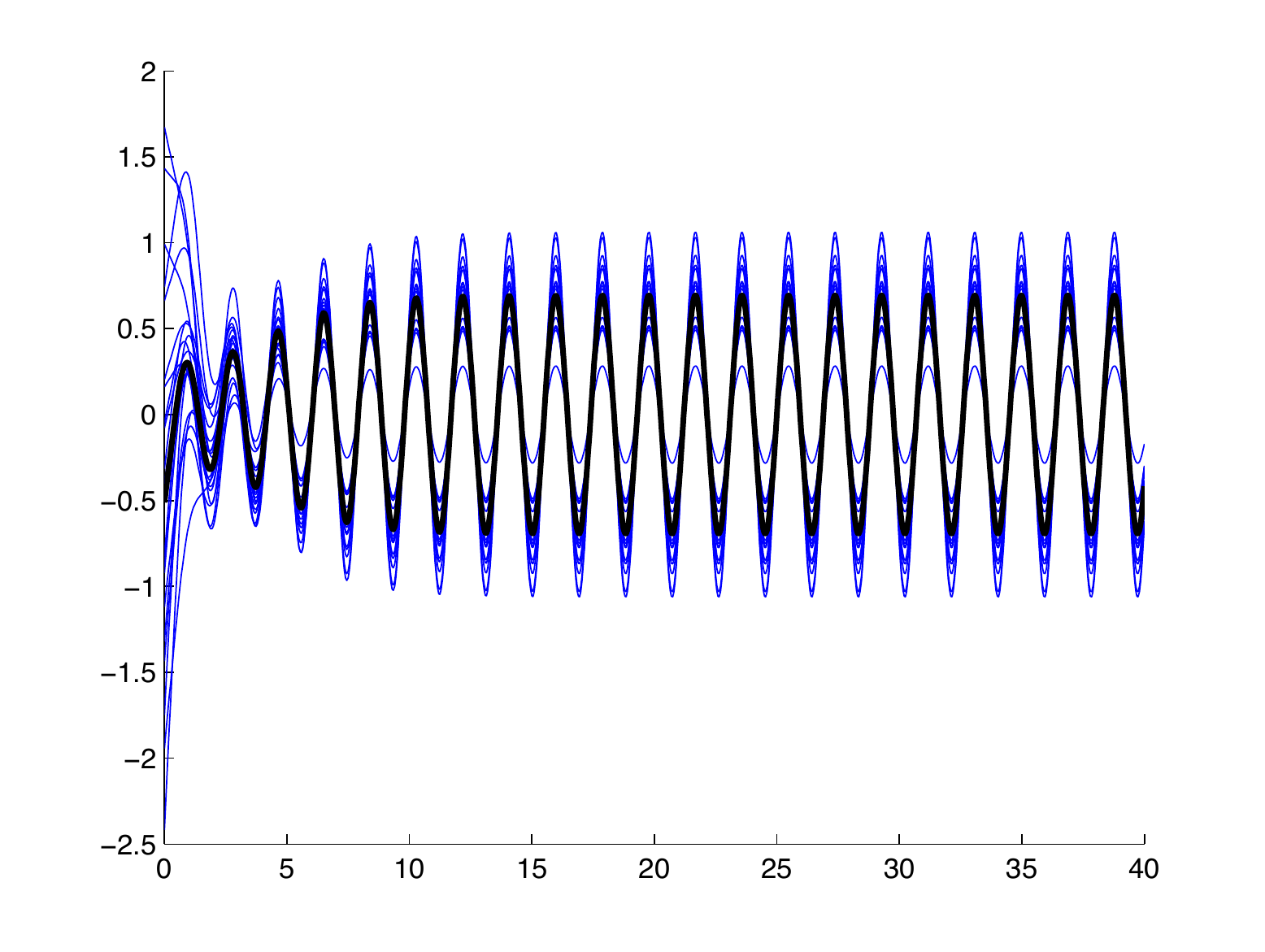}}
		\subfigure[$\tau=0.5,\; \sigma=1$]{\includegraphics[width=.4\textwidth]{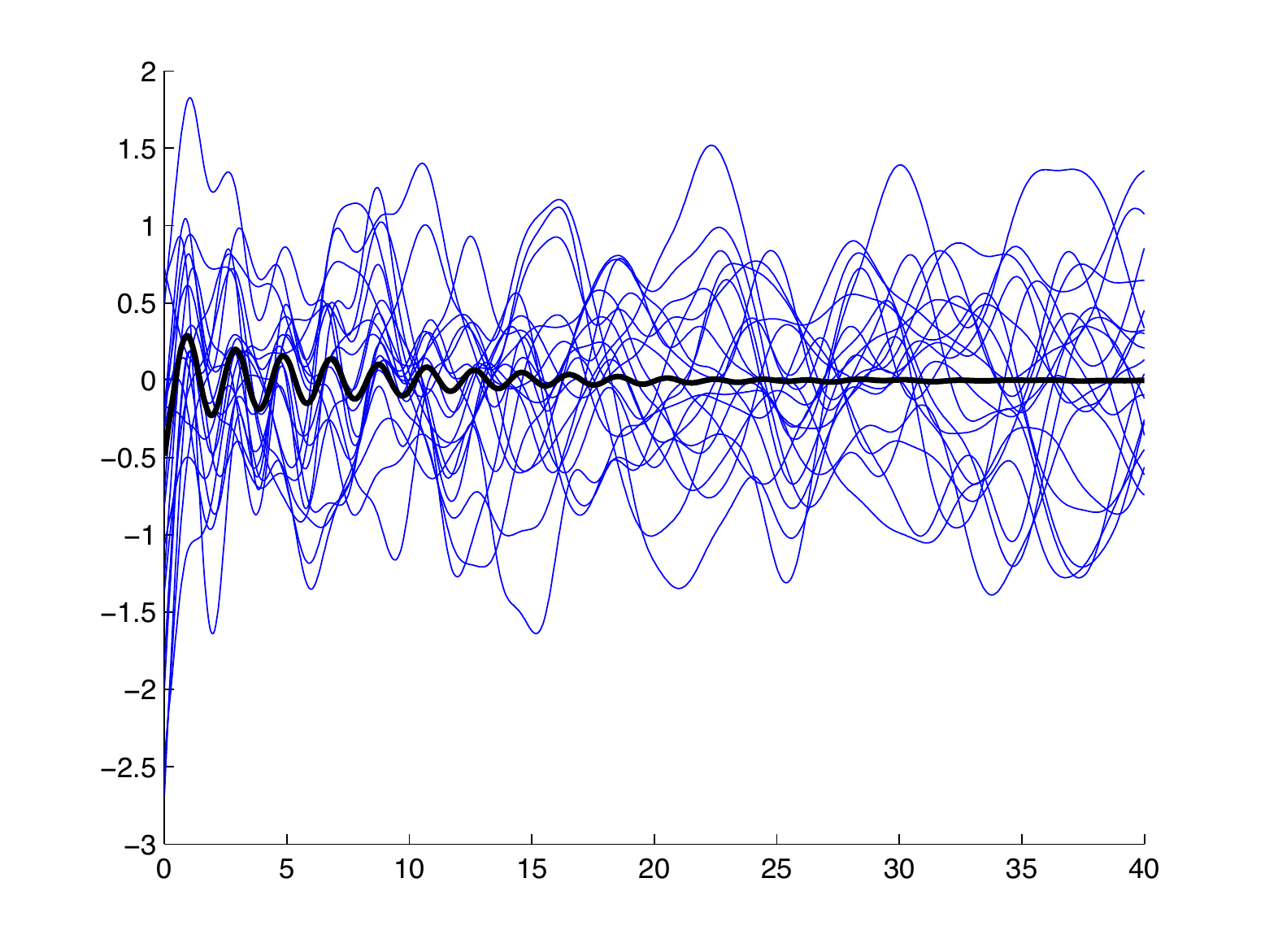}}
	\caption{One population delayed system, $\tau=1$, $S(x)=\erf(gx)$, $g=3$ and $J=-2$. (a) represents the curve of Turing-Hopf bifurcations in the plane $(\tau,\lambda)$ obtained analytically. (b): $\tau$ small: no oscillations. As the delays are increased, a Turing Hopf bifurcation occurs and oscillations appear (b), which disappear when the heterogeneity $\sigma$ is increased beyond a critical value in favor a chaotic activity (c). }
	\label{fig:DelayOscillations}
\end{figure}

\subsection{Multi-population networks}
In this section, we analyze the dynamics of randomly coupled neuronal networks in the case of the deterministic coupling of several original SCS networks, before turning the analysis of the dynamics of a more biologically plausible neuronal network composed of an excitatory and an inhibitory population. As demonstrated by Sompolinsky and coworkers in~\cite{sompolinsky-crisanti-etal:88}, the study of the stationary states using equations~\eqref{eq:StationMoments} is very useful to analyze the dynamics of their networks. Unfortunately, this method does not persists in higher dimensions, since the equation does not necessarily derive from a potential. Indeed, in order for the equation to derive from a potential $\Phi_M:\R^M\mapsto \R$ in dimension $M$ greater than $1$, we need that for any $\alpha\in\{1,\ldots,M\}$:
\[\derpart{\Phi_M}{C^{\alpha}}=-C^{\alpha}+\sum_{\beta=1}^M \sigma_{\alpha\beta}^2 \Delta^{\beta}\]
The only case where this is possible is the case where $\sigma_{\alpha\beta}^2=0$ for any $\alpha\neq \beta$. Indeed, shall the above relationship be true, the equality $\derpartsxy{\Phi_M}{C^{\alpha}}{C^{\beta}}=\derpartsxy{\Phi_M}{C^{\beta}}{C^{\alpha}}$ directly yields $\sigma_{\alpha\beta}^2\derpart{\Delta^{\beta}}{C^{\beta}}=\sigma_{\beta\alpha}^2\derpart{\Delta^{\alpha}}{C^{\alpha}}$. The lefthand side is a function of $C^{\alpha}$ only and the righthand side a function of $C^{\beta}$ only, they for $\alpha\neq\beta$ these functions are necessarily constant. For regular functions $S$, this necessitates to have $\sigma_{\alpha\beta}=\sigma_{\beta\alpha}=0$. This is precisely the case of deterministic lateral connections between randomly coupled networks, which will now study.

\subsubsection{Deterministic lateral coupling of SCS networks}\label{sec:SCSPPops}
In this section we analyze the coupling of different SCS networks, called lateral coupling, with deterministic coefficients. The only randomness in the models is included in the random synaptic coefficients between neurons belonging to the same population. In that particular case, equation~\eqref{eq:Sompo} derives from the potential $\Phi_M(C^1\cdots C^M)=\sum_{\alpha=1}^M \Phi_1(C^{\alpha})$, and in that case the analysis driven by Sompolinsky and collaborators can be adapted to the multi-dimensional case. Since the potential is now the sum of the individual potentials at each population, we observe a strange phenomenon of localization of chaos in the populations that display a large heterogeneity (namely, in our notations, when the SCS condition $\sigma_{\alpha\alpha}^2 S'_{\alpha\alpha}(0) \tau_{\alpha}>1$ is satisfied). Only the populations that individually would be in a chaotic state are in a chaotic state, and the other populations converge to zero with a Dirac delta covariance at zero, and the input received by such populations from chaotic populations do not perturb this state. Let us for instance illustrate this phenomenon on a two-populations network with parameters:
	\[\bar{J}=\left(\begin{array}{cc}
		0 & J_{21}\\
		J_{12}& 0
	\end{array}\right) \qquad \text{and} \qquad \sigma = 	\left(\begin{array}{cc}
			\sigma_1 & 0\\
			0 & \sigma_2
		\end{array}\right)\]
		Each population receives input from the neurons of the other population, with a constant synaptic weight equal to $J_{\alpha\beta}$, and the intra-population synaptic weights are noisy.
	
	
	\begin{figure}[h]
		\centering
			\subfigure[Shape of the potential]{\includegraphics[width=.5\textwidth]{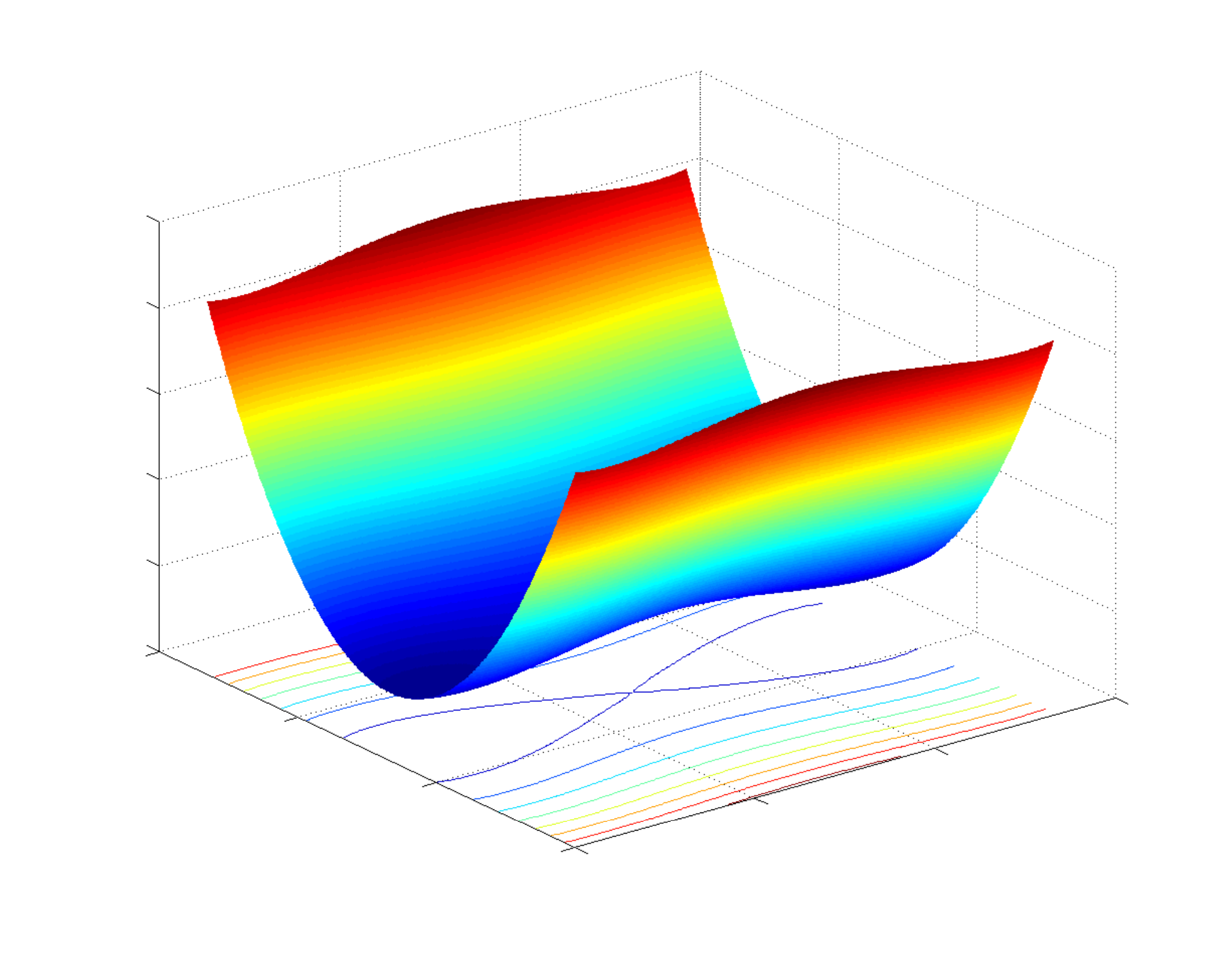}}
			\subfigure[Trajectories]{\includegraphics[width=.4\textwidth]{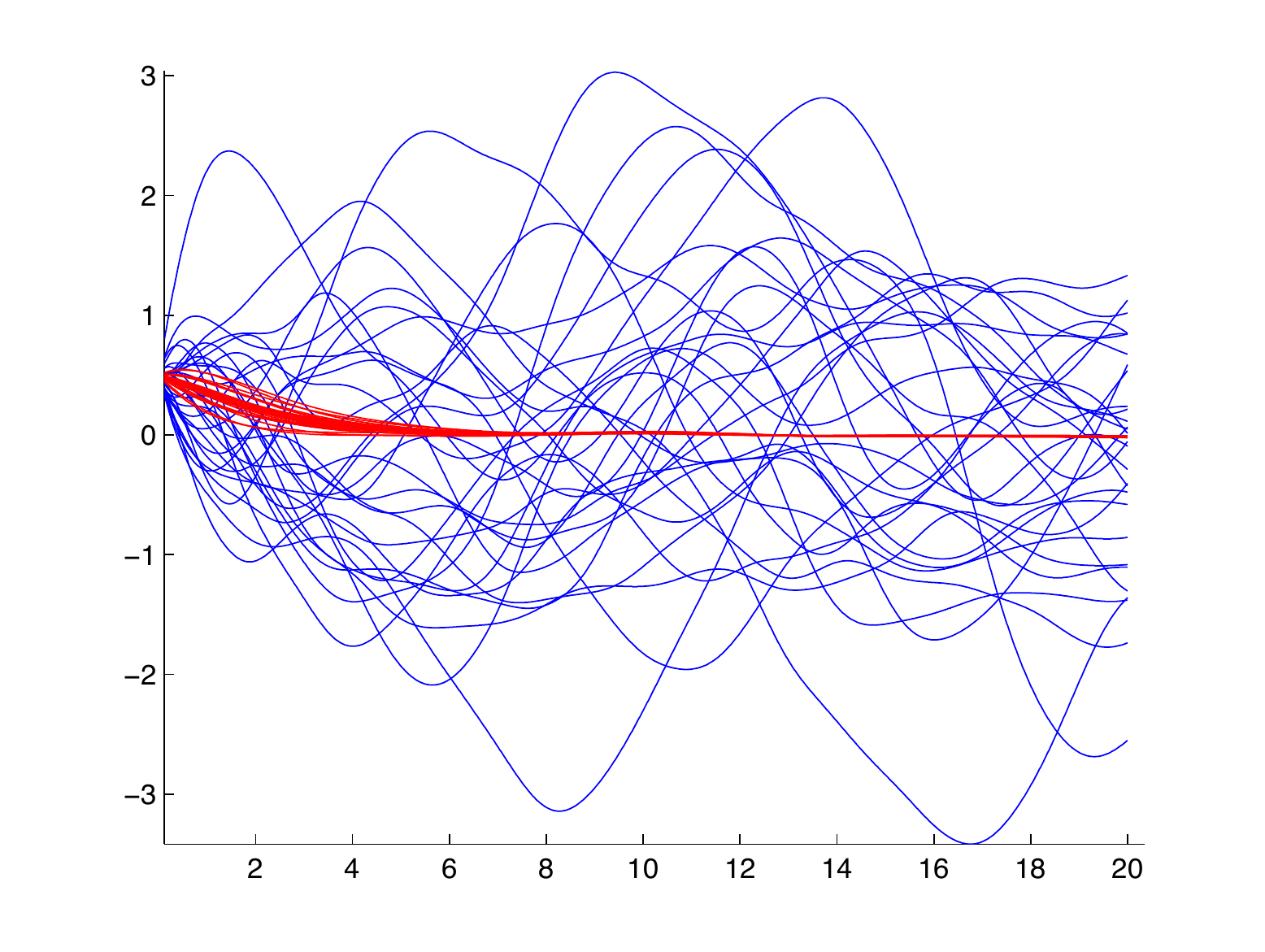}}
		\caption{Dynamics for a 2-populations network with Hamiltonian dynamics (no variance on the cross-population synaptic weights): $\tau_1=\tau_2=1$, $\sigma_1=3>1$, $\sigma_2=0.5<1$, $J_{12}=J21=3$. The potential shows a double-well shape, corresponding to a chaotic state on population 1 and a stationary state on population 2. Simulation of a $4\,000$ neurons network illustrate this phenomenon (right): blue (resp. red): $30$ arbitrarily chosen trajectories population from 1 (resp. 2). }
		\label{fig:label}
	\end{figure}
	Further analysis of this networks as a function of the coupling reveals a similar phenomenon as the one described in the one-population network of section~\ref{sec:OnePopNonCentered}. Indeed, as the strength of the lateral coupling $J_{12}$ and $J_{21}$ are increased, additional stationary solutions with non-zero covariance appear. Let us for instance denote by $\mu^{*}$ the mean of one of these stationary solutions. Following SCS analysis, we are ensured that the behavior of the trajectories of neurons in population $\alpha$ around $\mu^*_{\alpha}$ is stationary as long as $\sigma_{\alpha}<\frac{1}{\tau_{\alpha}S'(\mu^{*}_{\alpha})}$ and chaotic otherwise, and this independently of the behavior of the other population. This phenomenon is illustrated in figure Fig.~\ref{fig:LocChaosNoZero}.

\begin{figure}[h]
	\centering
		\subfigure[Stationary States]{\includegraphics[width=.4\textwidth]{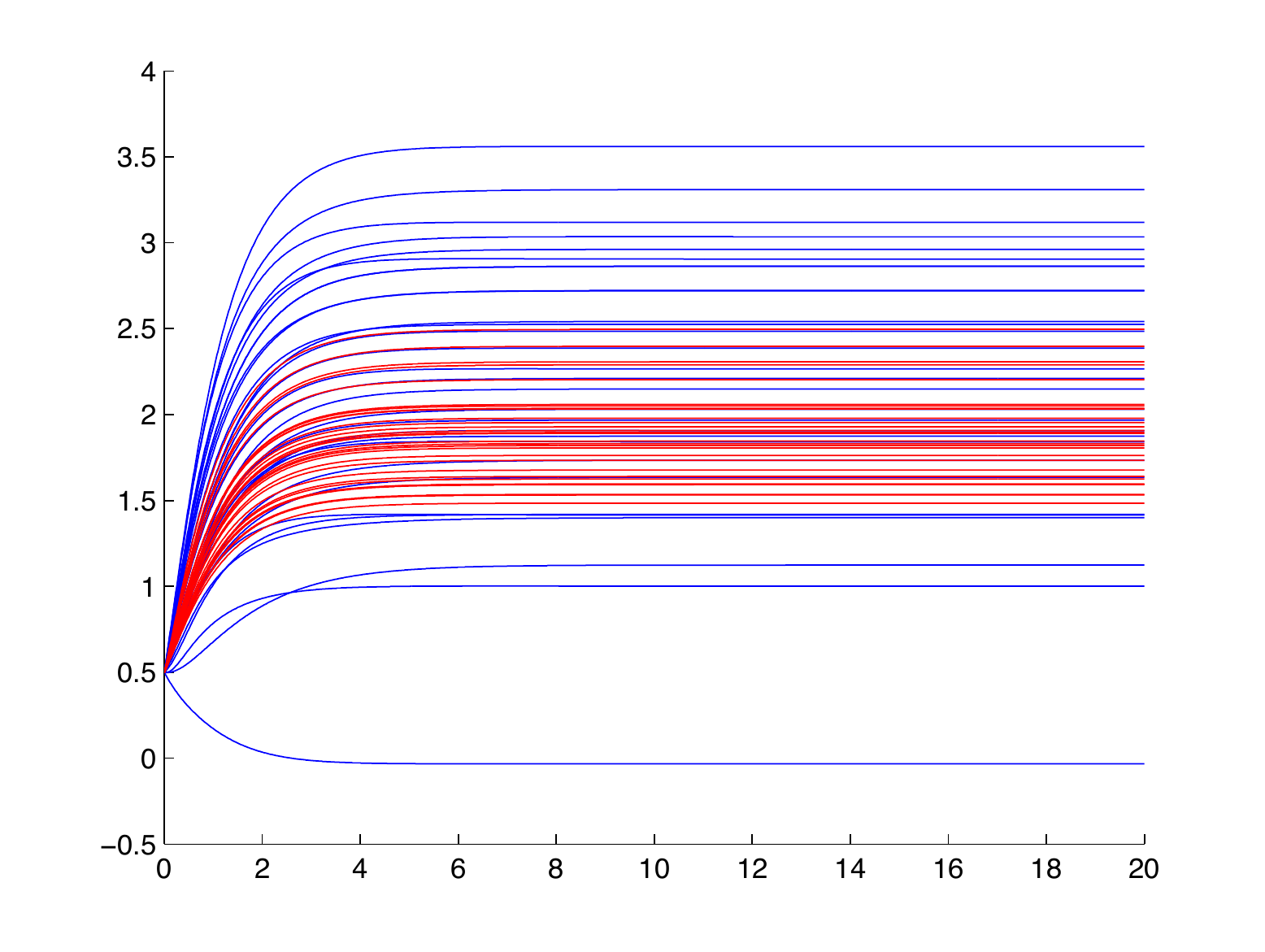}}\qquad
		\subfigure[Localized Chaos]{\includegraphics[width=.4\textwidth]{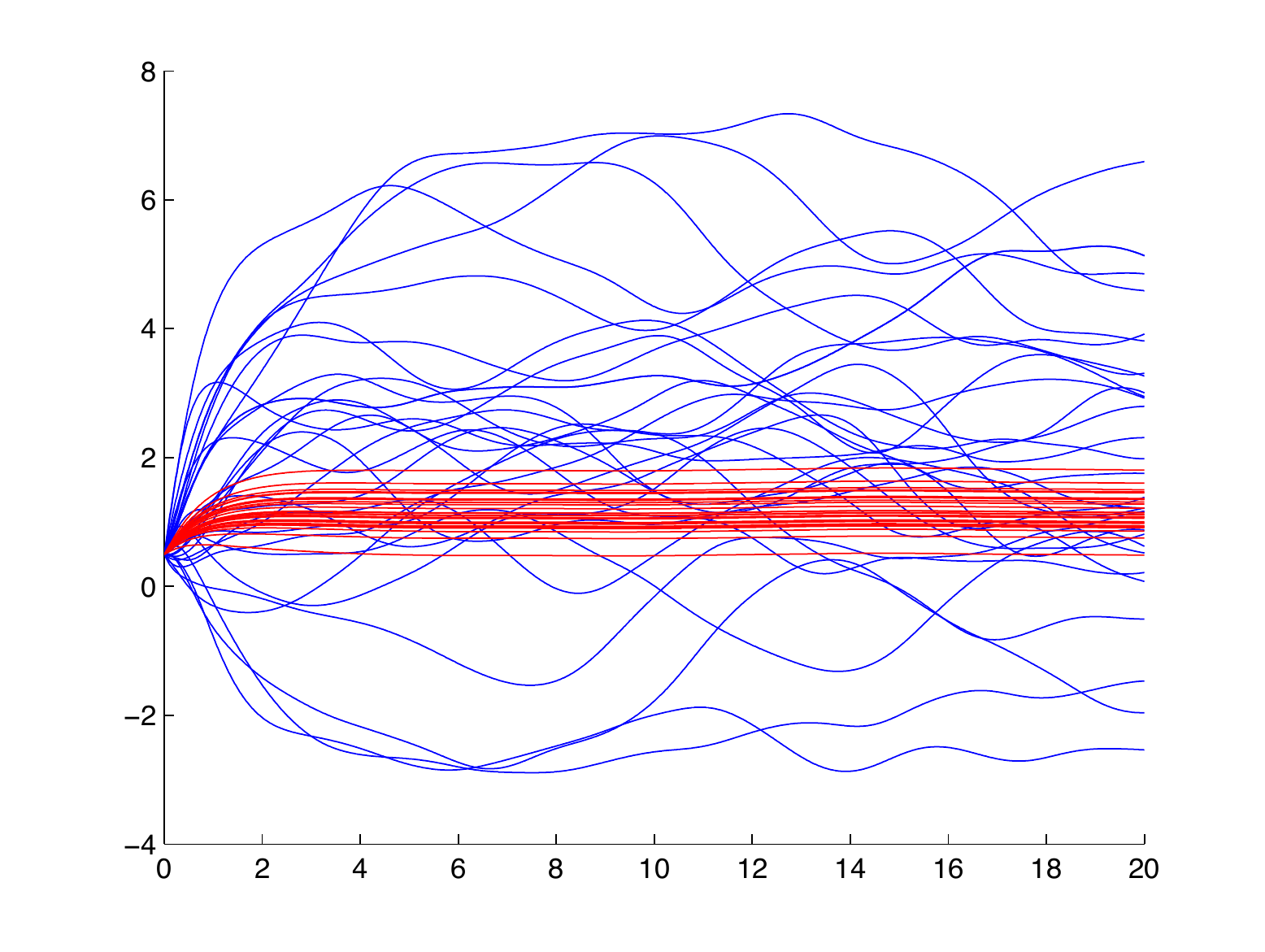}}
	\caption{Stationary and partially chaotic solutions of a two populations network with deterministic lateral around non-zero fixed points. $J_{12}=J_{21}=4$, $\sigma_2=0.5$, (left): $\sigma_1=2$: both populations display a stationary behavior, (right): $\sigma_1=5$: only population 1 is chaotic and the irregularity is not transmitted to population 2. }
	\label{fig:LocChaosNoZero}
\end{figure}

\subsection{Heterogeneity-induced oscillations in two-populations networks}\label{sec:2popsOscill}
	We eventually discuss the effect of heterogeneities in a more biologically plausible neuronal network including one excitatory and one inhibitory population, with strictly positive sigmoidal transforms (since these functions model the input to firing-rate transformation), that tend to zero at $-\infty$ and to $1$ at $\infty$. This system was analyzed in~\cite{touboul-hermann:12}. Considering $M=2$ populations, all sigmoids equal to $\erf(gx)=\int_{-\infty}^{gx} e^{-y^2/2}/\sqrt{2\pi}\,dy$ (yielding $f_{\alpha\beta}(x,v)=\erf(gx/\sqrt{1+g^2v})$), all time constants $\theta_{\alpha}=1$, and the connectivity matrix, inspired from the seminal article of Wilson and Cowan~\cite{wilson-cowan:72}:
\[\bar{J}=\left ( \begin{array}{cc}
		15 & -12\\
		16 & -5
	\end{array}
	\right).\]
we showed that the system presents phase transitions as a function of the heterogeneity parameter, between stationary distributions to periodic oscillations (see figure Fig.~\ref{fig:NoiseInducedOscillations}): considering all $\sigma_{\alpha\beta}$ equal and denoting $\sigma$ the common value, we observe that for small heterogeneity parameter $\sigma$, the network converges towards a stationary solution with non-zero mean. For intermediate values of the heterogeneity, phase-locked perfectly periodic behaviors appear at the network level, that disappear, as heterogeneity is further increased, through a SCS phase transition yielding chaotic activity.
	\begin{figure}
		\centering
			\subfigure[Bifurcation Diagram]{\includegraphics[width=.4\textwidth]{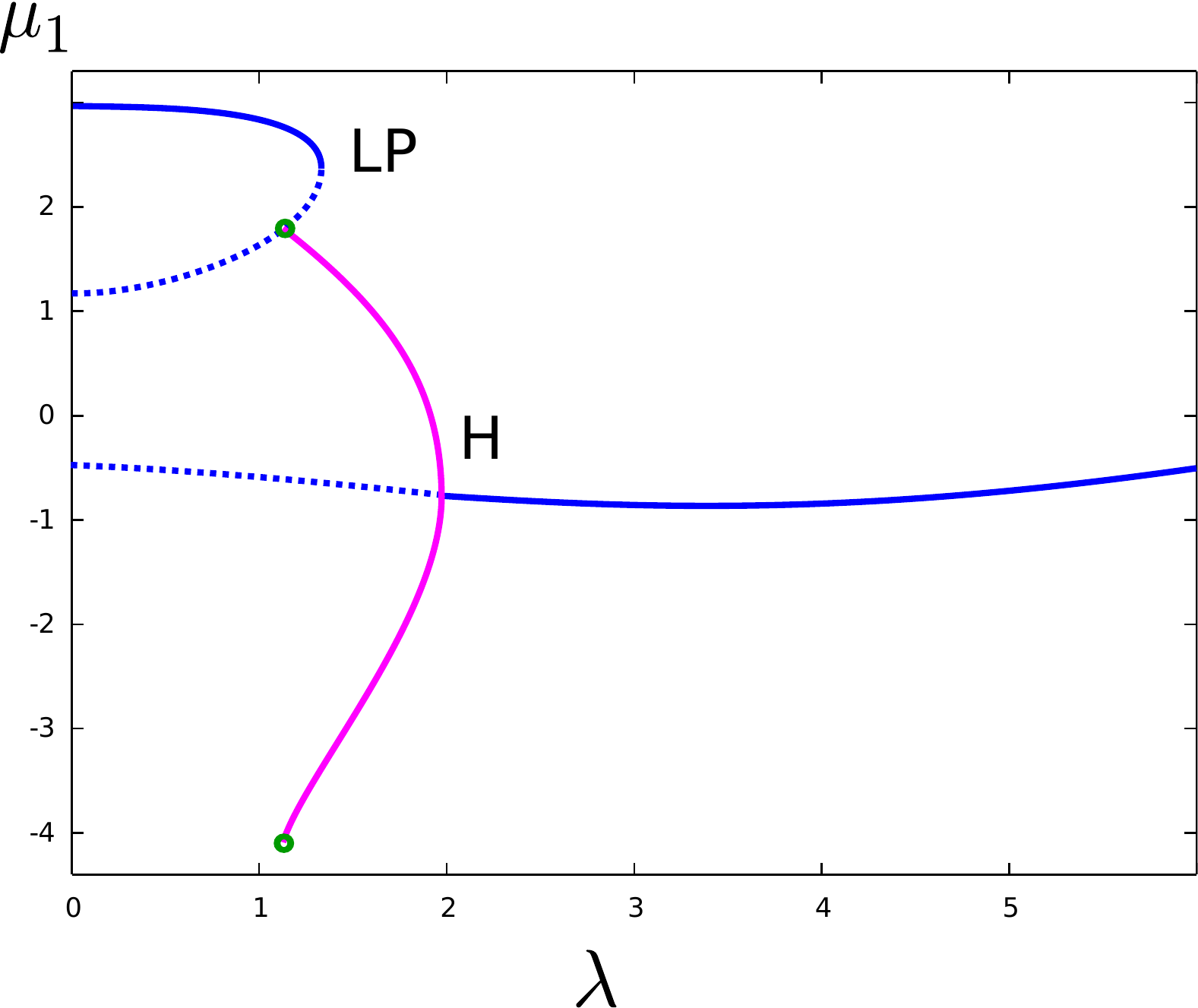}}\\
			\subfigure[$\sigma=0.9$]{\includegraphics[width=.3\textwidth]{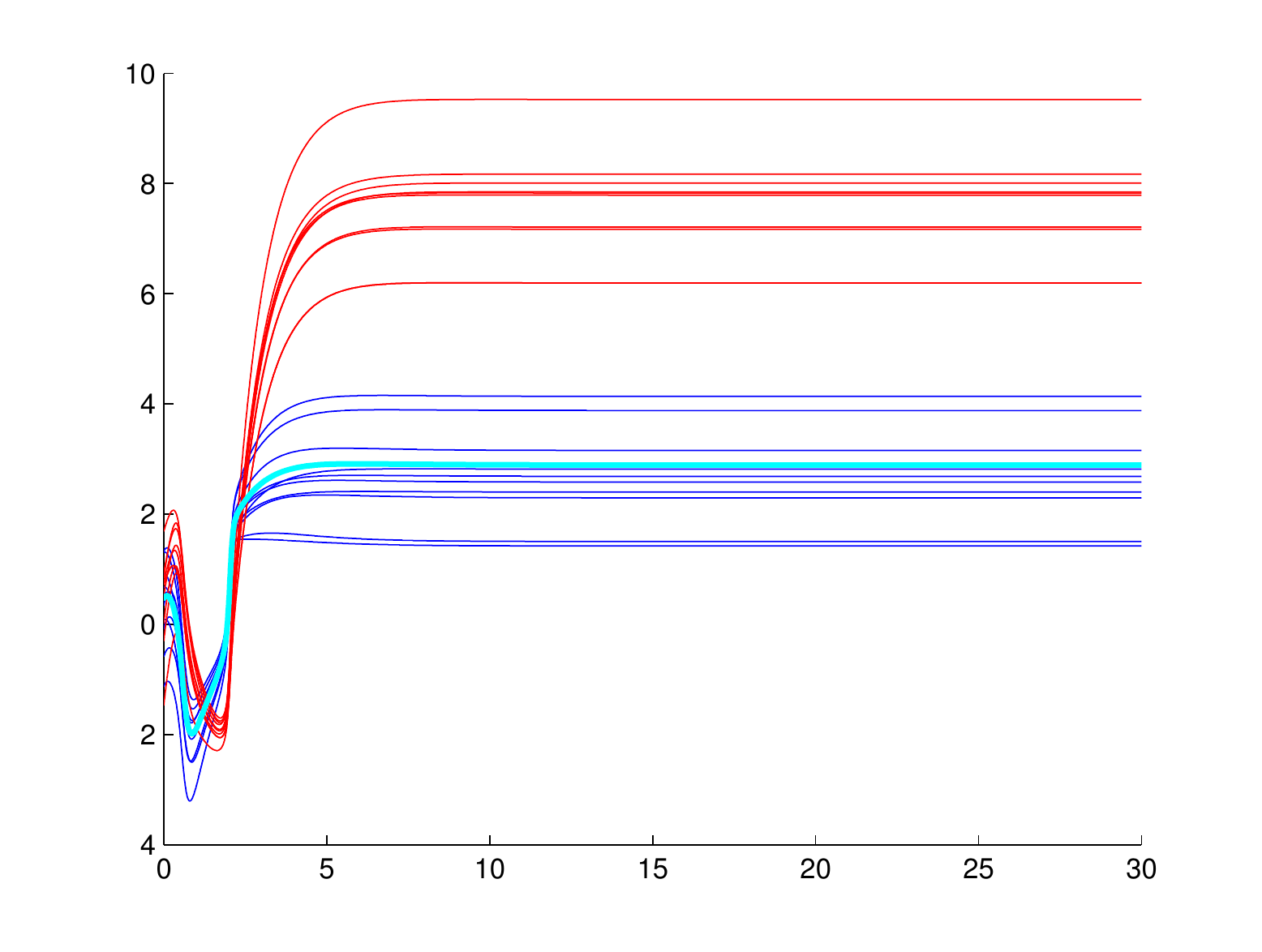}}
			\subfigure[$\sigma=1.6$]{\includegraphics[width=.3\textwidth]{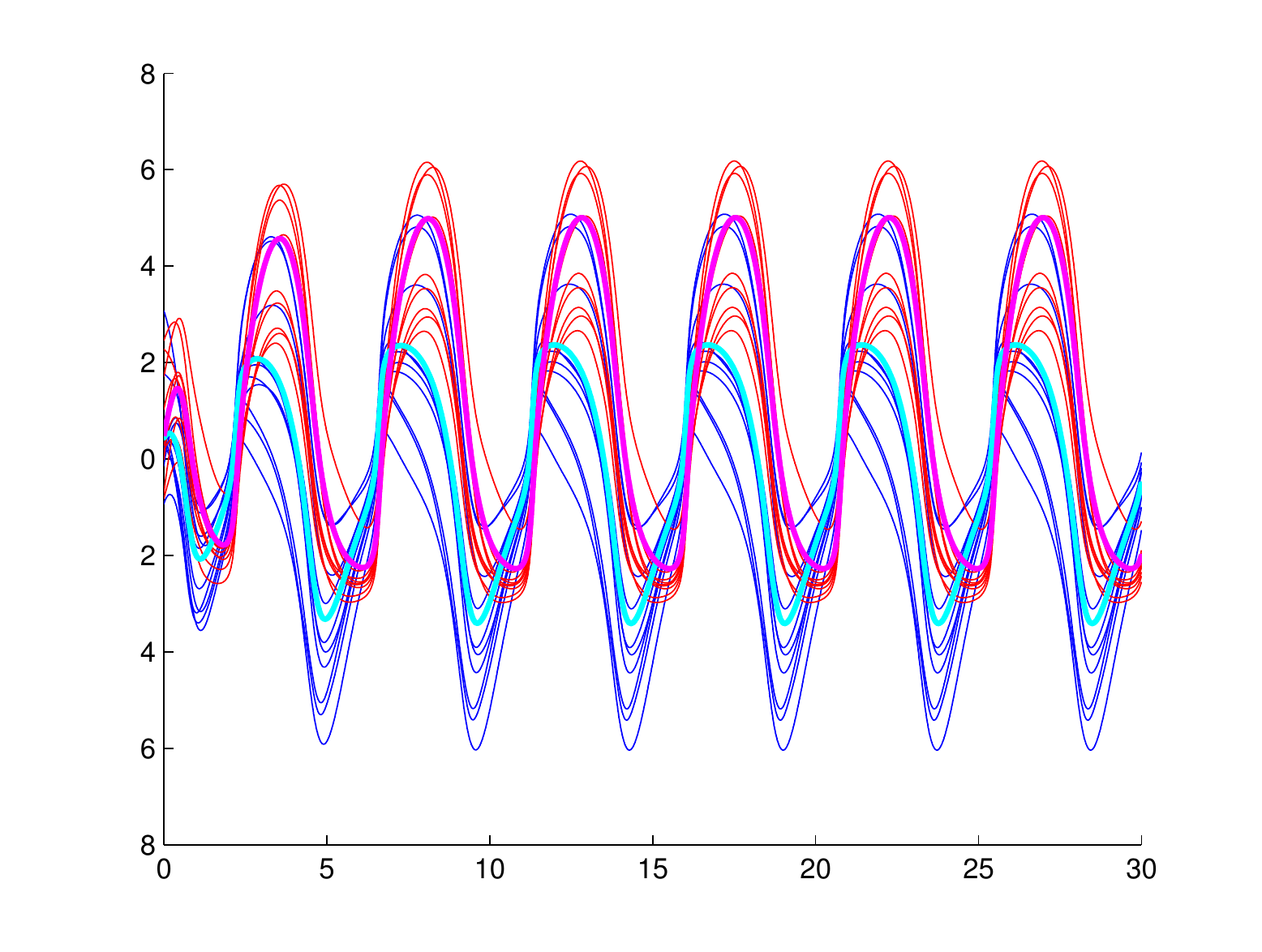}}
			\subfigure[$\sigma=3.5$]{\includegraphics[width=.3\textwidth]{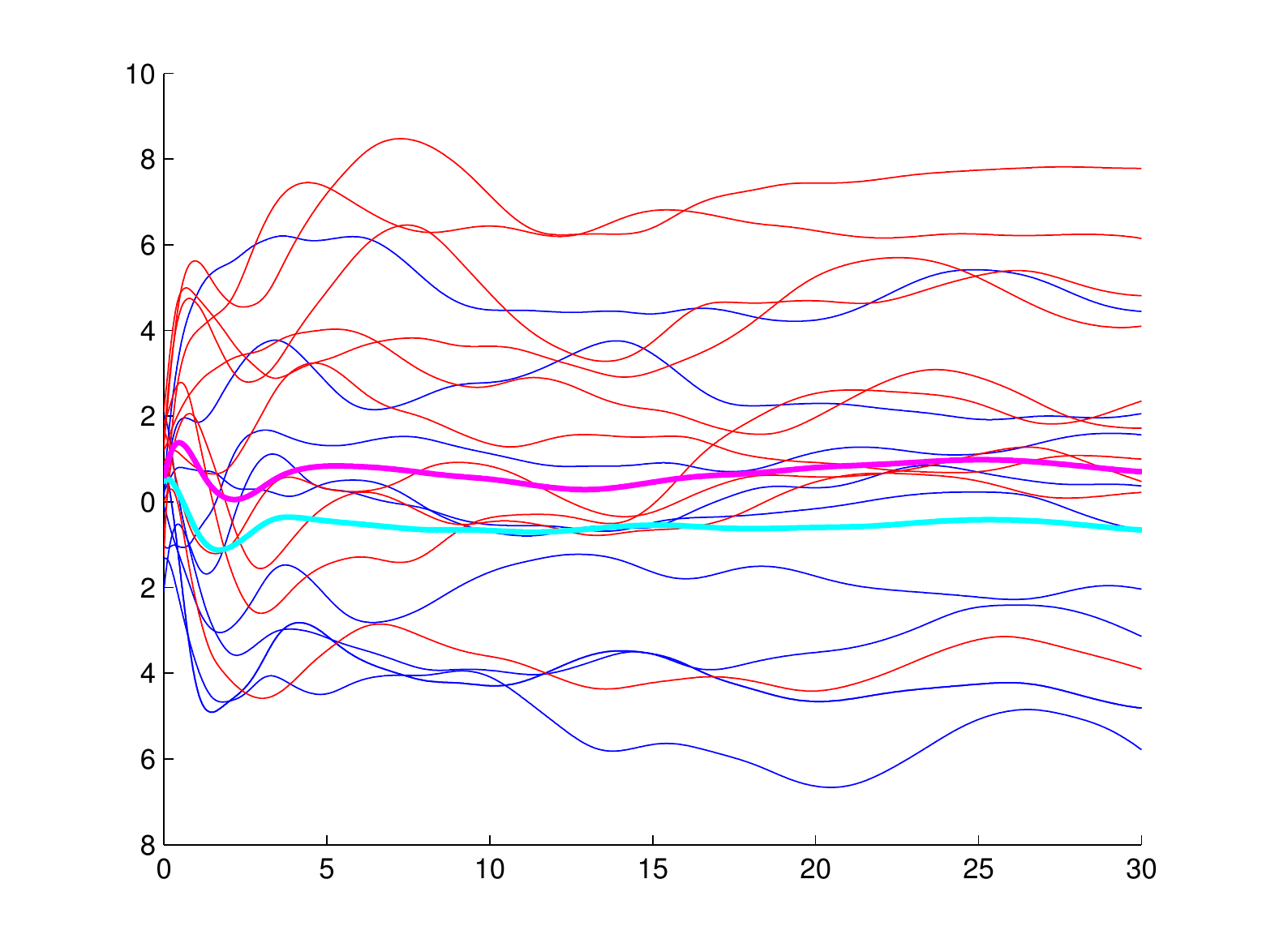}}
		\caption{Heterogeneity-induced oscillations in a two populations excitatory-inhibitory network. (a) Bifurcation diagram of the system of ODEs~\eqref{eq:Means} as a function of a presumably constant common value $\Gamma(\sigma)$. LP: saddle-node bifurcation, H: Hopf bifurcation, Sh: Saddle-homoclinic bifurcation, blue: fixed points (solid: stable, dashed: unstable), magenta: limit cycles. (b)-(d): simulations of the heterogeneous network with $2\,000$ neurons per population, for increasing values of the heterogeneity parameter $\sigma$ displays a transition from stationary to periodic phase-locked behaviors, and then to chaotic behavior. blue (resp. red): 30 arbitrary trajectories of population 1 (resp. 2), cyan (resp. magenta): average on all neurons of population 1 (resp. 2).}
		\label{fig:NoiseInducedOscillations}
	\end{figure}
	
	This phenomenon relates the level of heterogeneity to the presence of synchronized oscillations in networks, an essential phenomenon, as we discuss in the conclusion.
	
\section{Discussion}

In this manuscript, we analyzed randomly coupled neuronal networks and derived their limit as the number of neurons tends to infinity. To this purpose, we showed that the system satisfied a large deviation principle and exhibited the related good rate function. This approach generalized the work of Gerard Ben Arous and Alice Guionnet~\cite{ben-arous-guionnet:95,guionnet:97} developed for spin glasses in three main directions: (i) the synaptic weights are not centered, introducing additional, deterministic terms in the coupling, (ii) interactions are delayed, which projects the problem into infinite dimensions, and (iii) the system is composed of several populations, which was handled showing that empirical measures on each populations simultaneously satisfy a large deviation principle. The proof is made on a particular model very popular in physics and neurosciences, the Wilson and Cowan system, which is close of the famous Hopfield model, but as mentioned, can be easily generalized to nonlinear neuron models. Indeed, most of the proofs deal a quantity which is related to the density of the coupled network with respect to the uncoupled dynamics, and this quantity is independent of the dynamics of individual cells. Our approach can be also readily extended to networks with distributed delays. Eventually, let us note that this result provides large-deviations estimates on the convergence of deterministically coupled networks as studied in~\cite{touboul-hermann-faugeras:11}.

From the neuroscience viewpoint, this approach has the interest to justify an equation which has been widely used in the analysis of large-scale networks, and introduced in the seminal paper of Sompolinsky and collaborators~\cite{sompolinsky-crisanti-etal:88}. Moreover, our setting substantially extends their result by taking into account important features present in cortical networks: interconnection delays, multiple populations with non-zero average synaptic connection. All these refinements allowed going deeper into the understanding of the dynamics of neuronal networks. In particular, we showed that delays can induce oscillations in one-populations networks modulated by the level of heterogeneity, and that non-zero average connectivity yields non-trivial dynamics that were not present in the original SCS model. Moreover, we showed that networks with multiple populations can show relative counter-intuitive phenomena such as the localization of chaos: a few populations can have a chaotic behavior which is not transmitted to the other populations, whatever the connection strength. Another phenomenon we illustrated was the apparition of heterogeneity-induced oscillations, phenomenon first presented in a very recent article~\cite{touboul-hermann:12}. As discussed in that article, the latter phenomenon can be related to experimental studies that showed that the cortex of developing rats subject to absence seizures (abnormal synchronization of some cortical areas) was characterized by the same mean synaptic connectivity, but increased dispersion of the synaptic weights~\cite{aradi-soltesz:02}. We further showed here that such oscillations were facilitated by the presence of delays. Let us eventually underline that the particular form of our system is not essential in the apparition of such phenomena, and in~\cite{touboul-hermann:12}, it is shown that the transition to synchronized activity as a function of heterogeneity levels persists for realistic neuronal networks made of excitable cells, the Fitzhugh-Nagumo model.

An important observations is that in all the examples treated, the SCS phase transition to chaos is present as the heterogeneity is increased. This phenomenon seems relatively universal in this kind of randomly coupled neuronal networks. It was very recently related to the topological complexity of the underlying vector field in~\cite{wainrib-touboul:2012} in the original SCS framework, and we conjecture that the same phenomenon occurs in our more complex settings.

The analysis of the present manuscript underlines the fact that the structure of connectivity maps is essential to the function of the networks, and illustrated the fact that averaging effects do not cancel the structure into populations and allow serving functions such as oscillations. However, actual brain connectivity maps are not recurrent, and can display different topologies, with different computational capabilities. The extension of our methods to such networks is an active part of our future research. Moreover, our analysis did not take into account the plasticity mechanisms, resulting in the slow evolution of the synaptic weights as a function of the activity of neurons, which tends to correlate the synaptic weights to the voltage variables. Techniques to rigorously address the dynamics of neuronal networks with non-recurrent connectivity, with specific topologies, or with correlated synaptic weights, are deep questions that are still largely left unexplored, and we can expect that a wide range of novel phenomena will arise from the analysis of such networks.

\appendix

\section{Proofs}\label{append:Proofs}
\subsection{Proof of Lemma~\ref{lemma1}}\label{append:ProofLemma1}

This appendix is the proof of lemma~\ref{lemma1} concerned with the properties of the functions used in our large deviation principle in a discrete time framework.

\begin{proof}
	Before proceeding to the proof of the lemma, let us start by showing an inequality which will be very useful in several steps of the demonstration.
	We recall the following well known formula (see~\cite[Lemma 3.2.13]{deuschel-stroock:89} for instance):
	\[I(\mu\vert P) = \sup \left(\int_{\C} \Phi d\mu - \log \int_{\C} \exp \Phi dP \;;\; \Phi \in \C_b(\C)\right)\]
	so that, for any bounded measurable function $\Phi$ on $\C$, we have:
	\begin{equation}\label{eq:IneqRelativeEntropy}
		\int_{\C} \Phi d\mu \leq I(\mu\vert P) +  \log \int_{\C} \exp \Phi dP
	\end{equation}
    The result hold for positive measurable functions by the monotone convergence theorem.
	
	\noindent {\bf Proof of Lemma~\ref{lemma1}.(i).} \\
    Lets prove the lipschitzianity  result for $\Gamma^{\alpha, k}_1$.
	We have:
	\begin{align*}
	\Big|\log\Big( 1 + & \frac{\int \exp\big(  - \frac{1}{2} \sum_{l=0}^k {G^{\alpha}_{t_l}}^2(t_{l+1} - t_l)\big) d(\gamma_{K_{\mu}}-\gamma_{K_{\nu}})}{\int \exp\big(  - \frac{1}{2} \sum_{l=0}^k {G^{\alpha}_{t_l}}^2(t_{l+1} - t_l)\big) d\gamma_{K_{\nu}}}\Big)\Big| = \Big|\Gamma_1^{\alpha,k}(\mu)-\Gamma_1^{\alpha,k}(\nu)+\frac{1}{2}\int ((m^{\alpha}_{\mu})^2-(m^{\alpha}_{\nu})^2)(\tk)dt\Big| \\
& \leq \exp\Big\{\frac{\ka T}{2\la^2}\Big\}\Big|\int \exp\big(  - \frac{1}{2} \int_0^T {G^{\alpha}_{\tk}}^2dt\big) d(\gamma_{K_{\mu}}-\gamma_{K_{\nu}})\Big|
\end{align*}
Let $\xi$ be a probability measure on $\C \times \C$ with marginals $\mu$ and $\nu$, and let $\gamma_{\xi}$ be the law of a bidimensional centered gaussian process $(G^{\alpha}, \widetilde{G^{\alpha}})$ with covariance $K^{\alpha}_{\xi}$:
\begin{equation}
K^{\alpha}_{\xi}(s,t) = \sum_{\gamma=1}^M \frac{\sag^2}{\la^2}
\left(
\begin{array}{ccc}
\int \Sag(x^{\gamma}_{s-\taag})\Sag(x^{\gamma}_{t-\taag}) \, d\xi(x,y) & \int \Sag(x^{\gamma}_{s-\taag})\Sag(y^{\gamma}_{t-\taag}) \, d\xi(x,y) \\
\int \Sag(y^{\gamma}_{s-\taag})\Sag(x^{\gamma}_{t-\taag}) \, d\xi(x,y) & \int \Sag(y^{\gamma}_{s-\taag})\Sag(y^{\gamma}_{t-\taag}) \, d\xi(x,y) \\
\end{array}
\right) \label{kxi},
\end{equation}
Then,
\begin{multline*}
\Big|\int \exp\big(  - \frac{1}{2} \int_0^T {G^{\alpha}_{\tk}}^2dt\big) d(\gamma_{K_{\mu}}-\gamma_{K_{\nu}})\Big| = \Big|\int \bigg\{ \exp\big(  - \frac{1}{2} \int_0^T {G^{\alpha}_{\tk}}^2dt\big) - \exp\big(  - \frac{1}{2} \int_0^T \widetilde{G^{\alpha}_{\tk}}^2dt\big)\bigg\} d\gamma_{\xi}\Big|\\
\qquad \leq \frac 1 2 \int \int_0^T \big| {G^{\alpha}_{\tk}}^2-\widetilde{G^{\alpha}_{\tk}}^2 \big| dt  d\gamma_{\xi}\\
\leq \frac 1 2 \prod_{\varepsilon = \pm 1} \bigg(\int\int_0^T (G^{\alpha}_{\tk}+\varepsilon \widetilde{G^{\alpha}_{\tk}})^2 dt d\gamma_{\xi} \bigg)^{\frac 1 2}
\end{multline*}
by Cauchy-Schwarz inequality. Then, using the covariance of $(G^{\alpha},\widetilde{G^{\alpha}})$ under $\gamma_{\xi}$, we find:
\begin{align}
\Big|\Gamma_1^{\alpha,k}(\mu)-\Gamma_1^{\alpha,k}(\nu)+\frac{1}{2}\int ((m^{\alpha}_{\mu})^2-(m^{\alpha}_{\nu})^2)(\tk)dt\Big|\\
	& \leq \frac{1}{2} \exp\Big\{\frac{\ka T}{2\la^2}\Big\} \Big(\frac{4\ka T}{\la^2}\Big)^{\frac{1}{2}}  \bigg\{\frac{1}{\la^2}\sum_{\gamma=1}^M \sag^2 \int \int_0^T (\Sag(x^{\gamma}_{t-\taag})-\Sag(y^{\gamma}_{t-\taag}))^2 dt \, d\xi(x,y) \bigg\}^{\frac{1}{2}} \nonumber \\
	& \leq \frac{\ka}{\la^2} \sqrt{T} \exp\Big\{\frac{\ka T}{2\la^2}\Big\} \max_{\gamma=1\cdots M} \bigg\{ \int \int_0^T \big|\Sag(x^{\gamma}_{t-\taag})-\Sag(y^{\gamma}_{t-\taag})\big|^2 dt \, d\xi(x,y) \bigg\}^{\frac{1}{2}} \label{ineqlemma1}
	\end{align}

	Moreover, we have:
	\begin{align*}
	\Big|\int {m^{\alpha}_{\mu}(\tk)}^2 -{m^{\alpha}_{\nu}(\tk)}^2 dt\Big| & = \int \Big|\Big(m^{\alpha}_{\mu}(\tk) -m^{\alpha}_{\nu}(\tk)\Big)\Big(m^{\alpha}_{\mu}(\tk) + m^{\alpha}_{\nu}(\tk)\Big)\Big| dt \\
	& \leq 2 \frac{\Ja}{\la} \int \big|m^{\alpha}_{\mu}(\tk) -m^{\alpha}_{\nu}(\tk)\big| dt
	\end{align*}
	But
	\begin{align*}
	\int_0^T |(m^{\alpha}_{\mu}-m^{\alpha}_{\nu})(t)| dt & = \int_0^T \Big|\frac{1}{\la} \sum_{\gamma=1}^M \Jag \int \Sag (x^{\gamma}_{t-{\tau}}) d(\mu-\nu)(x)\Big| \, dt \\
	& \leq \frac{1}{\la} \sum_{\gamma=1}^M |\Jag| \:  \int_0^T \Big|\int \Sag (x^{\gamma}_{t-{\taag}}) d(\mu-\nu)(x)\Big| \, dt \\
	& \leq \frac{1}{\la} \sum_{\gamma=1}^M |\Jag| \:  \int \int_0^T  |\Sag (x^{\gamma}_{t-{\taag}}) - \Sag (y^{\gamma}_{t-{\tau}})| dt \, d\xi(x,y) \\
	& \leq \frac{ \Ja}{\la} \max_{\gamma=1\cdots M}  \Big( \int \int_0^T \big|\Sag (x^{\gamma}_{t-{\taag}}) - \Sag (y^{\gamma}_{t-{\tau}})\big|^2 dt \, d\xi(x,y)\Big)^{\frac{1}{2}} \\
	\end{align*}
	by Cauchy-Schwarz inequality. \\
	Consequently:
	\begin{multline}
	|\Gamma_1^{\alpha,k}(\mu) - \Gamma_1^{\alpha,k}(\nu)| \leq \big(\frac{ \Ja^2}{\la^2} + \frac{\ka}{\la^2} \sqrt{T} \exp\Big\{\frac{\ka T}{2\la^2}\Big\}\big)  \max_{\gamma=1\cdots M}  \Big( \int \int_0^T \big|\Sag (x^{\gamma}_{\tk-\taag}) - \Sag (y^{\gamma}_{\tk-\taag})\big|^2 dt \, d\xi(x,y)\Big)^{\frac{1}{2}} \label{gamma1}
	\end{multline}
	As the $\Sag$ are $K_S$ Lipschitz, we have: \\
	\begin{equation}
	|\Gamma_1^{\alpha,k}(\mu) - \Gamma_1^{\alpha,k}(\nu)| \leq K_S \sqrt{T} \big(\frac{\Ja^2}{\la^2} + \frac{\ka}{\la^2} \sqrt{T} \exp\Big\{\frac{\ka T}{2\la^2}\Big\}\big)   d_T(\mu,\nu),
	\end{equation}
	so that $\Gamma_1^{\alpha,k}$ is Lipschitz for the Vaserstein distance. Using the triangle inequality, the result holds for  $\Gamma_1^k$.\\

\noindent {\bf Proof of Lemma\ref{lemma1}.(ii):}\\
Let
\begin{equation*}
F_{\mu}(x)= \log\bigg\{ \int  \exp\Big\{ \int_0^T \big(\Gv_{\tk}(\omega)+\m_{\mu}(\tk)\big)' \cdot d\W_t(x) \\
- \frac{1}{2} \int_0^T \Big\|\Gv_{\tk}(\omega)+\m_{\mu}(\tk)\Big\|^2 dt \Big\} d\gamma_{\mu} \bigg\}.
\end{equation*}
This function is a.s. finite but not bounded, let us hence define for $A \in \R^{+}$
\begin{equation*}
F_{\mu}^A(x)= \log\bigg\{ \int  A\wedge\exp\Big\{ \int_0^T \big(\Gv_{\tk}(\omega)+\m_{\mu}(\tk)\big)' \cdot d\W_t(x) \\
- \frac{1}{2} \int_0^T \Big\|\Gv_{\tk}(\omega)+\m_{\mu}(\tk)\Big\|^2 dt \Big\} d\gamma_{\mu} \bigg\}.
\end{equation*}
By the monotone convergence theorem and using equation~\eqref{eq:IneqRelativeEntropy}, we have for any $a\geq1$:
\begin{equation*}
a \int F_{\mu}(x) d\mu(x) \leq I(\mu|P)+ \log\bigg\{ \int \exp{ aF_{\mu}(x)} dP(x)\bigg\}.\\
\end{equation*}
By Jensen inequality and Fubini theorem,\\
\begin{align*}
\int \exp{ (aF_{\mu}(x)}) dP(x) & \leq \int \prod_{\alpha=1}^M \int  \exp\Big\{ a \int_0^T \big(G^{\alpha}_{\tk} + m^{\alpha}_{\mu}(\tk) \big) dW^{\alpha}_t(x)\Big\} d\pa(x) \\
& \exp\Big\{ - \frac{a}{2} \int_0^T \Big\|\Gv_{\tk}(\omega)+\m_{\mu}(\tk)\Big\|^2 dt \Big\}  d\gamma_{\mu}.\\
\end{align*}
But, as $W^{\alpha}$ is a $\pa$-Brownian motion,
\begin{equation*}
\int \exp\bigg\{ a \int_0^T \big(G^{\alpha}_{\tk} + m^{\alpha}_{\mu}(\tk) \big) dW^{\alpha}_t(x)\bigg\} d\pa(x) = \exp\bigg\{ \frac{a^2}{2} \int_0^T \big(G^{\alpha}_{\tk} + m^{\alpha}_{\mu}(\tk) \big)^2 dt \bigg\} \bigg].\\
\end{equation*}
so that,
\begin{equation*}
a \int F_{\mu}(x) d\mu(x) \leq I(\mu|P)+ \log\bigg\{ \int \exp\Big\{\frac{a^2-a}{2} \int_0^T \Big\|\Gv_{\tk}(\omega)+\m_{\mu}(\tk)\Big\|^2 dt \Big\}  d\gamma_{\mu} \bigg\}.\\
\end{equation*}
Letting $a=1$ proves that $\Gamma^k \leq I(|P)$.\\

\noindent{\bf Proof of Lemma\ref{lemma1} (iii):}\\
As the components of $\Gv$ are independent under $\gamma_{\mu}$, we only have to check that, for every $b>0$, there exists a finite constant $C_b$ such that \\
\begin{equation}
{\E}_{\mu} \bigg[ \exp{\Big(\frac{b}{2} \int_0^T \big(G^{\alpha}_s+m^{\alpha}_{\mu}(s)\big)^2 ds \Big) } \bigg] \leq \exp{\frac{b C_b \ka T}{\la^2}} \label{ineqA}.
\end{equation}
It was proved in~\cite[Lemma A.3(2)]{ben-arous-guionnet:95} in their particular framework that for every $b$ verifying $\frac{b\ka T}{\la^2}<1$, there exists a finite constant $c_b$ such that:
\begin{equation*}
\int \exp{\Big(\frac{b}{2} \int_0^T G_s^2 ds \Big) } d\gamma_{\mu} \leq \exp{\frac{b c_b \ka T}{\la^2}}.
\end{equation*}
In our case, the covariance function is slightly different of that of~\cite{ben-arous-guionnet:95}, but the proof and result remain unchanged and can be readily extended.

Moreover, since we have
\begin{equation*}
\big(G^{\alpha}_s + m^{\alpha}_{\mu}(s)\big)^2 \leq 2 {G^{\alpha}_s}^2 + 2 {m^{\alpha}_{\mu}(s)}^2 \leq 2{G^{\alpha}_s}^2 + 2\frac{\Ja^2}{\la^2}
\end{equation*}
we obtain the desired result with the following constant $C_b = 2 c_{2b} + \frac{\Ja^2}{\ka}$, under the condition $\frac{2b \ka T}{\la^2}<1$. \\

\noindent {\bf Proof of Lemma\ref{lemma1}.(iv)}\\
As above, lets prove the result for $\big| \Gamma^{\alpha,k}_{2,\nu} - \Gamma^{\alpha,k}_2 \big|$.
 We have:
\begin{multline*}
|\Gamma_{2,\nu}^{\alpha,k}(\mu) - \Gamma_2^{\alpha,k}(\mu)| \leq  \frac{1}{2} \Big|\int \int \bigg\{ \Big( \int G^{\alpha}_{\tk} (dW^{\alpha}_t - m^{\alpha}_{\mu}(\tk)dt) \Big)^2 - \Big( \int G^{\alpha}_{\tk} (dW^{\alpha}_t - m^{\alpha}_{\nu}(\tk)dt) \Big)^2 \bigg\} d\gamma_{\widetilde{K}_{\mu}^{T,k}} d\mu \Big| \\ +\frac{1}{2} \Big|\int \int \Big( \int G^{\alpha}_{\tk} (dW^{\alpha}_t - m^{\alpha}_{\nu}(\tk)dt) \Big)^2 d\big(\gamma_{\widetilde{K}_{\nu}^{T,k}}-\gamma_{\widetilde{K}_{\mu}^{T,k}}\big) d\mu | + |\int \int (m^{\alpha}_{\nu}-m^{\alpha}_{\mu})(\tk)dW^{\alpha}_t d\mu \Big|
\end{multline*}

Let $\xi$ be a probability measure on $\C \times \C$ with marginals $\mu$ and $\nu$, and let $\gamma_{\xi}$ be the law of a bidimensional centered gaussian process $(G^{\alpha}, \widetilde{G^{\alpha}})$ with covariance $K^{\alpha}_{\xi}$.
Let
\begin{equation*}
\Lambda_T^{\alpha,k}(G^{\alpha})= \frac{\exp{\bigg( -\frac{1}{2} \int_0^T {G^{\alpha}_{\tk}}^2 dt\bigg)}}{\int \exp{\bigg( -\frac{1}{2} \int_0^T {G^{\alpha}_{\tk}}^2 dt\bigg)} \, d\gamma_{\xi}}.
\end{equation*}
As in \cite[lemma 3.4]{ben-arous-guionnet:95}, we can show that:
\begin{multline}
|\Gamma_{2,\nu}^{\alpha,k}(\mu) - \Gamma_2^{\alpha,k}(\mu)| \leq \frac{1}{2} \overbrace{\int \int \Big|\Lambda_T^{\alpha,k}(G^{\alpha})-\Lambda_T^{\alpha,k}(\widetilde{G^{\alpha}})\Big| \Big( \int G^{\alpha}_{\tk} (dW^{\alpha}_t - m^{\alpha}_{\nu}(\tk)dt) \Big)^2 d\gamma_{\xi} d\mu }^{B_1} \\ + \underbrace{\frac{1}{2} \! \prod_{\varepsilon=\pm1} \Bigg( \int \! \int \Lambda_T^{\alpha,k}(\widetilde{G^{\alpha}}) \Big( \int (G^{\alpha}_{\tk} +\varepsilon \widetilde{G^{\alpha}}_{\tk}) (dW^{\alpha}_t - m^{\alpha}_{\nu}(\tk)dt) \Big)^2 d\gamma_{\xi} d\mu \Bigg)^\frac{1}{2}}_{B_2} \\ + \frac{1}{2} \underbrace{ |\int \int \Lambda_T^{\alpha,k}(G^{\alpha}) \bigg\{ \Big( \int G^{\alpha}_{\tk} (dW^{\alpha}_t - m^{\alpha}_{\mu}(\tk)dt) \Big)^2 - \Big( \int G^{\alpha}_{\tk} (dW^{\alpha}_t - m^{\alpha}_{\nu}(\tk)dt) \Big)^2 \bigg\} d\gamma_{\xi} d\mu |}_{B_3} \\ + \underbrace{ \bigg( \int \Big| \int (m^{\alpha}_{\nu}-m^{\alpha}_{\mu})(\tk)dW^{\alpha}_t \Big|^2 d\mu \bigg)^{\frac{1}{2}}}_{B_4} \label{ineqgamma2}
\end{multline}

and
\begin{equation*}
\Lambda_T^{\alpha,k}(G^{\alpha})= \exp\bigg\{ - \Gamma_1^{\alpha,k}(\mu) - \frac{1}{2} \int_0^T {m^{\alpha}_{\mu}}^2(\tk) dt - \frac{1}{2} \int_0^T {G^{\alpha}_{\tk}}^2 dt \bigg\}
\end{equation*}
Hence, we have by Jensen inequality,
\begin{align*}
\Lambda_T^{\alpha,k}(G^{\alpha}) \leq \exp\Big\{\frac{\ka T}{2\la^2}\Big\}
\end{align*}
so that
\begin{align*}
\big|\Lambda_T^{\alpha,k}(G^{\alpha})-\Lambda_T^{\alpha,k}(\widetilde{G^{\alpha}})\big| \leq \exp\Big\{\frac{\ka T}{2\la^2}\Big\} \Big(\frac{1}{2}\int_0^T \big| {G^{\alpha}_{\tk}}^2-\widetilde{G^{\alpha}}_{\tk}^2\big|dt + \big|\Gamma_1^{\alpha,k}(\mu)-\Gamma_1^{\alpha,k}(\nu)+\frac{1}{2}\int ((m^{\alpha}_{\mu})^2-(m^{\alpha}_{\nu})^2)(\tk)dt\big|\Big)
\end{align*}
which eventually gives,
\begin{multline*}
B_1 \leq \frac{1}{2}\exp\Big\{\frac{\ka T}{2\la^2}\Big\}\bigg(\big|\Gamma_1^{\alpha,k}(\mu)-\Gamma_1^{\alpha,k}(\nu)+\frac{1}{2}\int ((m^{\alpha}_{\mu})^2-(m^{\alpha}_{\nu})^2)(\tk)dt\big| \int\int \Big( \int_0^T G^{\alpha}_{\tk} (dW^{\alpha}_t - m^{\alpha}_{\nu}(\tk)dt) \Big)^2 d\gamma_{\xi}d\mu \\
+ \int\int \Big(\int_0^T \big|(G^{\alpha}_{\tk})^2-\widetilde{G^{\alpha}}_{\tk}^2\big| dt\Big) \Big( \int_0^T G^{\alpha}_{\tk} (dW^{\alpha}_t - m^{\alpha}_{\nu}(\tk)dt) \Big)^2 d\gamma_{\xi}d\mu \bigg)
\end{multline*}

Let $h, m \in L^2([0;T], dt)$, with $m$ bounded. By Cauchy-Schwarz and the relative entropy inequality~\eqref{eq:IneqRelativeEntropy} (with $\Phi(x)=\Big( \int_0^T h_t dW^{\alpha}_t(x) \Big)^2 \sim \mathcal{N}\left(0,\int_0^T h_t^2 dt\right)^2$ under $\pa$, and besides is a positive and measurable function of $\C$), we have the existence of a finite constant $C$ such that,
\begin{align}
\int \Big( \int_0^T h_t (dW^{\alpha}_t(x) - m(t)dt) \Big)^2 d\mu(x) & \leq  2\bigg\{ \int \Big(\int_0^T h_t dW^{\alpha}_t \Big)^2 + \Big( \int_0^T h_t m_t dt \Big)^2 d\mu \bigg\} \nonumber \\
& \leq 2 \bigg\{ \bigg(C\big(1+I(\mu|P)\big) + m^2_{\infty}T \bigg) \Big( \int_0^T h_t^2 dt \Big)\nonumber \\
& \leq C' \big(1+I(\mu|P)\big) \Big( \int_0^T h_t^2 dt \Big) \label{ineqphi}
\end{align}
We can now bound the different terms in inequality (\ref{ineqgamma2}).

In fact, as $h_t=G^{\alpha}_{\tk}$ and $m_t=m^{\alpha}_{\nu}(\tk)$ verify the required condition, (\ref{ineqphi}) gives the existence of $c_T$,
\begin{align*}
\int\Big( \int_0^T G^{\alpha}_{\tk} (dW^{\alpha}_t(x) - m^{\alpha}_{\nu}(\tk)dt) \Big)^2 d\mu(x) & \leq  \: c_T \big(1+I(\mu|P)\big) \int_0^T (G^{\alpha}_{\tk})^2 dt
\end{align*}

Hence, we can find a finite constant $c'_T$ such that
\begin{align*}
B_1 \leq c'_T \big(1+I(\mu|P)\big) \max_{\gamma=1\cdots M} \Big( \int \int_0^T \big|\Sag(x^{\gamma}_{\tk-\taag})-\Sag(y^{\gamma}_{\tk-\taag})\big|^2 dt \, d\xi(x,y) \Big)^{\frac{1}{2}}
\end{align*}

Similarly, there exists a constant $c_T$ such that
\begin{align*}
B_2 & \leq \frac{1}{2}\exp\Big\{\frac{\ka T}{2\la^2}\Big\} \prod_{\varepsilon=\pm1} \Bigg( c_T \big(1+I(\mu|P)\big)  \int \int (G^{\alpha}_{\tk} +\varepsilon \widetilde{G^{\alpha}}_{\tk})^2 dt  d\gamma_{\xi} \Bigg)^{\frac{1}{2}} \\
& \leq c'_T \big(1+I(\mu|P)\big) \max_{\gamma=1\cdots M} \Big( \int \int_0^T \big|\Sag(x^{\gamma}_{\tk-\taag})-\Sag(y^{\gamma}_{\tk-\taag})\big|^2 dt \, d\xi(x,y) \Big)^{\frac{1}{2}}
\end{align*}

To bound $B_3$, we first use Cauchy-Schwarz inequality:
\begin{align}
B_3 & \leq \frac{1}{2} \exp\Big\{ \frac{\ka T}{2\la^2} \Big\} \prod_{\varepsilon = \pm 1} \bigg\{ \int \int \Big| \int_0^T G^{\alpha}_{\tk} \big( (1+\varepsilon)dW^{\alpha}_t - (m^{\alpha}_{\nu}(\tk) + \varepsilon m^{\alpha}_{\mu}(\tk))dt \big) \Big|^2 d\gamma_{\xi} d\mu \bigg\}^{\frac{1}{2}}  \label{ineqB3}
\end{align}
But
\begin{align*}
\Big| \int_0^T G^{\alpha}_{\tk} \big(m^{\alpha}_{\mu}(\tk) -m^{\alpha}_{\nu}(\tk)\big) dt \Big|^2 & \leq \Big(\int_0^T {G_{\tk}^{\alpha}}^2 dt\Big) \Big( \int_0^T \big(m^{\alpha}_{\mu}(\tk) -m^{\alpha}_{\nu}(\tk)\big)^2 dt \Big)
\end{align*}

Remark that
\begin{align*}
\int_0^T \Big(m^{\alpha}_{\mu}(\tk) - m^{\alpha}_{\nu}(\tk)\Big)^2 dt & = \int_{0}^T \Big( \sum_{\gamma=1}^M \frac{\Jag}{\la} \int \Sag(x^{\gamma}_{\tk-\taag}) d(\mu - \nu)(x) \Big)^2 dt \\
& \leq \frac{\Ja^2}{\la^2}  \max_{\gamma=1\cdots M} \int_{0}^T  \Big( \int \Sag(x^{\gamma}_{\tk-\taag}) d(\mu - \nu)(x) \Big)^2 dt \\
& \leq \frac{\Ja^2}{\la^2}  \max_{\gamma=1\cdots M} \int_{0}^T  \Big( \int |\Sag(x^{\gamma}_{\tk-\taag}) - \Sag(y^{\gamma}_{\tk-\taag})| d\xi(x,y) \Big)^2 dt
\end{align*}
So that
\begin{multline*}
\Big| \int_0^T G^{\alpha}_{\tk} \big(m^{\alpha}_{\mu}(\tk) -m^{\alpha}_{\nu}(\tk)\big) dt \Big|^2 \\
\leq \frac{\Ja^2}{\la^2} \Big(\int_0^T G^2_{\tk} dt\Big) \max_{\gamma=1\cdots M} \int_{0}^T  \Big( \int |\Sag(x^{\gamma}_{\tk-\taag}) - \Sag(y^{\gamma}_{\tk-\taag})| d\xi(x,y) \Big)^2 dt
\end{multline*}

Moreover, (\ref{ineqphi}) gives:
\begin{align*}
\int \Big\{ \int_0^T 2 G^{\alpha}_{\tk}\big(dW^{\alpha}_t - \frac{m^{\alpha}_{\mu}(\tk)+m^{\alpha}_{\nu}(\tk)}{2} dt\big) \Big\}^2 d\mu & \leq c_T \big( 1+ I(\mu|P) \big) 4 \int_0^T {G^{\alpha}_{\tk}}^2 dt
\end{align*}

Using the last two inequalities in (\ref{ineqB3}) we have:
\begin{align*}
B_3 & \leq \frac{1}{2} \exp\Big\{ \frac{\ka T}{2\la^2} \Big\} \Big\{ \int c_T \big( 1+ I(\mu|P) \big) 4 \Big(\int_0^T {G^{\alpha}_{\tk}}^2 dt\Big) d\gamma_{\xi} \Big\}^{\frac{1}{2}} \\
 \Big\{ \int & \frac{\Ja^2}{\la^2} \Big(\int_0^T {G^{\alpha}_{\tk}}^2 dt\Big) \max_{\gamma=1\cdots M} \int_{0}^T  \Big( \int |\Sag(x^{\gamma}_{\tk-\taag}) - \Sag(y^{\gamma}_{\tk-\taag})| d\xi(x,y) \Big)^2 dt d\gamma_{\xi} \Big\}^{\frac{1}{2}} \\
& \leq c'_T \big( 1+ I(\mu|P) \big) \max_{\gamma=1\cdots M} \bigg\{\int_{0}^T \int |\Sag(x^{\gamma}_{\tk-\taag}) - \Sag(y^{\gamma}_{\tk-\taag})|^2 d\xi(x,y)  dt\bigg\}^\frac{1}{2}
\end{align*}
as $I(|P) \geq 0$.

As of the last term, we have
\begin{align*}
B_4 & \leq \Bigg( c_T \big(1+I(\mu|P)\big) \int_0^T \big(m^{\alpha}_{\mu}(\tk) - m^{\alpha}_{\nu}(\tk)\big)^2 dt \Bigg)^{\frac{1}{2}}\\
& \leq \frac{\Ja}{\la} \Big( c_T \big(1+I(\mu|P)\big)\Big)^{\frac{1}{2}} \max_{\gamma=1\cdots M} \bigg(\int_{0}^T  \Big( \int |\Sag(x^{\gamma}_{\tk-\taag}) - \Sag(y^{\gamma}_{\tk-\taag})| d\xi(x,y) \Big)^2 dt \bigg)^{\frac{1}{2}}\\
& \leq c'_T \big(1+I(\mu|P)\big) \max_{\gamma=1\cdots M} \bigg(\int_{0}^T  \int |\Sag(x^{\gamma}_{\tk-\taag}) - \Sag(y^{\gamma}_{\tk-\taag})|^2 d\xi(x,y) dt \bigg)^{\frac{1}{2}}
\end{align*}

We have proved that there exist a constant $c_T$ such that
\begin{multline}
|\Gamma_{2,\nu}^{\alpha,k}(\mu) - \Gamma_2^{\alpha,k}(\mu)| \leq c_T \big(1+I(\mu|P)\big) \max_{\gamma=1\cdots M} \bigg(\int_{0}^T  \int |\Sag(x^{\gamma}_{\tk-\taag}) - \Sag(y^{\gamma}_{\tk-\taag})|^2 d\xi(x,y) dt \bigg)^{\frac{1}{2}}\\ \label{gamma2}
\end{multline}
And, therefore
\begin{equation*}
|\Gamma_{2,\nu}^{\alpha,k}(\mu) - \Gamma_2^{\alpha,k}(\mu)| \leq c_T K_S \sqrt{T} \big(1+I(\mu|P)\big) d_T(\mu,\nu).
\end{equation*}
so that using the triangle inequality
\begin{equation*}
|\Gamma_{2,\nu}^{k}(\mu) - \Gamma_2^{k}(\mu)| \leq C_T\big(1+I(\mu|P)\big) d_T(\mu,\nu).
\end{equation*}.
{\bf Proof of Lemma\ref{lemma1}.(v)}:

For all $\alpha \in \{1\cdots M\}$, let
\begin{align*}
dQ_{\nu}^{\alpha, k}(x) & = \exp{\Gamma_{\nu}^{\alpha,k}(\delta_x)}d\pa(x)  \\
& = \int \exp{ \Bigg( \int_0^T (G^{\alpha}_{\tk} + m^{\alpha}_{\nu}(\tk)) \, dW^{\alpha}_t(x) - \frac{1}{2} \int_0^T (G^{\alpha}_{\tk} + m^{\alpha}_{\nu}(\tk))^2 dt \Bigg) } \, d\gamma_{\nu} \, d\pa(x)
\end{align*}
The equality between the two expression of $Q_{\nu}^{\alpha,k}$ is easily obtained by gaussian calculus (see the proof of proposition~\ref{pro:DécompositionGamma}). We deduce by the martingale property of this density that it is a probability measure on $\C([-\tau,T],\R)$.

Remark that
\begin{equation*}
Q_{\nu}^k = \otimes_{\alpha=1}^M Q_{\nu}^{\alpha, k}
\end{equation*}
\begin{equation*}
\der{Q_{\nu}^k}{P}(x) = \prod_{\alpha=1}^M \der{Q_{\nu}^{\alpha, k}}{\pa}(x^{\alpha})
\end{equation*}

It follows that $Q_{\nu}^k \in \M_1^+(\C)$, and
\begin{align*}
dQ_{\nu}^k(x) & = \exp{\Gamma_{\nu}^{k}(\delta_x)}dP(x)  \\
& = \int \exp{ \Bigg( \int_0^T \big( \Gv_{\tk} + \m_{\nu}(\tk) \big)' \cdot \, d\W_t(x) - \frac{1}{2} \int_0^T \Big\| \Gv_{\tk} + \m_{\nu}(\tk)\Big\|^2 dt \Bigg) } \, d\gamma_{\nu} \, dP(x)
\end{align*}

Lets now prove that $H^k_{\nu} =I(.|Q^k_{\nu})$.
We will first show that $I(Q^k_{\nu}|P)$ is finite.
In fact,
\begin{multline}
\der{Q_{\nu}^{\alpha, k}}{\pa}(x) = \bigg(\int \exp\bigg\{-\frac{1}{2} \int_0^T {G^{\alpha}_{\tk}}^2 + {m^{\alpha}_{\nu}}^2(\tk) dt\bigg\} d\gamma_{\nu}\bigg) \exp\Big\{\int_0^T m^{\alpha}_{\nu}(\tk) dW^{\alpha}_t(x)\Big\} \\ \exp\bigg\{\frac{1}{2} \int \bigg(\int_0^T G^{\alpha}_{\tk}\Big( dW^{\alpha}_t(x) - m^{\alpha}_{\nu}(\tk) dt\Big)\bigg)^2 d\gamma_{\widetilde{K}_{\nu}^{T,k}}\bigg\}
\end{multline}

Which becomes, after some gaussian computations (see \cite[Lemma 5.15]{ben-arous-guionnet:95}):
\begin{align*}
\der{Q_{\nu}^{\alpha, k}}{\pa} = \exp\bigg\{\int_0^T H^{\alpha}_{\tk}(Q_{\nu}^{\alpha, k})dW^{\alpha}_t-\frac{1}{2}\int_0^T {H^{\alpha}_{\tk}}^2(Q_{\nu}^{\alpha, k})dt\bigg\}
\end{align*}

where
\begin{align*}
H^{\alpha}_{\tk}(Q^{\alpha, k}_{\nu}) & = \bigg(\int G^{\alpha}_{\tk} \int_0^t G^{\alpha}_{\sk} \big(dW^{\alpha}_s - m^{\alpha}_{\nu}(\sk)ds\big) \; d\gamma_{\widetilde{K}_{\nu}^{t,k}}\bigg) + m^{\alpha}_{\nu}(\tk)\\
& = \int_0^t \big(\widetilde{K}_{\nu}^{t,k}(\tk,\sk)\big)_{\alpha\alpha} \big(dW^{\alpha}_s - m^{\alpha}_{\nu}(\sk)ds\big) + m^{\alpha}_{\nu}(\tk) \\
& = \sum_{l=0,..,k ; t_{l+1}\leq t} \Big(W^{\alpha}_{t_{l+1}}-W^{\alpha}_{t_l} - m^{\alpha}_{\nu}(t_l)(t_{l+1}-t_l)\Big) \big(\widetilde{K}_{\nu}^{t,k}(\tk, t_k)\big)_{\alpha\alpha}  + m^{\alpha}_{\nu}(\tk)
\end{align*}

Hence, according to Girsanov Theorem, there exists a $Q_{\nu}^{\alpha, k}$-brownian motion $B^{\alpha}$ such that:
\begin{equation*}
W^{\alpha}_t=B^{\alpha}_t+\int_0^t H^{\alpha}_{\sk}(Q_{\nu}^{\alpha, k})ds
\end{equation*}
As $W$ is an affine function of $B^{\alpha}$, it is, under $Q_{\nu}^{\alpha, k}$, a gaussian variable with finite moments.
In particular, $I(Q_{\nu}^{\alpha, k}|\pa)$ is finite.
But
\begin{equation*}
I(Q_{\nu}^k|P)= \int_{\C} \log\big( \der{Q^k_{\nu}}{P}(x) \big) dQ^k_{\nu}(x)
= \sum_{\alpha=1}^M  \int_{\C([-\tau,T],\R)} \log\big(  \der{Q^{\alpha, k}_{\nu}}{\pa}(x^{\alpha}) \big) dQ^{\alpha, k}_{\nu}(x^{\alpha}) = \sum_{\alpha=1}^M I(Q_{\nu}^{\alpha, k}|\pa)
\end{equation*}
so that $I(Q_{\nu}^k|P)$ is finite.

Now, let
\begin{equation*}
\bigg(\tau^{\alpha}_{m}(x)=\inf\bigg\{t \geq 0 ; \bigg| \int_0^t m^{\alpha}_{\nu}(\sk) dW^{\alpha}_s(x)\bigg| \geq m\bigg\}\bigg)_{m \in \N}
\end{equation*}
be a sequence of stopping times for the brownian filtration $\sigma( W^{\alpha}_s, 0\leq s \leq T)$.
As $W^{\alpha}$ is a $\pa$-MB, we have $\lim_{m\to\infty} \tau^{\alpha}_{m} \to \infty$ almost surely under $\pa$.
$\big( \tau_m= \min_{\alpha=1 \cdots M} \tau^{\alpha}_{m}\big)_{m\in \N}$ defines a sequence stopping times for $\sigma( W^{\alpha}_s, \alpha=1\cdots M, 0\leq s \leq T)$ which tends to infinity along with $m$ $P$ a.s.
We define:
\begin{equation*}
Q_{\nu,t}^k=\int \exp\bigg\{ \Bigg( \int_0^t (\Gv_{\sk} + \m_{\nu}(\sk))' \cdot \, d\W_s - \frac{1}{2} \int_0^t \Big\|\Gv_{\sk} + \m_{\nu}(\sk)\Big\|^2 ds \Bigg) \bigg\} \, d\gamma_{\nu} \, P
\end{equation*}

\begin{multline*}
\Gamma^{k}_{\nu,t}(\mu) = \int_{\C} \log\bigg( \int \exp\bigg\{ \int_0^t \big(\Gv_{\sk}(\omega)+\m_{\nu}(\sk)\big)' \cdot d\W_s(x) \\
- \frac{1}{2} \int_0^t \Big\|\Gv_{\sk}(\omega)+\m_{\nu}(\sk)\Big\|^2ds\bigg\}  d\gamma_{K^t_{\nu}}(\omega) \bigg) d\mu(x)
\end{multline*}
 as well as $Q_{\nu,t}^{\alpha,k}$ and $\Gamma^{\alpha, k}_{\nu,t}$
where
\[
K_{\mu}^t(s,u)=\Big( \mathbf{1}_{\alpha=\beta} \displaystyle{\frac{1}{\la^2} \sum_{\gamma=1}^M \sag^2 \int_{\C} \Sag(x^{\gamma}_{s-\taag})\Sag(x^{\gamma}_{u-\taag}) d\mu(x)} \Big)_{\alpha, \beta \in \{1 \cdots M\}}.
\]
is define on $[0,t]^2$.

These functions are clearly continuous in time on $[0,T]$.
The result $I(Q_{\nu,\tau_{m}\wedge T}^k|P) < \infty$ obviously remains.

By Jensen inequality, we have:
\begin{align*}
\der{Q_{\nu,\tau^{\alpha}_{m}\wedge T}^{\alpha, k}}{\pa}(x) & \geq \exp\Big\{-\frac{\ka (\tau^{\alpha}_{m}(x)\wedge T)}{2\la^2}\Big\} \exp\Big\{-\frac{\Ja^2 (\tau^{\alpha}_{m}(x)\wedge T)}{2 \la^2}\Big\} \exp\bigg\{\int_0^{\tau^{\alpha}_{m}(x)\wedge T} m^{\alpha}_{\nu}(\tk) dW^{\alpha}_t(x)\bigg\}\\
& \geq  \exp\Big\{-\frac{(\ka + \Ja^2) T}{2\la^2}-m\Big\}
\end{align*}
so that,
\begin{equation*}
\der{Q_{\nu,\tau_{m}\wedge T}^k}{P}(x)  \geq \exp\Big\{-\sum_{\alpha=1}^M \frac{(\ka + \Ja^2) T}{2\la^2}-Mm\Big\}
\end{equation*}

We then apply the same proof as in \cite[Appendix B]{ben-arous-guionnet:95} to find:
\begin{equation*}
\forall \mu \in \M_1^+(\C),  \; \; H_{\nu,\tau_{m}\wedge T}^k = I(\mu|Q_{ \nu,\tau_{m}\wedge T}^k)
\end{equation*}
Letting $m$ to infinity, we conclude using the continuity of $H_{\nu,t}^k$ and  $I(.|Q_{\nu,t}^k)$ on $[0,T]$.

\noindent{\bf Proof of Lemma\ref{lemma1}.(vi):}
In order to demonstrate that $H^k$ is a good rate function, we need to show that it is lower semi-continuous and that it has compact level sets, i.e. $\{H^k\leq L\}$ is a compact set for any $L>0$. This is a direct consequence of points (i)-(iv) proved above.
\end{proof}

\subsection{Proof of Lemma~\ref{lemma3.1}}\label{append:ProofLemma3.1}
This appendix is concerned with the proof of lemma~\ref{lemma3.1} ensuring an exponential bound that will be used to show a tightness result on the sequence of empirical laws.

\begin{proof}
Let
\begin{equation*}
B^n = \int_{\hat{\mu}_n \in B(\nu,\delta)} \exp\Big\{ a n \big(\Gamma(\hat{\mu}_n)- \Gamma_{\nu}(\hat{\mu}_n)\big)\Big\} dQ_{\nu}^n
\end{equation*}
Writing the definitions of $\Gamma$ and $\Gamma_{\nu}$, we find:
\begin{equation*}
B^n = \int_{\hat{\mu}_n \in B(\nu,\delta)} \prod_{i=1}^n \Bigg( \frac{{\E}_{\hat{\mu}_n} \bigg[\exp\bigg\{ \int_0^T \big( \Gv_t + \m_{\hat{\mu}_n}(t) \big)' \cdot d\W^i_t -\frac{1}{2} \int_0^T \Big\| \Gv_t + \m_{\hat{\mu}_n}(t) \Big\|^2 dt \bigg\} \bigg]}{{\E}_{\nu} \bigg[\exp\bigg\{ \int_0^T \big( \Gv_t + \m_{\nu}(t) \big)' \cdot d\W^i_t - \frac{1}{2} \int_0^T \Big\| \Gv_t + \m_{\nu}(t) \Big\|^2 dt \bigg\} \bigg]} \Bigg)^a d(Q_{\nu})^{\otimes n}
\end{equation*}

Let $\xi$ be a probability measure on $\C \times \C$ with marginals $\hat{\mu}_n$ and $\nu$. We then have:

\begin{equation*}
B^n = \int_{\hat{\mu}_n \in B(\nu,\delta)} \prod_{i=1}^n \Bigg( \frac{{\E}_{\xi} \bigg[\exp \bigg\{ \int_0^T \big( \Gv_t + \m_{\hat{\mu}_n}(t) \big)' \cdot d\W^i_t - \frac{1}{2} \int_0^T \Big\| \Gv_t + \m_{\hat{\mu}_n}(t) \Big\|^2 dt \bigg\} \bigg]}{{\E}_{\xi} \bigg[\exp\bigg\{ \int_0^T \big( \Gv_t' + \m_{\nu}(t) \big)d\W^i_t - \frac{1}{2} \int_0^T \Big\| \Gv_t' + \m_{\nu}(t) \Big\|^2 dt \bigg\} \bigg]} \Bigg)^a d(Q_{\nu})^{\otimes n}
\end{equation*}
where $(\Gv,\Gv')$ is a 2M-dimensional gaussian centered process  with covariance $K_{\xi}$ (see (\ref{kxi})).
\begin{equation}
K_{\xi}(s,t) =
\left(
\begin{array}{ccc}
K_{\mu}(s,t) & \big(\mathbf{1}_{\{\alpha=\gamma\}} K^{\alpha}_{\xi}(s,t)\big)_{\alpha, \gamma = 1 \cdots M} \\
\big(\mathbf{1}_{\{\alpha=\gamma\}} K^{\alpha}_{\xi}(s,t)\big)_{\alpha, \gamma = 1 \cdots M} & K_{\nu}(s,t) \\
\end{array}
\right) \label{kxi2},
\end{equation}

Let
\begin{equation*}
Y_i= \int_0^T \big( \Gv_t + \m_{\hat{\mu}_n}(t) \big)' \cdot d\W^i_t -\frac{1}{2} \int_0^T \Big\| \Gv_t + \m_{\hat{\mu}_n}(t) \Big\|^2 dt,
\end{equation*}
\begin{equation*}
Y_i'= \int_0^T \big( \Gv'_t + \m_{\nu}(t) \big)' \cdot d\W^i_t -\frac{1}{2} \int_0^T \Big\| \Gv_t' + \m_{\nu}(t) \Big\|^2 dt.
\end{equation*}

Then
\begin{align*}
B^n & = \int_{\hat{\mu}_n \in B(\nu,\delta)} \prod_{i=1}^n \bigg( \frac{{\E}_{\xi} \Big[\exp \Big\{ Y_i \Big\} \Big]}{{\E}_{\xi} \Big[\exp\Big\{ Y_i' \Big\} \Big]} \bigg)^a d(Q_{\nu})^{\otimes n} \\
& = \int_{\hat{\mu}_n \in B(\nu,\delta)} \prod_{i=1}^n \Bigg( {\E}_{\xi} \Bigg[ \frac{\exp{Y_i'}}{{\E}_{\xi} \Big[\exp{Y_i'} \Big]} \exp{ \Big( Y_i - Y_i' \Big) } \Bigg] \Bigg)^a d(Q_{\nu})^{\otimes n} \\
& \leq \int_{\hat{\mu}_n \in B(\nu,\delta)} \prod_{i=1}^n {\E}_{\xi} \Bigg[ \frac{\exp{Y_i'}}{{\E}_{\xi} \Big[\exp{ Y_i'} \Big]} \exp{ a \Big( Y_i - Y_i' \Big) } \Bigg] d(Q_{\nu})^{\otimes n} \\
\end{align*}
by Jensen inequality.\\
Then, using Holder inequality twice with conjugate exponents $(p,q)$ and $(\sigma, \eta)$, one finds: \\

\begin{multline}
B^n  \leq  \Bigg\{ \overbrace{\int \prod_{i=1}^n  \frac{{\E}_{\xi} \Big[\exp { pY_i'}\Big]}{\Big( {\E}_{\xi} \Big[\exp{Y_i'} \Big] \Big)^p } d(Q_{\nu})^{\otimes n}}^{B^n_1} \Bigg\}^{\frac{1}{p}} \Bigg\{ \overbrace{\int \exp\Big\{n \sigma \Gamma_{\nu}(\hat{\mu}_n) \Big\} dP^{\otimes n}}^{B^n_2} \Bigg\}^{\frac{1}{q\sigma}} \\
\times \Bigg\{ \underbrace{\int_{\hat{\mu}_n \in B(\nu,\delta)} \prod_{i=1}^n {\E}_{\xi} \Big[ \exp{ a \eta q \Big( Y_i -Y_i' \Big) } \Big] dP^{\otimes n}}_{B^n_3} \Bigg\}^{\frac{1}{q\eta}} \label{ineqlemma}
\end{multline}

We first bound the first term of the right hand side. Let $\gamma_{p,\widetilde{K}^T_{\mu}}$ be a probability measure on $\Omega$ such that $d\gamma_{p,\widetilde{K}^T_{\mu}} = \frac{\prod_{\gamma=1}^M \exp{-\frac{p}{2}\int_0^T {G^{\gamma}_t}^2 dt} }{\int \prod_{\gamma=1}^M \exp{-\frac{p}{2}\int_0^T {G^{\gamma}_t}^2 dt} d\gamma_{\mu} } d\gamma_{\mu}$. \\
According to appendix A of \cite{ben-arous-guionnet:95} (where $p=\beta^2$), we have, for any $p \geq 0$, that $\Gv$ is a M-dimensional centered gaussian process under $\gamma_{p,\widetilde{K}^T_{\mu}}$.\\
Consequently, using the independence of $(\Gv,\Gv')$'s components:

\begin{align*}
{\E}_{\xi} \Big[\exp { pY_i'}\Big] & = \prod_{\alpha=1}^M \exp\bigg\{ p \Big( \int_0^T m_{\nu}^{\alpha}(t) dW^{i_{\alpha}}_t - \frac{1}{2} \int_0^T {m^{\alpha}_{\nu}}^2(t) dt \Big) \bigg\} \; {\E}_{\xi}\Big[ \exp\bigg\{-\frac{p}{2} \int_0^T {{G^{\alpha}_t}'}^2 dt\bigg\} \Big] \\
&\times \int \exp\bigg\{ p\Big(\int_0^T {G^{\alpha}_t}' \big(dW^{i_{\alpha}}_t - m_{\nu}^{\alpha}(t) dt\big) \Big) \bigg\}d\gamma_{p,\widetilde{K}^T_{\nu}}\\
& = \prod_{\alpha=1}^M \exp\bigg\{ p \Big( \int_0^T m_{\nu}^{\alpha}(t) dW^{i_{\alpha}}_t - \frac{1}{2} \int_0^T {m^{\alpha}_{\nu}}^2(t) dt \Big) \bigg\} \; {\E}_{\xi}\Big[ \exp\bigg\{-\frac{p}{2} \int_0^T {{G^{\alpha}_t}'}^2 dt\bigg\} \Big] \\
& \exp\bigg\{ \Big( \frac{p^2}{2} \int  \Big(\int_0^T G^{\alpha}_t \big(dW^{i_{\alpha}}_t - m_{\nu}^{\alpha}(t) dt\big) \Big)^2 \frac{\exp{\left(-\frac{p}{2}\int_0^T {G^{\alpha}_t}^2 dt\right)} }{\int \exp{\left(-\frac{p}{2}\int_0^T {G^{\alpha}_t}^2 dt\right)} d\gamma_{\nu}} d\gamma_{\nu} \Big)\bigg\}\\
\end{align*}

Hence,
\begin{align*}
\frac{{\E}_{\xi}\Big[ \exp{pY_i'} \Big]}{{\E}_{\xi}\Big[ \exp{Y_i'} \Big]^p} & = \prod_{\alpha=1}^M \underbrace{\frac{{\E}_{\xi}\Big[ \exp{-\frac{p}{2} \int_0^T {{G^{\alpha}_t}'}^2 dt} \Big]}{{\E}_{\xi}\Big[ \exp{-\frac{1}{2} \int_0^T {{G^{\alpha}_t}'}^2 dt} \Big]^p}}_{f^{\alpha}(p)} \exp\bigg\{ \frac{p}{2} \int  \Big(\int_0^T G^{\alpha}_t \big(dW^{i_{\alpha}}_t - m_{\nu}^{\alpha}(t) dt\big) \Big)^2 \\
& \times \underbrace{\big(p\frac{\exp{-\frac{p}{2}\int_0^T {G^{\alpha}_t}^2dt} }{\int \exp{-\frac{p}{2}\int_0^T {G^{\alpha}_t}^2dt} d\gamma_{\nu}} - \Lambda_T(G^{\alpha}) \big)}_{g^{\alpha}(p)} d\gamma_{\nu}  \bigg\}
\end{align*}

Remark that$ f^{\alpha}(p)$ is bounded for $p>0$. In fact, on one hand it is clear that
\[\forall p \in [0,+\infty[, {\E}_{\xi}\Big[ \exp{-\frac{p}{2} \int_0^T {{G^{\alpha}_t}'}^2 dt} \Big]\leq 1,\]
furthermore Jensen inequality and Fubini theorem give us:
\begin{equation*}
{\E}_{\nu}\Big[\exp\Big\{-\frac{p}{2}\int_0^T {G^{\alpha}_t}^2 dt\Big\} \Big] \geq \exp\bigg\{ - \frac{p}{2} \int_0^T {\E}_{\nu}\Big[ {G^{\alpha}_t}^2 \Big] dt \bigg\} \geq \exp\Big\{ -\frac{p T \ka}{2 \la^2} \Big\}.\\
\end{equation*}
Therefore, bounded convergence monotone gives $f^{\alpha}(p) \to 1$  and, similarly, $g^{\alpha}(p) \to 0$ as $p \searrow 1$.
Moreover, as $\Big(\int_0^T {G^{\alpha}_t}' \big(dW^{i_{\alpha}}_t - m_{\nu}^{\alpha}(t) dt\big) \Big)^2$ has finite moments under $\gamma_{\nu}$, we can find a finite constant $C_1(p)$, $C_1(p) \searrow 0$ as $p \searrow 1$, such that:
\begin{equation}
B^n_1 = \Bigg( \int \frac{{\E}_{\xi} \Big[\exp { pY_i'}\Big]}{\Big( {\E}_{\xi} \Big[\exp{ Y_i'} \Big] \Big)^p } dQ_{\nu}\Bigg)^n \leq e^{ C_1(p) n} \label{t1}
\end{equation}

Moreover:
\begin{align*}
B^n_2 & = \Bigg( \int {\E}_{\nu}\bigg[ \exp\bigg\{ \int_0^T \big(\Gv_t + \m_{\nu}(t)\big)' \cdot d\W_t(x) - \frac{1}{2} \int_0^T \Big\|\Gv_t + \m_{\nu}(t)\Big\|^2 dt\bigg\} \bigg]^{\sigma} dP(x) \Bigg)^n \\
& \leq \Bigg( {\E}_{\nu}\bigg[ \int \exp\bigg\{ \sigma \int_0^T \big(\Gv_t + \m_{\nu}(t)\big)' \cdot d\W_t(x) \bigg\} dP(x) \; \exp\bigg\{ - \frac{\sigma}{2} \int_0^T \Big\|\Gv_t + \m_{\nu}(t)\Big\|^2 dt\bigg\} \bigg] \Bigg)^n \\
& = \Bigg( {\E}_{\nu}\bigg[ \exp\bigg\{ \frac{\sigma^2-\sigma}{2} \int_0^T \Big\|\Gv_t + \m_{\nu}(t)\Big\|^2 dt\bigg\} \bigg] \Bigg)^n
\end{align*}

So that if we take $\sigma$ close enough to 1, we can find a finite constant $C_2(\sigma), \lim_{\sigma \to 1} C_2(\sigma) =0$, such that (see inequality (\ref{ineqA})):

\begin{equation}
B^n_2 \leq \Bigg( {\E}_{\nu}\bigg[ \exp\bigg\{ \frac{\sigma^2-\sigma}{2} \int_0^T \Big\|\Gv_t + \m_{\nu}(t)\Big\|^2 dt\bigg\} \bigg] \Bigg)^n < e^{ C_2(\sigma) n } \label{t2}
\end{equation}

We now will bound the last term of the right hand side of (\ref{ineqlemma}). By Cauchy-Schwarz inequality, if $\kappa=qa\eta$:
\begin{align*}
B^n_3 \leq \Bigg\{ \int_{\hat{\mu}_n \in B(\nu,\delta)} \prod_{i=1}^n &{\E}_{\xi} \bigg[ \exp\Big\{ 2\kappa \int_0^T \big(\Gv_t-\Gv_t'+ (\m_{\hat{\mu}_n}-\m_{\nu})(t)\big)' \cdot d\W^i_t\\
- 2\kappa^2 \! \int_0^T \! \Big\|\Gv_t \! - \! \Gv_t' \! +& \! (\m_{\hat{\mu}_n} \!-\m_{\nu})(t)\Big\|^2 dt \Big\}\bigg] dP^{\otimes n} \Bigg\}^{\frac{1}{2}} \! \Bigg\{ \! \int_{\hat{\mu}_n \in B(\nu,\delta)} \! {\E}_{\xi} \bigg[ \exp\Big\{ 2\kappa^2 \! \int_0^T \! \Big\|\Gv_t \!- \! \Gv_t'\! + \!(\m_{\hat{\mu}_n} \! - \!\m_{\nu})(t)\Big\|^2 \!dt\\
& - \kappa \int_0^T \Big\|\Gv_t+\m_{\hat{\mu}_n}(t)\Big\|^2 - \Big\| \Gv_t' + \m_{\nu}(t)\Big\|^2 dt \Big\} \bigg]^n dP^{\otimes n} \Bigg\}^{\frac{1}{2}}
\end{align*}
The first term is the square root of a martingale's expectation, thus equal to one. For the second term, we remark that:
\begin{multline*}
- \int_0^T \big(G^{\alpha}_t+m^{\alpha}_{\hat{\mu}_N}(t)\big)^2 - \big({G^{\alpha}_t}' + m_{\nu}^{\alpha}(t)\big)^2 dt \leq \frac{\delta^{\frac{1}{2}}}{2} \bigg( \frac{1}{\delta} \int_0^T \big(G^{\alpha}_t- {G^{\alpha}_t}' + (m^{\alpha}_{\hat{\mu}_N}(t)  - m_{\nu}^{\alpha})(t)\big)^2 dt \\
 + \int_0^T \big(G^{\alpha}_t+ {G^{\alpha}_t}' + (m^{\alpha}_{\hat{\mu}_N}(t)+m^{\alpha}_{\nu})(t)\big)^2 dt \bigg)
\end{multline*}
 so that, by Cauchy-Schwarz inequality:
\begin{align*}
B^n_3 & \leq \Bigg\{ \int {\E}_{\xi} \bigg[ \exp\Big\{ \big(4\kappa^2+\kappa\delta^{-\frac{1}{2}}\big) \int_0^T \Big\| \Gv_t-\Gv_t'+ (\m_{\hat{\mu}_n}-\m_{\nu})(\tk)\Big\|^2 dt \Big\} \bigg]^n dP^{\otimes n}\Bigg\}^{\frac{1}{4}} \\
& \times \Bigg\{ \int {\E}_{\xi} \bigg[ \exp\Big\{ \kappa\delta^{\frac{1}{2}}\int_0^T \Big\| \Gv_t+\Gv_t'+ (\m_{\hat{\mu}_n}+\m_{\nu})(t)\Big\|^2 dt \Big\} \bigg]^n dP^{\otimes n} \Bigg\}^{\frac{1}{4}} \\
& \leq \exp{\sum_{\alpha=1}^n \Big\{\frac{1}{2} (4\kappa^2+\kappa\delta^{-\frac{1}{2}})\frac{\Ja^2 K_S^2}{\la^2}T\delta^2 + (2\kappa\delta^{\frac{1}{2}}T\frac{\Ja^2}{\la^2}) \Big\}}  \\
& \times \Bigg\{ \int {\E}_{\xi} \bigg[ \exp\Big\{ 2\big(4\kappa^2+\kappa\delta^{-\frac{1}{2}}\big) \int_0^T \Big\|\Gv_t-\Gv_t'\Big\|^2 dt \Big\} \bigg]^n dP^{\otimes n}\Bigg\}^{\frac{1}{4}} \\
& \times \Bigg\{ \int {\E}_{\xi} \bigg[ \exp\Big\{ 2\kappa\delta^{\frac{1}{2}}\int_0^T \Big\|\Gv_t+\Gv_t'\Big\|^2 dt \Big\} \bigg]^n dP^{\otimes n} \Bigg\}^{\frac{1}{4}}
\end{align*}
As ${\E}_{\xi}\Big[ \int_0^T \Big\|\Gv_t-\Gv_t'\Big\|^2 dt \Big] = \sum_{\alpha, \gamma=1}^M \frac{\sag^2 K_S^2}{\la^2} \int_0^T \int \big| \Sag(x^{\gamma}_{t-\taag}) - \Sag(y^{\gamma}_{t-\taag})\big|^2 d\xi(x,y) dt \leq \sum_{\alpha, \gamma=1}^M \frac{\sag^2 T}{\la^2} \delta^2$ (in fact, $d_T(\hat{\mu}_n,\nu) \leq \delta)$, we can use appendix A, lemma 3.2 of \cite{ben-arous-guionnet:95} to bound the two last term of the previous inequality: there exists two finite constants $C^{\kappa}_1(\delta)$ and $C^{\kappa}_2(\delta)$ such that,
\begin{equation*}
{\E}_{\xi} \bigg[ \exp\Big\{ 2\big(4\kappa^2+\kappa\delta^{-\frac{1}{2}}\big) \int_0^T \Big\|\Gv_t-\Gv_t'\Big\|^2 dt \Big\} \bigg] \leq \exp\Big\{ 4\big(4\kappa^2+\kappa\delta^{-\frac{1}{2}}\big) C^{\kappa}_1(\delta) C_T \delta^2 \Big\}
\end{equation*}
\begin{equation*}
{\E}_{\xi} \bigg[ \exp{ \Big\{ 2\kappa\delta^{\frac{1}{2}}\int_0^T \Big\|\Gv_t+\Gv_t'\Big\|^2 dt \Big\}} \bigg] \leq \exp{\Big\{ 4\kappa\delta^{\frac{1}{2}} C^{\kappa}_2(\delta) \sum_{\alpha=1}^M \frac{4\ka T}{\la^2} \Big\} }
\end{equation*}

Hence, we can find $C_{\kappa}(\delta), \lim_{\delta \to 0} C_{\kappa}(\delta) = 0$, such that
\begin{equation}
B^n_3 \leq e^{C_{\kappa}(\delta) n} \label{t3}
\end{equation}

We conclude by using (\ref{t1}), (\ref{t2}) and (\ref{t3}) in (\ref{ineqlemma}).
\end{proof}

\end{document}